\documentclass[11pt,a4paper]{article}
\usepackage{a4wide}
\usepackage{amsmath,amssymb,amsthm}
\usepackage[mathic]{mathtools}
\usepackage[T1]{fontenc}
\usepackage{xcolor}
\usepackage{enumerate}
\usepackage{graphicx}
\usepackage{subfig}
\usepackage{adjustbox}
\usepackage{paralist}
\usepackage[font=small,skip=0pt]{caption}

\usepackage[linesnumbered,ruled,vlined,boxed]{algorithm2e}
\newlength{\commentWidth}
\let\oldnl\nl
\newcommand{\nonl}{\renewcommand{\nl}{\let\nl\oldnl}}
\DontPrintSemicolon
\SetKwInOut{Input}{Input}\SetKwInOut{Output}{Output}

\usepackage[
	style=alphabetic,
	backref=true,
	doi=false,
	url=false,
	maxcitenames=3,
	mincitenames=3,
	maxbibnames=10,
	minbibnames=10,
	backend=bibtex8,
	sortlocale=en_US
]{biblatex}

\usepackage[ocgcolorlinks]{hyperref} 
\colorlet{DarkRed}{red!50!black}
\colorlet{DarkGreen}{green!50!black}
\colorlet{DarkBlue}{blue!50!black}

\hypersetup{
	linkcolor = DarkRed,
	citecolor = DarkGreen,
	urlcolor = DarkBlue,
	bookmarksnumbered = true,
	linktocpage = true
}

\newtheorem{theorem}{Theorem}
\newtheorem{lemma}{Lemma}
\newtheorem{definition}{Definition}
\newtheorem{corollary}{Corollary}

\newcommand{\bx}{\mathbf{x}}

\newcommand{\bz}{\mathbf{z}}
\newcommand{\bm}{\mathbf{m}}
\newcommand{\bp}{\mathbf{p}}
\newcommand{\bu}{\mathbf{u}}

\DeclareMathOperator{\UB}{UB}
\DeclareMathOperator{\LB}{LB}

\newcommand{\x}{x}

\newcommand{\f}{\varphi}

\DeclareMathOperator{\LOG}{\ensuremath{\operatorname{log\mathllap{\raisebox{1.35ex}{\rule{0.95em}{0.08ex}}\hspace{0.025em}}}}}

\DeclareMathOperator{\res}{Res}
\newcommand{\RES}[2]{\res(#1, #2)}

\providecommand{\f}{\phi}
\providecommand{\F}{\Phi}
\providecommand{\g}{\gamma}

\newcommand{\SSS}{\mathcal{S}}
\newcommand{\FFF}{\mathcal{F}}
\newcommand{\GGG}{\mathcal{G}}
\newcommand{\TTT}{\mathcal{T}}

\newcommand{\NN}{\ensuremath{\mathbb{N}}}

\newcommand{\ZZ}{\ensuremath{\mathbb{Z}}}
\newcommand{\QQ}{\ensuremath{\mathbb{Q}}}
\newcommand{\RR}{\ensuremath{\mathbb{R}}}
\newcommand{\CC}{\ensuremath{\mathbb{C}}}

\newcommand{\PP}{\ensuremath{\mathbb{P}}}

\newcommand{\ZZnz}{\ensuremath{\ZZ_{\neq 0}}}

\newcommand{\ZZnneg}{\ensuremath{\ZZ_{\ge 0}}}

\DeclareMathOperator{\Mea}{Mea}

\newcommand{\cert}{\#\textbf{\textsc{PolySol}}}

\title{Counting Solutions of a Polynomial System\\
 Locally and Exactly}
\date{}
\author{Ruben Becker \and Michael Sagraloff}

\begin{document}
\maketitle
\begin{abstract}
	We propose a symbolic-numeric algorithm to count the number of solutions of a polynomial system within a local region. More specifically, given a zero-dimensional system $f_1=\cdots=f_n=0$, with $f_i\in\CC[x_1,\ldots,x_n]$, and a polydisc $\mathbf{\Delta}\subset\CC^n$, our method aims to certify the existence of $k$ solutions (counted with multiplicity) within the polydisc. 
	In case of success, it yields the correct result under guarantee. Otherwise, no information is given. However, we show that our algorithm always succeeds if $\mathbf{\Delta}$ is sufficiently small and well-isolating for a $k$-fold solution $\bz$ of the system. 

	Our analysis of the algorithm further yields a bound on the size of the polydisc for which our algorithm succeeds under guarantee. This bound depends on local parameters such as the size and multiplicity of $\bz$ as well as the distances between $\bz$ and all other solutions. Efficiency of our method stems from the fact that we reduce the problem of counting the roots in $\mathbf{\Delta}$ of the original system to the problem of solving a truncated system of degree $k$. In particular, if the multiplicity $k$ of $\bz$ is small compared to the total degrees of the polynomials $f_i$, our method considerably improves upon known complete and certified methods.
	
	For the special case of a bivariate system, we report on an implementation of our algorithm, and show experimentally that our algorithm leads to a significant improvement, when integrated as inclusion predicate into an elimination method.
\end{abstract}


\section{Introduction}
In this paper, we propose a randomized but certified (i.e.~Las-Vegas type) algorithm, denoted $\cert$, to count the number of solutions of a zero-dimensional polynomial system $\FFF$ within a given polydisc $\mathbf{\Delta}\subset\CC^n$. Let 
\begin{align}\label{formula:original_system}
	\FFF: f_1(\bx)=\ldots = f_n(\bx) = 0, \quad\text{with }f_i\in\CC[\bx]=\CC[x_1,\ldots,x_n]\text{ for all }i=1,\ldots,n,
\end{align}
be a zero-dimensional\footnote{There are only finitely many solution in complex projective $n$-space.} polynomial system. We further assume that each of the coefficients $c_{i,\alpha}$ of the polynomials 
\[
f_i=\sum_{\alpha=(\alpha_1,\ldots,\alpha_n)}c_{i,\alpha}\cdot \bx^{\alpha}=\sum_{\alpha=(\alpha_1,\ldots,\alpha_n)}c_{i,\alpha}\cdot x_1^{\alpha_1}\cdots x_n^{\alpha_n}
\] 
can be approximated to any desired precision. That is, for any given non-negative integer (precision) $\rho$, we can ask for a dyadic approximation $\tilde{c}_{i,\alpha}\in 2^{-\rho}\cdot (\ZZ+\mathbf{i}\cdot\ZZ)$ of $c_{i,\alpha}$ with $|\tilde{c}_{i,\alpha}-c_{i,\alpha}|<2^{-\rho}$ for the cost of reading the approximations.

Given a polydisc $\mathbf{\Delta}=\mathbf{\Delta}_r(\bm)=\{\bz\in\CC^n:\|\bz-\bm\|_\infty<r\}$ of radius $r$ centered at $\bm$, we aim to compute the number of solutions of $\FFF=0$ in $\mathbf{\Delta}$. Here, solutions are counted with multiplicity. 
As input, our algorithm $\cert$ receives (arbitrary good approximations of) the coefficients of $\FFF$, the polydisc $\mathbf{\Delta}$, and an integer $K\in\{0,1,\ldots,d_\FFF\}$, where $d_\FFF:=\max_i d_i$ is defined as the maximum of the degrees $d_i$ of the polynomials $f_i$. As output, it returns an integer $k\in \NN\cup\{-1\}$. If $k=-1$, nothing can be said, that is, the algorithm fails to provide an answer to our request. Otherwise, $k$ equals the number of solutions of $\FFF=0$ in $\mathbf{\Delta}$. In this case, we say that the method succeeds. 
We further show that our method always succeeds if (1) $r$ is small enough, (2) $K\ge k$, and (3) the smaller polydisc $\mathbf{\Delta}':=\mathbf{\Delta}_{r'}(\bm)$, with $r':=\frac{r}{64 n (K+1)^n}$, contains a $k$-fold solution of $\FFF$. We also derive a bound on the size of $r$ that guarantees success of our method if the other two requirements are fulfilled. 
The given bound is adaptive in the sense that it does not only depend on global parameters such as the degree and the size of the coefficients of the polynomials $f_i$, but also on solution-specific parameters, that is, the multiplicity and the size of $\bz$ as well as the distances between $\bz$ and the other solutions of $\FFF$. Here, we state our main result for the special case, where $\FFF$ is defined over the integers. For a more general statement, see~Theorem~\ref{thm:main}.
\begin{theorem}
    Suppose that $\bz$ is a $k$-fold solution of a polynomial system $\FFF$ as in (\ref{formula:original_system}) with polynomials $f_i\in\ZZ[\bx]$ of total degree $d_i$ and with \emph{integer} coefficients $c_{i,\alpha}$ of bit-size less than $\tau_{\FFF}$. Then, for any $K\ge k$, there exists an $L^*\in\NN$ with 
    \begin{align*}
      L^*&=\tilde{O}\left(D_{\FFF}\cdot \max_{i=1,\ldots,n}\frac{d_{\FFF}+\tau_{\FFF}}{d_i}+D_{\FFF}\cdot\LOG(\bz)+\LOG(\partial(\bz,\FFF)^{-1})+(K+1)^n\cdot \LOG(\sigma(\bz,\FFF)^{-1})\right)
    \end{align*}
    such that, with probability at least $1/2$, the algorithm $\cert(\FFF,\mathbf{\Delta},K)$ returns $k$ for any disc $\mathbf{\Delta}=\mathbf{\Delta}_r(\bm)$ with $r\le 2^{-L^*}$ and $\|\bm-\bz\|_\infty<r$.  
    Here, we use the definitions $\LOG(\bx):=\max(1,\log\max(1,\|\bx\|_\infty))$, and
    \begin{align*}
        d_{\FFF}&:=\max_i d_i,\text{ }D_\FFF:=\prod\nolimits_{i=1}^n d_i,\\ 
        \sigma(\bz_i,\FFF)&:=\min_{j\neq i}\|\bz_i - \bz_j\|\\
        \partial(\bz_i,\FFF) &:= \prod_{j\neq i} \|\bz_i-\bz_j\|^{\mu(\bz_j, \FFF)},
    \end{align*}
    where $\bz_1,\ldots, \bz_N$ denote the distinct solutions of $\FFF$ and $\mu(\bz_i, \FFF)$ the multiplicity of $\bz_i$.
\end{theorem}

Notice that our method never yields the exact multiplicity of a solution, even in the case where there is a well separated $k$-fold solution $\bz$ in $\mathbf{\Delta}$. Instead, we only obtain the sum of the multiplicities of all solutions contained in $\mathbf{\Delta}$.
However, in the considered computational model, where only approximations of the coefficients of the input polynomials are known, it is simply not possible to achieve a stronger result. This is due to the fact that arbitrary small perturbations of the input already destroy the multiplicity structure of non-simple roots.

We see a series of applications of our method. For instance, our method can be used to verify correctness of the result provided by a numerical (non-certified) method such as 
homotopy (e.g.~\cite{DBLP:journals/toms/Verschelde99,bates2013numerically}) or subdivision methods (e.g.~\cite{DBLP:journals/jsc/MourrainP09,DBLP:conf/issac/BurrCGY08}). 
Corresponding implementations of such methods (e.g. Bertini, PHCpack, axel) are available and have proven to be efficient and reliable in practice. Suppose that such a method returns an approximation $\zeta$ of a $k$-fold solutions $\bz$ such that $\|\zeta-\bz\|_\infty<2^{-L}$, however, \emph{without any guarantee on the correctness of the result}. Now, in order to show correctness, we may run the algorithm $\cert$ with input $\FFF$, $K=k$, and $\mathbf{\Delta}=\mathbf{\Delta}_{64n(k+1)^n\cdot 2^{-L}}(\zeta)$. According to the above theorem, the method returns $k$ if the claimed result is actually correct and $L$ is large enough. 
Hence, we eventually succeed if the numerical solver provides  a sufficiently good approximation of $\bz$ together with the correct multiplicity. Again, we remark that the method does not provide a proof that there is exactly one root of multiplicity $k$, but only a proof that there $k$ roots counted with multiplicity in $\mathbf{\Delta}$.

For polynomial systems that are defined over the integers, there exist complete and certified methods (e.g.~\cite{Rouillier1999,LAZARD2009222,DBLP:conf/issac/BrandS16}) to compute isolating regions for all solutions together with the corresponding multiplicities, however, their possible application is limited in practice. In particular, if the polynomials $f_i$ are of large degree,  the running time for the necessary symbolic computations (e.g. that of a Gr\"obner Basis or resultants) becomes prohibitive. Combining our method with a numerical solver may instead yield a certified result on the existence of solutions in a certain region.

In Section~\ref{sec:experiments}, we report on preliminary implementation of our method for the special case of a bivariate system. That is, we integrated an implementation of our method in \textsc{Bisolve}~\cite{DBLP:journals/tcs/BerberichEKS13,DBLP:journals/jc/KobelS15}, a highly efficient algorithm for isolating the solutions of a bivariate polynomial systems with integer coefficients. There, it serves as an inclusion predicate to verify the existence of a $k$-fold solution of the system. Compared to the original approach in \textsc{Bisolve}, we observe a considerable improvement with respect to running time and precision demand.

\paragraph{Overview of the Algorithm.}
There exists a simple method, also known as Pellet's Theorem, to count the number of roots of a univariate polynomial $f\in\CC[x]$ in a disc $D_r(m)=\{x\in \CC: |x-m|\le r\}$ of radius $r$ centered at a point $m\in\CC$. The method works as follows: We first compute the Taylor-expansion 
\[
    f[m](x):=f(m+x)=\sum\nolimits_{i\ge 0} c_i\cdot x^i=\sum\nolimits_{i\ge 0} \frac{f^{(i)}(m)}{i!}\cdot x^i
\] 
at $m$ and then check whether $|c_k|\cdot r^k>\sum_{i\neq k}|c_i|\cdot r^i$ for some $k$. Notice that the latter inequality implies that the part $c_k\cdot x^k$ of $f[m]$ of degree $k$ dominates the remaining parts on the boundary of the disc $D_r(0)$. If this is the case, then $D_r(m)$ contains exactly $k$ roots of $f$, which follows directly from Rouch\'e's Theorem applied to $f[m]$ and its degree $k$-part $c_k\cdot x^k$. 
In~\cite{DBLP:journals/corr/BeckerS0Y15}, we give sufficient conditions on $r$ and the locations of the roots with respect to $m$ such that the above inequality is fulfilled. In particular, for $m$ being a $k$-fold root of $f$, we give a bound $r_0$ in terms of the degree of $f$ and the separation of $m$ such that Pellet's Theorem applies for any $r<r_0$; see Lemma~\ref{pellet} for details. 

Our algorithm $\cert$ can be considered as an extension of Pellet's Theorem to polynomial systems. Similar as in the one-dimensional case, we make crucial use of the fact that, for a sufficiently small neighborhood $\mathbf{\Delta}$ of a $k$-fold solution $\bz$ of $\FFF$, the system 
\[
    \FFF[\bz]:f_1(\bx+\bz)=\cdots=f_n(\bx+\bz)=0 
\]
obtained by shifting each of the polynomials $f_i$ by $\bz$ is dominated by terms of degree $k$ or less. 
Hence, in order to study the local behavior of $\FFF$ at $\bz$, it should suffice to consider the truncation $\FFF[\bz]_{\le k}$ of $\FFF[\bz]$, 
where we only consider the part $f_i[\bz]_{\le k}=\sum_{\alpha:|\alpha|\le k}c_{i,\alpha}'\cdot \bx^\alpha$ of each $f_i[\bz]=f(\bx+\bz)=\sum_{\alpha}c_{i,\alpha}'\cdot \bx^\alpha$ that is of degree $k$ or less. 
In fact, in Corollary~\ref{samemultiplicity:truncated}, we prove that, for any $K\ge k$, the system $\FFF[\bz]_{\le K}$ has a $k$-fold solution at the origin, and we give a bound on its separation in terms of the separation of $\bz$ as a solution of the original system $\FFF$. 
%
In Theorem~\ref{thm:countroots}, we even show that if $K\ge k$, and if $\|\bm-\bz\|_\infty<2^{-L}$ for a sufficiently large $L$, then we can work with $\FFF[\bm]_{\le K}$ instead of $\FFF[\bz]$. 
Namely, in this case, $\FFF[\bm]_{\le K}$ has $k$ solutions of norm less than $4\cdot 2^{-L}$, whereas all remaining solutions have considerably larger norm, that is, larger than some value that does not depend on $L$.

We now provide an overview of our approach. For the sake of simplicity, we omit technical details and only give the main ideas. Also, we do not treat any special cases, which considerably simplifies the approach when compared to the actual algorithm as given in Section~\ref{sec:algorithm}. 
We first define $L:=\lceil\log\frac{r}{32n (K+1)^n}\rceil$ such that $\frac{r}{64n (K+1)^n}\le 2^{-L}\le\frac{r}{32n (K+1)^n}=r'$. 
Obviously, we cannot check in advance whether the above requirements on $\bm$ and $L$ are fulfilled, however, we can check whether $\FFF[\bm]_{\le K}$ has a cluster of solutions near the origin. For this, we use a complete and certified 
algorithm to compute isolating 
regions of all solutions of $\FFF[\bm]_{\le K}$ that are contained in the polydisc $\mathbf{\Delta}=\mathbf{\Delta}_r(0)$. Notice that if $K$ is small compared to the degrees of the 
polynomials $f_i$, then the cost for computing the solutions of $\FFF[\bm]_{\le K}$ is much lower 
than solving the original system directly. In particular, for $K=1$, the truncated system $\FFF[\bm]_{\le K}$ becomes a linear system in $n$ variables. Now, suppose that $\mathbf{\Delta}$ 
contains $k'$ solutions of $\FFF[\bm]_{\le K}$ ($k'$ does not have to be equal to $k$) that are well separated from the remaining solutions, then we are left to show that $\FFF[\bm]$ contains the same number of solutions in $\mathbf{\Delta}$. For this, we use a generalization of Rouch\'e's Theorem that applies to 
analytic functions in $n$-dimensional complex space; see Theorem~\ref{Rouche}. This approach requires to compute a lower 
bound $\LB$ for $\|\FFF[\bm]_{\le K}(\bx)\|_\infty:=\max_i |f_i(\bx)|$ on the boundary of $\mathbf{\Delta}$ 
as well as a corresponding upper bound $\UB$ on the error $\|\FFF[\bm](\bx)_{\le K}-\FFF[\bm](\bx)\|_\infty$ that occurs when passing from $\FFF[\bm]$ to the truncated system $\FFF[\bm]_{\le K}$. While the computation of $\UB$ is straightforward (see \eqref{def:UB} in Section~\ref{sec:algorithm}), the computation of $\LB$ is more involved. 
Namely, we first compute the hidden-variable resultant $R_\ell:=\res(\FFF[\bm]_{\le K},x_\ell)\in\QQ[x_\ell]$ with respect to each of the variables $x_\ell$; see Section~\ref{sec:preliminaries} for details on the hidden variable approach. 
The roots of $R_\ell$ are the projections of the solutions of $\FFF[\bm]_{\le K}$ on the $x_\ell$-axis, and $R_\ell$ is contained in the ideal given by the polynomials $f_{i}[\bm]_{\le K}$, that is, there exist $g_{\ell,1},\ldots,g_{\ell,n}\in\QQ[\bx]$ with
\begin{align}
    R_\ell=g_{\ell,1}\cdot f_{i}[\bm]_{\le K}+\cdots +g_{\ell,1}\cdot f_{i}[\bm]_{\le K}.
\end{align}
Using a recent result~\cite{D2013} on the arithmetic Nullstellensatz, we derive upper bounds on the absolute value of the coefficients of the polynomials $g_{\ell,1}$; see Corollary~\ref{cor:arithmetichilbert} and (\ref{bound:B}) in Section~\ref{sec:algorithm}. In addition, we use our results on Pellet's Theorem from~\cite{DBLP:journals/corr/BeckerS0Y15} to derive a lower bound for $|R_\ell|$ on the boundary of the disc $D_r(0)\subset\CC$, which is the projection of the polydisc $\mathbf{\Delta}$ into one-dimensional space; see Lemma~\ref{pellet}. Combining the latter two bounds then yields $\LB$. 
Finally, we check whether $\LB>\UB$, in which case we conclude from Rouch\'e's Theorem that $\FFF[\bm]$ has the same number of solution in $\mathbf{\Delta}$ as the truncated system $\FFF[\bm]_{\le K}$. If $\UB<\LB$, we return $-1$.

In the analysis of our algorithm, we show that if $\|\bm-\bz\|_\infty<\frac{r}{64n(K+1)^n}$ for a sufficiently small $r$, then $\LB$ approximately scales like $c\cdot r^k$ for some constant $C$, whereas $\UB$ scales like $C'\cdot r^{-(K+1)L}$ for some constant $C'$. Thus, in this case, our algorithm eventually succeeds if $K\ge k$. As already mentioned, we omitted many details in the above description. In particular, for completeness, we needed to address certain special cases. In particular, this comprises the case where $\FFF[\bm]_{\le k}$ has distinct solutions whose projections on one of the coordinate axis are (almost) equal or solutions at infinity that yield roots of the hidden variable resultant. We show how to handle such situations by means of a random rotation of the coordinate system without harming the claimed complexity bounds.

\paragraph{Implementation for the Bivariate Case.}
For the special case of a polynomial system $\FFF: f_1(x_1,x_2)=f_2(x_1,x_2)=0$ in two 
variables, with $f_1,f_2\in\ZZ[x_1,x_2]$, we implemented our algorithm in \textsc{Sage}. As an oracle for computing an arbitrary 
good approximation of a solution $\bz$ of $\FFF$, we used a subroutine of the so-called 
\textsc{Bisolve} algorithm from~\cite{DBLP:journals/tcs/BerberichEKS13,DBLP:journals/jc/KobelS15}, which currently constitutes one of the fastest exact and complete algorithm for solving bivariate systems. \textsc{Bisolve} is a classical elimination approach that 
projects the solutions of the system on each of the two coordinate axis in a first step by means of 
resultant computation and root isolation. This yields a set of points on a two-dimensional grid that 
are all possible candidates for the solutions of the system. Also, the candidates can be approximated to an arbitrary precision using root refinement for univariate polynomials. Then, in a second step, in order to check whether a 
certain candidate is a solution or not, \textsc{Bisolve} combines interval arithmetic and an inclusion 
test based on bounds on the cofactors $g_1$ and $g_2$ in the representation $R=g_1\cdot f_1+g_2\cdot g_2$ of the resultant polynomials $R$ as 
an element in the ideal $\langle f_1,f_2\rangle$. This inclusion test is similar to our approach proposed in 
this paper, however, no truncation of the original system is considered. Also, it is tailored to the bivariate case and does not yield the multiplicity of a solution. In our experiments, we replaced the original inclusion test in the \textsc{Bisolve} algorithm by $\cert$ and compared the precision demand and the 
running time to that of the original variant. We observed that, for a multiplicity $k$ of $\bz$ that is small in comparison to the degrees of the input polynomials, our novel approach outperforms the original variant. At least, for the considered instances, we observed a sub-linear dependency of the needed precision on the degrees of the input polynomials. Notice that this is not in line with the derived bounds on the precision demand, which suggest at least a quadratic dependency. However, we remark that the given bounds are just worst-case bounds. In addition, our experiments can only be considered as preliminary at the current time, nevertheless we are confident that 
future work on this topic will support our first impressions.


\paragraph{Related Work.}
The literature on solving zero-dimensional polynomial systems is vast and we can only give an incomplete overview. A historical summary and an overview of known techniques can be found in~\cite{LAZARD2009222} and~\cite{dickenstein2006solving}, respectively. 

 There are roughly two different classes of methods -- numeric and symbolic methods. To the best of our knowledge, all existing complete and certified algorithms are based on elimination techniques. Using Gr\"obner bases~\cite{BUCHBERGER2006475,Faugere:2002:NEA:780506.780516} or resultants, they reduce the problem of solving a multivariate system to the problem of computing the roots of a univariate polynomial. 
Such methods further allow us to compute the coordinates of all solutions in terms of rational functions in the roots of a univariate polynomial (also called Rational Univariate Representation). A corresponding implementation~\cite{Rouillier1999} has proven to be quite efficient for systems of moderate size. 
Also, these methods are well understood in theory and corresponding complexity bounds are available~\cite{DBLP:conf/issac/BrandS16}. The major drawback of these methods is that the cost for the considered symbolic operations becomes prohibitive for larger systems.  
 In contrast, numerical methods, e.g. based on subdivision techniques or homotopy continuation, often allow us to compute good approximations of the solutions. Unfortunately, they typically fail to give guarantees on the correctness of the computed results. 

One classical numeric approach is Newton's method, see~\cite[Section 13]{DBLP:conf/issac/Rump10} for a general description and an approach that uses Newton's iteration with interval arithmetic. Shub and Smale introduced $\alpha$-theory~\cite{blum1998complexity}, where they provide conditions on a simple solution such that Newton iteration is guaranteed to yield quadratic convergence. 
Recent work~\cite{HAUENSTEIN2017575} uses Newton iteration and $\alpha$-theory to verify the existence of simple solutions  of systems of polynomial-exponential equations, however, the approach does not extend to multiple solutions. In~\cite{Zhi2017}, an extension of $\alpha$-theory is introduced that allows us to also certify multiple solutions of a polynomial system in a ``numerical fashion'' as studied in this paper.

Another very popular numeric approach are homotopy continuation methods. There has been also quite some implementation effort, see PHCpack~\cite{DBLP:journals/toms/Verschelde99} and Bertini~\cite{bates2013numerically}. 
In particular, we want to mention the work by Verschelde and Haegemans~\cite{DBLP:journals/tcs/VerscheldeH94}. From a high-level point of view, their approach is similar to ours as it is also based on Rouch\'e's theorem. Their method relies on finding a sparse part of the polynomial system that dominates the rest of the system on the border of a considered region and can be used as a better starting system for homotopy based techniques. The main differences to our approach are the following. First, we use our technique to directly certify the existence of a zero, not only in order to construct a starting system for a numerical method. Moreover, the system that we use in order to approximate the input system is of lower degree, more precisely our ``dominating part'' is always of degree $k$ if $k$ is the multiplicity of the zero in the given region.\footnote{Note that it is not strictly necessary to know this parameter $k$, since a binary search for $k$ can find a good enough approximation.} In contrast to their result, we also show that the precision that is needed in order to do so directly depends on the arrangements of the zeros of the system. Van der Hoeven~\cite{vdH:homotopy} describes methods for tracking homotopy paths in a certified manner. Using an analytic variant of the geometric resolution method~\cite{Giusti1995}.   

Subdivision methods~\cite{DBLP:journals/jsc/MourrainP09,DBLP:conf/issac/BurrCGY08} are usually incomplete in the sense that they only provide exclusion predicates and lack inclusion predicates. Thus they can be used in order to compute regions that are guaranteed to be free of solutions to the system but cannot ultimately guarantee that a region contains a zero. We want to stress that our work now provides an inclusion predicate that could be included in these approaches in order to turn these methods into complete methods. 


\section{Mise en place}\label{sec:preliminaries}
\subsection{Notation and Definitions}\label{subsec:defnot}
We start by introducing frequently used notation and important definitions.
\begin{enumerate}
	\item For a point $\bx=(x_1,\ldots,x_n)\in\CC$, we define the \emph{norm $\|\bx\|$ of  $\bx$} to be the $\infty$-norm by default, that is, $$\|\bx\|:=\|\bx\|_\infty=\max_{i\in [n]} |x_i|.$$ In addition, we define $M(\bx):=\max(1,\|\bx\|)$ and $\LOG(\bx) = M(\log(M(\bx)))$.
	\item For a polynomial $f=\sum_{\alpha}c_{\alpha}\bx^{\alpha}\in\CC[\bx]$, we define
	      $$
	      \deg(f):=\max_{\alpha=(\alpha_1,\ldots,\alpha_n):c_{\alpha\neq 0}} \alpha_1 + \ldots + \alpha_n
	      $$ to be the \emph{(total) degree of $f$}. The \emph{norm of $f$} is defined as
	      $$\|f\|:=\|f\|_\infty=\max_{\alpha}|c_{\alpha}|.$$
	      We further define $\tau_f:=\lceil \log \|f\|\rceil$
	\item For a polynomial system $\FFF=(f_1, \ldots, f_n)$, with $f_i\in\CC[\bx]$ of total degree $d_i$, we define
	      \[
	      	d_\FFF:=\max_{i\in[n]}d_i,\quad D_\FFF:=\prod_{i\in [n]}d_i,\quad\text{and}\quad\tau_\FFF:=\max_{i\in [n]}\tau_{f_i}.
	      \]
	     $D_{\FFF}$ is also called the \emph{B\'ezout bound} in the literature. It constitutes an upper bound on the total number of solutions (counted with multiplicities) of a zero-dimensional system $\FFF$. For a system $\FFF$ with generic coefficients, it actually equals the number of solutions. 
	\item We further say that a polynomial $f=\sum_{\alpha}c_{\alpha}\bx^{\alpha}\in\ZZ[\bx]$ with \emph{integer} coefficients has \emph{magnitude} $(d,\tau)$ if $d_f\le d$ and $\tau_f\le \tau$. A system $\FFF=(f_1,\ldots,f_n)$ with $f_i\in\ZZ[\bx]$ has magnitude $(d,\tau)$ if each polynomial has magnitude $(d,\tau)$.
	\item For a polynomial $f=\sum_{\alpha}c_{\alpha}\bx^{\alpha}\in\CC[\bx]$ and a positive integer $\kappa$, we say that $\phi=\sum_{\alpha}\tilde{c}_{\alpha}\bx^{\alpha}$ is an (absolute) $\kappa$-bit approximation of $f$ if each $\tilde{c}_{\alpha}$ is a dyadic number of the form $(m+m'\cdot \mathbf{i})\cdot 2^{-(\kappa+1)}\in\QQ+\mathbf{i}\cdot\QQ$, with $m,m'\in\ZZ$, and $\|f-\phi\|\le 2^{-\kappa}$. In other words, each $\tilde{c}_{\alpha}$ approximates $c_{\alpha}$ to $\kappa$ bits after the binary point.
	\item For $\bz\in\CC^n$ and a polynomial $f\in\CC[\bx]$, we define
	      \[
	      	f[\bz](\bx):=f(\bx+\bz)=\sum_{\mathbf{\alpha}}c_{\mathbf{\alpha}} (\bx + \bz)^\mathbf{\alpha}
	      \]
	      to be the \emph{shift of $f$ to $\bz$}. For a system $\FFF=(f_1, \ldots, f_n)$, we define the \emph{shift of $\FFF$ to $\bz$} as $\FFF[\bz] = (f_1[\bz], \ldots, f_n[\bz])$.
	\item For $k\in[d]$, we denote with $f_{\le k}:=\sum_{\mathbf{\alpha}:|\alpha|\le k}c_{\mathbf{\alpha}} \bx^\mathbf{\alpha}$ the \emph{truncation of $f$ of degree $k$}.
	      For a system $\FFF=(f_1,\ldots, f_n)$, we define the \emph{truncation of $\FFF$ of degree $k$} as $\FFF_{\le k}=({f_1}_{\le k},\ldots, {f_n}_{\le k})$.
\end{enumerate}

\subsection{Error Bounds for Shifting, Truncation, and Rotation}

We first collect some bounds on the size of $|f(\bz)|$ and $\|f[\bz]\|$ depending on the modulus of some point $\bz\in\CC^n$ and the norm $\|f\|$ of some polynomial $f\in\CC[\bx]$. We also give bounds on the error that occurs when computing $f(\bz)$ or $f[\bz]$ not exactly at $\bz$ but at a nearby point $\zeta$.

\begin{lemma}\label{generalbounds}
	Let $f(\bx)=\sum_{\mathbf{\alpha}}c_{\mathbf{\alpha}} \bx^\mathbf{\alpha}\in\CC[\bx]=\CC[x_1,\ldots,x_n]$ of total degree $d$ and with $\|f\|\le 2^\tau$. Moreover, let $k\in [d]=\{1,\ldots,d\}$, $\bz\in\CC^n$, and $\zeta$ be an approximation of $\bz$ with $\|\zeta-\bz\|<2^{-L}$, then it holds:
	\begin{enumerate}[(a)]
		\item\label{parta} $|f_{\le k}(\bz)|\le \binom{n+k}{k}\cdot 2^\tau\cdot M(\bz)^k$, and in particular $|f(\bz)|<\binom{n+d}{d}\cdot 2^\tau\cdot M(\bz)^d$.
		\item\label{partb} If $\|\bz\|\le 1$, then $|f(\bz)-f_{\le k}(\bz)|\le \|\bz\|^{k+1}\cdot \binom{n+d}{d}\cdot 2^\tau$.
		\item\label{partc} $|f(\zeta) - f(\bz)| \le 2^{\tau - L}\cdot\binom{n+d}{d} \cdot M(\bz)^d$.
		\item\label{partd} $\|f[\bz]\|<d^n\cdot\binom{n+d}{d}\cdot 2^\tau\cdot M(\bz)^d$ and $\|f[\zeta]-f[\bz]\|<2^{\tau - L}\cdot d^n\cdot \binom{n+d}{d} \cdot M(\bz)^d$.
	\end{enumerate}
\end{lemma}
\begin{proof}
	Part~\eqref{parta} and~\eqref{partb} follow immediately from the fact that $f_{\le k}$ has at most $\binom{n+k}{k}$ coefficients and each occurring term $c_\alpha\cdot \bz^{\alpha}$ has absolute value bounded by $2^\tau\cdot \|\bz\|^{|\alpha|}$.
	Part~\eqref{partc} is a direct consequence of~\cite[Theorem 12]{DBLP:journals/comgeo/MehlhornOS11}, which provides general bounds on the error when evaluating a multivariate polynomial using floating point computation. For the last claim, notice that
	\[
		f[\bz] = \sum_{\alpha \in \ZZnneg^n} \tfrac{\partial^\alpha f (\bz)}{\alpha!} \bx^\alpha \text{ and }
		f[\zeta]= \sum_{\alpha \in \ZZnneg^n} \tfrac{\partial^\alpha f (\zeta)}{\alpha!} \bx^\alpha,
	\]
	where $\alpha=(\alpha_1,\ldots,\alpha_n)$,  $\alpha!:=\alpha_1!\cdots\alpha_n!$, and $\partial^\alpha f:=\tfrac{\partial^{|\alpha|}f}{\partial x_1^{\alpha_1}\cdots\partial x_n^{\alpha_n}}$.
	The polynomials $\tfrac{\partial^\alpha f}{\alpha!}$ have total degree bounded by $d$ and their norm is upper bounded by $2^\tau\cdot d^n=2^{\tau+n\log d}$. Hence, Part~\eqref{parta} implies the first part of~\eqref{partd}. The second part follows from Part~\eqref{partc} because, for any $\alpha$, it holds that
	\[
		|\tfrac{\partial^\alpha f (\bz)}{\alpha!} - \tfrac{\partial^\alpha f (\zeta)}{\alpha!}|
		\le 2^{\tau - L}\cdot d^{n}\cdot \binom{n+d}{d}\cdot M(\bz)^d.\qedhere
	\]
\end{proof}
We further provide the following lemma that investigates the influence of considering only an approximation of a polynomial $f$ when looking at shift and truncation.
\begin{lemma}\label{lemma:truncated approximation}
	Let $f(\bx)\in\CC[\bx]$ be a polynomial of total degree $d$ with norm $\|f\|\le 2^\tau$, and let $\bz\in\CC^n$ and $\zeta$ such that $\|\zeta-\bz\|<2^{-L}$. Furthermore, let $\phi$ be an approximation of $f[\zeta]_{\le k}$ of total degree at most $k$, with $k\in [d]$, such that
	$\|\phi - f_i[\zeta]_{\le k}\|\le 2^{-(k+1)L}$.
	Then, for any $\bx$ with $\|\bx\| \in [2^{-L},1]$, it holds
	\begin{align*}
		|\phi(\bx)-f[\zeta](\bx)|
		\le
		\|\bx\|^{k+1} \cdot d^n 2^{\tau + 1} [M(\zeta)\cdot (n+d)^2]^d.
	\end{align*}
\end{lemma}
\begin{proof}
	We first observe that using the triangle inequality, simple bounds on the number of monomials of lower ($\le k$) and higher ($\ge k+1$) degree, and the fact that $\|\bx\|\le 1$ yields
	\begin{align*}
		|\phi(\bx)-f[\zeta](\bx)|
		  & \le |f[\zeta](\bx)-f[\zeta]_{\le k}(\bx)| + |\phi(\bx)-f[\zeta]_{\le k}(\bx)| \\
		  & \le \|\bx\|^{k+1}\cdot \|f[\zeta]-f[\zeta]_{\le k}\|\cdot \binom{n+d}{d}
		+ \|\phi-f[\zeta]_{\le k}\|\cdot \binom{n + k}{k}.
	\end{align*}
	Then, applying Lemma~\ref{generalbounds} part~\eqref{partd} to the left summand and the condition on the approximation $\|\phi - f[\zeta]_{\le k}\|\le 2^{-(k+1)L}$, we conclude that
	\begin{align*}
		|\phi(\bx)-f[\zeta](\bx)|
		  & \le \|\bx\|^{k+1}\cdot d^n\cdot\binom{n+d}{d}^2\cdot 2^\tau\cdot M(\zeta)^d+2^{-(k+1)L}\cdot\binom{n + k}{k} \\
		  & \le \|\bx\|^{k+1}\cdot
		[d^n \cdot 2^\tau \cdot (n+d)^{2d}\cdot M(\zeta)^d+(n+d)^d]\\
		  & \le \|\bx\|^{k+1} \cdot d^n 2^{\tau + 1} [M(\zeta)\cdot (n+d)^2]^d,
	\end{align*}
	where the second to last inequality follows from
	$\|\bx\|\ge 2^{-L}$.
\end{proof}

In our algorithm, we will consider a transformation of the coordinate system induced by a rotation $\bx\mapsto S\cdot \bx$, where $S\in \operatorname{SO}(n)$ is a rotation matrix with rational entries. The following lemma quantifies the impact of such a rotation on the bit-size of the coefficients of a given polynomial $f$.
\begin{lemma}\label{lem:rotation_coeffsize}
	Let $f=\sum_{\alpha}c_{\alpha}\cdot \bx^{\alpha}\in\CC[\bx]$ be a polynomial of total degree $d$ and $S\in \operatorname{SO}(n)$ be a rotation matrix. Then, $f^*:=f\circ S$, it holds that $\|f^*\|\le 2^{\tau_\FFF}\cdot \binom{n+d}{d}^2$.
\end{lemma}
\begin{proof}
	Notice that each of the entries $a_{r,s}$ of the rotation matrix $S=(a_{r,s})_{r,s}$ has absolute value at most $1$. Thus, $f^*(\bx)=f\circ S(\bx)=\sum_{\alpha:c_{\alpha}}c_{\alpha}\cdot [(a_{11}x_1+\cdots+a_{1,n}x_n)^{\alpha_1}\cdots (a_{n1}x_1+\cdots+a_{n,n}x_n)^{\alpha_n}]$ has coefficients of absolute value bounded by
	$
		2^{\tau_\FFF}\cdot \binom{n+d}{d}^2$
		as, when expanding the product $(a_{11}x_1+\cdots+a_{1,n}x_n)^{\alpha_1}\cdots (a_{n1}x_1+\cdots+a_{n,n}x_n)^{\alpha_n}$ for a fixed $\alpha$, there can be at most $\binom{n+d}{d}$ terms contributing to a specific monomial $\bx^{\alpha'}$.
		\end{proof}

\subsection{The Hidden-Variable Approach}

Let us assume that an arbitrary zero-dimensional system $\mathcal{F}=(f_1,\ldots, f_n)$ as in~\eqref{formula:original_system} is given. That is, $f_i$ has total degree $d_i$, $\|f_i\|<2^{\tau_i}$ for all $i$, and it is assumed that the total number of solutions of $\mathcal{F}=0$, also at ``infinity'' (see the considerations below for an explanation), is finite.
We now briefly describe the so-called hidden-variable approach that allows us to project the zeros of the system on an arbitrary coordinate axis. For more details, we recommend the excellent textbook~\cite{cox2005using} by Cox, Little, and O'Shea.

In a first step, we consider a homogenization of the system, that is, we introduce an additional (homogenizing) variable $x_{n+1}$ and multiply each occurring term in each $f_i$ with a suitable power of $x_{n+1}$  such that the so obtained polynomials $f_i^h\in\CC[x_1,\ldots, x_{n+1}]$ are homogenous and of total degree $d_i$, respectively; see also the example below. Notice that each solution $(x_1,\ldots,x_n)\in\CC$ of $\mathcal{F}=0$ yields a solution $(x_1,\ldots,x_n,1)$ of the  homogenized system
\begin{align}\label{formula:homog_system}
	\mathcal{F}^h:\quad f_1^h(x_1,\ldots, x_{n+1})=\ldots =f_n^h(x_1,\ldots, x_{n+1}) = 0
\end{align}
In addition, if $(x_1,\ldots,x_{n+1})\in\mathbb{C}^{n+1}$ is a solution of  $\mathcal{F}^h=0$, then $(t\cdot x_1,\ldots,t\cdot x_{n+1})$ is a zero of $\mathcal{F}^h$ for all $t\in\mathbb{C}$. In particular, if $x_{n+1}\neq 0$, we can set $t=1/x_{n+1}$, which yields the solution $(x_1/x_{n+1},\ldots,x_{n}/x_{n+1})$ of $\FFF=0$.
It is thus preferable to consider the set $S$ of solutions of the above homogenized system as a set of points in the $n$-dimensional projective space $\PP^n$. The set $S$ then decomposes into the set $S_{<\infty}=\{(x_1:\ldots:x_{n+1})\in S:x_{n+1}=1\}$ of so-called \emph{affine solutions}, for which $x_{n+1}=1$, and the set $S_{\infty}=\{(x_1:\ldots:x_{n+1}):x_{n+1}=0\}$ of \emph{solutions at infinity}, for which $x_{n+1}=0$. Notice that there is a one-to-one correspondence between the affine solutions of the homogenized system and the solutions of the original system~\eqref{formula:original_system}.

As mentioned above, we aim to compute the projections of the solutions of $\mathcal{F}=0$ on one of the coordinate axis, say w.l.o.g., $x=x_1$. For this, suppose that we fix some value $\xi$ for $x_1$. Plugging $x_1=\xi$ into the initial system then yields the specialized system
\begin{align*}
	\mathcal{F}^{[\xi]}:\quad f_1^{[\xi]}(x_2,\ldots,x_{n})=\ldots = f_n^{[\xi]}(x_2,\ldots, x_{n}) = 0,
\end{align*}
with $f_i^{[\xi]}(x_2,\ldots,x_{n}):=f_i(\xi,x_2,\ldots,x_{n})$
and the corresponding homogenized system
\begin{align}\label{formula:homogeneous_system:specialized}
	(\mathcal{F}^{[\xi]})^h:\quad (f_1^{[\xi]})^h(x_2,\ldots, x_{n+1})=\ldots = (f_n^{[\xi]})^h(x_2,\ldots, x_{n+1}) = 0,
\end{align}
where $(f_i^{[\xi]})^h$ denotes the homogenization of $f^{[\xi]}_i$.
Notice that, in general, $(f_i^{[\xi]})^h$ does not equal $(f_i^h)^{[\xi]}=f_i^h(\xi,x_2,\ldots,x_{n+1})$, that is, we cannot deduce the system in (\ref{formula:homogeneous_system:specialized}) from plugging $\xi$ into the homogenized system in (\ref{formula:homog_system}). The reason is that the total degree of $f_i$ may become smaller for certain values for $\xi$, and thus homogenization does not commute with specialization.\\

\noindent\textit{Example.} For $f:=x_1 x_2^3-2 x_2^3+x_3 x_1+x_3^2$ and $\xi=2$, we have $f^h=x_1 x_2^3-2 x_2^3 x_4+x_3 x_1 x_4^2+x_3^2 x_4^2$, $f^{[\xi]}=f(2,x_2,x_3)=2x_3+x_3^2$, and $(f^{[\xi]})^h=2x_3x_4+x_3^2$, which does not equal $(f^h)^{[\xi]}=f^h(2,x_2,x_3,x_4)=2 x_2^3-2 x_2^3x_4+2x_3 x_4^2+x_3^2 x_4^2$.\\

You may notice that (\ref{formula:homogeneous_system:specialized}) is a polynomial system consisting of $n$ homogenous polynomials in $n$ variables. If the initial homogenized system had a solution with $x_1=\xi$, then this would yield a solution of (\ref{formula:homogeneous_system:specialized}) and vice versa. In other words, $\xi$ would be the projection of a solution of the initial system.
The following important result now gives a necessary and sufficient criteria to check whether this is actually the case.

\begin{theorem}[\cite{cox2005using}, Chapter 3, Theorems 2.3 and 3.1]\label{thm:resultant}
	Let $\GGG$ be a system of $n$ homogeneous polynomials in $n$ variables of total degrees $d_1,\ldots,d_n$. Then, there is a unique polynomial
	\footnote{We remark that $\res$ only depends on the actual degrees of the polynomials.} $\res(\GGG)=\res_{d_1,\ldots,d_n}\in\ZZ[\mathbf{u}]$ in the coefficients $\mathbf{u}$ of $\GGG$ if and only if $\GGG(\bx) = 0$ has a non-trivial solution $\bx\in\mathbb{P}^{n-1}$. $\res(\GGG)$ is homogeneous in the variables of $f_i$ of degree $d_1\cdots d_{i-1}\cdot d_{i+1}\cdots d_n$
	and its total degree equals $\sum_{i=1}^{n}d_1\cdots d_{i-1}\cdot d_{i+1}\cdots d_n$.
\end{theorem}

\noindent\textit{Example.}  The homogeneous system $\GGG: ax_1^2+bx_1x_2+cx_2^2=dx_1+ex_2=0$ (with general coefficients $a,b,c,d$, and $e$) has a solution in $\mathbb{P}^1$ if and only if the involved coefficients fulfill the equality $\res(\GGG)=\res_{2,1}=ae^2-be+cd=0$.\\

For an arbitrary polynomial system $\GGG$ consisting of $n+1$ (not necessarily homogenous) polynomials in $\CC[x_1,\ldots,x_n]$, we simply define $\res(\GGG)=\res(\GGG^h)$. Since $\GGG$ has the same coefficients as $\GGG^h$, it still holds that $\res(\GGG)$ is a polynomial in the coefficients of $\GGG$. In addition, since there is a one-to-one correspondence between the solutions of $\GGG$ and the affine solutions of $\GGG^h$, it follows that $\res(\GGG) = 0$ if and only if $\GGG^h=0$ has a solution in $\mathbb{P}^n$.

Now, in order to compute all values $\xi$ such that there exists a solution $(x_1,\ldots,x_n)$ of our initial system $\mathcal{F}=0$ with $x_1=\xi$, we aim to apply the above theorem to the system as defined in (\ref{formula:homogeneous_system:specialized}), however we now consider $\xi$ as an indeterminate (so called \emph{hidden variable}) rather than a fixed value. There are some subtleties with this approach. In particular, the degrees of the polynomials
$f_i^{[\xi]}$ may be different for certain values for $\xi$, which is crucial as the definition of the resultant polynomial $\res$ strongly depends on the degrees of the given polynomials. However, we can avoid such critical situations if we assume that the given polynomials $f_i$ fulfill some mild prerequisites.

\begin{lemma}\label{lem:mildconditions}
	Suppose that each polynomial $f_i$ contains a term of total degree $d_i$ \emph{that does not depend on $x_1$} and write $$f_i=\sum_{\alpha=(\alpha_2,\ldots,\alpha_n)} c_{i,\alpha}(x_1)\cdot x_2^{\alpha_2}\cdots x_n^{\alpha_n}\in\mathbb{C}[x_1][x_2,\ldots,x_n]$$
	as a polynomial in $x_2,\ldots,x_n$ with coefficients $c_{i,\alpha}\in\mathbb{C}[x_1]$. Furthermore, let
	$$
	F_i:=\sum_{\alpha=(\alpha_2,\ldots,\alpha_n)} c_{i,\alpha}(x_1)\cdot x_2^{\alpha_2}\cdots x_n^{\alpha_n}\cdot x_{n+1}^{d_i-\alpha_2-\ldots-\alpha_n}
	$$
	be its corresponding homogenization (with respect to the variables $x_2,\ldots,x_n$), then it holds:
	\begin{itemize}
		\item[(a)] For all $\xi\in\mathbb{C}$, we have $(f_i^{[\xi]})^{h}=F_i^{[\xi]}$ and $(f_i^{[\xi]})^h$ has total degree $d_i$.
		\item[(b)] Each root $x_1=\xi\in\CC$ of $R(x_1):=\res(F_1,\ldots,F_n)$ yields a solution $(x_1,\ldots,x_n)\in\CC^{n}$ of $\mathcal{F}=0$ with $x_1=\xi$ and vice versa.
	\end{itemize}
\end{lemma}

\begin{proof}
	Part (a) follows directly from the fact that the total degree of $f_i^{[\xi]}$ is equal to $d_i$ for all $\xi$ as there exists a term of degree $d_i$ that does not depend on $\xi$. For (b), we first remark that the resultant of the polynomials $F_i$ is a polynomial in the coefficients of the $F_i$, and thus a polynomial in $x_1$. Since the degree of each $f_i$ does not depend on the choice of $x_1=\xi$, we also have $R(\xi)=\res(F_1|_{x_1=\xi},\ldots,F_n|_{x_1=\xi})$. Now, let $x_1=\xi$ be a complex root of $R$, then according to Theorem~\ref{thm:resultant}, there must exist a solution $(\xi_2:\ldots:\xi_{n+1})\in\mathbb{P}^{n-1}$ of the system $(F_i|_{x_1=\xi})_{i=1,\ldots,n}$. In order to prove that this solution is an affine solution (i.e. a solution of $\FFF$), we assume for contradiction that $\xi_{n+1}=0$. Plugging $x_{n+1}=0$ into the polynomials $F_i$ yields
	\[
		F_i|_{x_{n+1}=0}=\sum_{\alpha=(\alpha_2,\ldots,\alpha_n):\alpha_2+\ldots+\alpha_n=d_i} c_{i,\alpha}(x_1)\cdot x_2^{\alpha_2}\cdots x_n^{\alpha_n}.
	\]
	Hence, each of the terms $c_{i,\alpha}(x_1)$ occurring in the above sum is a constant that does not depend on $x_1$. Since $(\xi:\xi_2:\ldots:\xi_n)$ is a solution of the system $F_{i}|_{x_{n+1}=0}$, we conclude that $(x_1:\xi_2:\ldots:\xi_n)$ is a solution of $F_{i}|_{x_{n+1}=0}$ for any $x_1$. This contradicts our assumption that $\mathcal{F}$ has only finitely many solutions.
	It follows that $(\xi_2/\xi_{n+1},\ldots,\xi_n/\xi_{n+1},1)$ is a solution of $(F_i|_{x_1=\xi})_{i=1,\ldots,n}$, and thus $(\xi,\xi_2/\xi_{n+1},\ldots,\xi_n/\xi_{n+1})$ is a solution of $\mathcal{F}=0$.
	For the other direction, let $(\xi_1,\ldots,\xi_n)$ be a solution of $\mathcal{F}=0$, then $(\xi_1:\ldots:\xi_n:1)$ is an affine solution of the corresponding homogenized system, and thus $(\xi_2:\ldots:\xi_n:1)$ a solution of the system $(F_i|_{x_1=\xi})_{i=1,\ldots,n}=0$. This implies that $R(\xi_1)=0$.
\end{proof}

Obviously, the above considerations apply for any coordinate (hidden-variable) $x_k$ onto which we aim to project the solutions. The corresponding resultant polynomial $\RES{\FFF}{x_k}\in\CC[x_k]$ is called the \emph{hidden-variable resultant with respect to $x_k$}. The following theorem~\cite{DBLP:conf/issac/BrandS16} bounds the cost for computing the hidden-variable resultant in the special case where the polynomials $f_i$ have integer coefficients. The technique is based on a method due to Emiris and Pan~\cite{DBLP:journals/jc/EmirisP05} and an asymptotically fast algorithm for determinant computation due to Storjohann~\cite{DBLP:journals/jc/Storjohann05}.

\begin{theorem}[{\cite[Prop.~1]{DBLP:conf/issac/BrandS16}}]
    \label{complexity:resultant}
	Let $\FFF=(f_i)_{i=1,\ldots,n}$ be a polynomial system with integer polynomials $f_i\in\mathbb{Z}[x_1,\ldots,x_n]$ of magnitude $(d,\tau)$.
	There is a Las-Vegas algorithm to compute $\RES{\FFF}{x_k}$ in an expected number of bit operations bounded by\footnote{$\omega$ denotes the exponent in the complexity of matrix multiplication. The current record bound for $\omega$ is $\omega\le 2.3728639$ according to~\cite{DBLP:conf/issac/Gall14a}}
	$$\tilde{O}(n^{(n-1)(\omega+1)}(d+\tau)d^{(\omega+2)n-\omega-1}).$$
\end{theorem}

We further remark that a root $\xi$ of $\RES{\FFF}{x_k}$ might origin from several solutions $\bz=(z_1,\ldots,z_n)$ of $\FFF=0$ sharing the same $x_k$-coordinate $x_k=\xi$. Under the requirements from Lemma~\ref{lem:mildconditions}, it holds that the multiplicity of $\xi$ as a root of $\RES{\FFF}{x_k}$ equals the sum of the multiplicities of all these solutions $\bz$. 
Also, the roots of $\RES{\FFF}{x_k}$ are exactly the projections of the finite solutions onto the $x_k$-coordinate, and vice versa. Furthermore, if there no solution at infinity, then $\RES{\FFF}{x_k}$ has degree $D_{\FFF}$ as the system has exactly $D_{\FFF}$ solutions (counted with multiplicity), which are all finite, and the roots of $\RES{\FFF}{x_k}$ are exactly the projections of these solutions onto the $x_k$-coordinate.

\begin{lemma}\label{lcresonlydependsonhighestdegree}
	Let $\FFF=(f_i)_{i=1,\ldots,n}$, with $f_i=\sum_{\alpha:|\alpha|\le d_i}c_{i,\alpha}\cdot \bx^{\alpha} \in\mathbb{C}[\bx]$ of total degree $d_i$, be a polynomial system in $n$ variables $\bx=(x_1,\ldots,x_n)$ with general coefficients $c_{i,\alpha}$.  Consider the decomposition
	$$f_i(\bx)=\underbrace{\sum_{\alpha:|\alpha|= d_i} c_{i,\alpha}\cdot \bx^{\alpha}}_{:=f_{i,d_i}(\bx)}+\underbrace{\sum_{\alpha:|\alpha|< d_i} c_{i,\alpha}\cdot \bx^{\alpha}}_{=:f_{i,<d_i}(\bx)}$$
	of each $f_i$ into a sum of terms of degree $d_i$ and into a sum of terms of degree less than $d_i$. Then, for any $k\in[n]$, it holds that the leading coefficient $\operatorname{LC}(\RES{\FFF}{x_k})$ of the (general) hidden variable resultant $\RES{\FFF}{x_k}\in\ZZ[c_{i,\alpha}][x_k]$ only depends on the coefficients $c_{i,\alpha}$ of $f_{i,d_i}$ (i.e. on the coefficients $c_{i,\alpha}$ with $|\alpha|=d_i$).
\end{lemma}

\begin{proof}
	Let $x_{n+1}$ be a homogenizing variable and
	$$
	f_i^h=\sum_{\alpha:|\alpha|= d_i}c_{i,\alpha}\cdot \bx^{\alpha}+\sum_{\alpha:|\alpha|< d_i}c_{i,\alpha}\cdot \bx^{\alpha}\cdot x_{n+1}^{d_i-|\alpha|}
	$$
	be the corresponding homogenization of $f_i$.
	For generic choice of the coefficients $c_{i,\alpha}$ with $|\alpha|=d_i$, the above system is
	zero-dimensional and has no solution at infinity. Namely, for $x_{k+1}=0$, the system writes as $f_i^h(x_1,\ldots,x_n,0)=f_{i,d_i}$, and a generic system of $n$ homogenous polynomials in $n$ variables has no solution. Thus, there exists no solution at infinity, which also rules out the
	possibility of the system being non zero-dimensional. Now, suppose that the coefficients are
	generically chosen such that all solutions are finite. Then, the total number of solutions equals
	the B\'ezout number $D_{\FFF}$ and the degree of $\RES{\FFF}{x_k}$ equals $D_{\FFF}$. According to Theorem~\ref{thm:resultant}, $\operatorname{LC}(\RES{\FFF}{x_k})\in\ZZ[c_{i,\alpha}]$
	is a polynomial in the coefficients $c_{i,\alpha}$. Now, if $\operatorname{LC}(\RES{\FFF}{x_k})$ would depend on some
	coefficient $c_{i,\alpha}$ with $|\alpha|<d_i$, then, for generic choice of all other coefficients,
	we could choose such a $c_{i,\alpha}$ in a way such that the leading coefficient becomes zero,
	and thus $\deg\RES{\FFF}{x_k}<D_{\FFF}$, a contradiction. This shows that, for generic choice of the coefficients $c_{i,\alpha}$, the leading coefficient $\operatorname{LC}(\RES{\FFF}{x_k})$ does not depend on the coefficients of the polynomials $f_{i,<d_i}$. From this, we conclude that  $\operatorname{LC}(\RES{\FFF}{x_k})$ does not depend on the coefficients of the polynomials $f_{i,<d_i}$ in general.
\end{proof}

\begin{corollary}\label{artificialperturbation}
	Let $\FFF=(f_i)_{i=1,\ldots,n}$ be an arbitrary polynomial system as in Lemma~\ref{lcresonlydependsonhighestdegree} with $d_i=d$ for all $i$, and let
	\[
		\bar{\FFF}:\quad \bar{f_i}:=\sum_{\alpha:|\alpha|=d+1}^n c_{i,\alpha}\cdot \bx^\alpha+f_i, \quad\text{with }c_{i,\alpha}\in\ZZ\text{ for all }\alpha\text{ with }|\alpha|=d+1,
	\]
	be the system obtained by adding polynomials of the form $\sum_{\alpha:|\alpha|=d+1}^n c_{i,\alpha}\cdot \bx^\alpha$ to each $f_i$. If $\bar{\FFF}$ does not have any solution at infinity (which is the case for generic choice of the coefficients $c_{i,\alpha}$), then it holds that
	$
	\operatorname{LC}(\RES{\bar{\FFF}}{x_k})\in\ZZ_{\neq 0}.
	$
\end{corollary}

\begin{proof}
	If $\bar{\FFF}$ has no solution at infinity, then $\bar{\FFF}$ is zero-dimensional and, in addition, $\RES{\bar{\FFF}}{x_k}$ has degree $D_{\bar{\FFF}}=(d+1)^n$. From Lemma~\ref{lcresonlydependsonhighestdegree}, we further conclude that $\operatorname{LC}(\RES{\bar{\FFF}}{x_k})$ only depends on the coefficients $c_{i,\alpha}$ of the degree $(d+1)$-parts $\bar{f}_{i,d+1}$ of the polynomials $\bar{f_i}$. Hence, we have
	$\operatorname{LC}(\RES{\bar{\FFF}}{x_k})\in\ZZ_{\neq 0}$.
\end{proof}

\noindent\textit{Example:} Let $\bar{f_{i}}=\sum_{j=1}^n a_{ij}\cdot x_j^{d+1}+f_{i}$ with $f_{i}\in\mathbb{C}[\bx]_{\le d}$ polynomials of total degree at most $d$. Then, it holds that $$\RES{\bar{\FFF}}{x_k}=\pm \det(a_{ij}) \cdot x_k^{(d+1)^n}+\cdots$$
Namely, if $\det (a_{i,j})\neq 0$, then $\bar{\FFF}$ has no solution at infinity as each such solution would yield a non-trivial solution of the linear system $\sum_{j=1}^n a_{i,j}\cdot X_j=0$. Thus, $\bar{\FFF}$ is zero-dimensional in this case and $\RES{\bar{\FFF}}{x_k}$ has degree $D_{\bar{\FFF}}=(d+1)^n$. From Lemma~\ref{lcresonlydependsonhighestdegree}, we further conclude that $\operatorname{LC}(\RES{\bar{\FFF}}{x_k})$ only depends on the coefficients $a_{i,j}$ of the degree $(d+1)$-parts $\bar{f}_{i,d+1}$ of the polynomials $\bar{f_i}$. Hence, we have
$\operatorname{LC}(\RES{\bar{\FFF}}{x_k})=\operatorname{LC}(\RES{\bar{f}_{1,=d+1},\ldots,\bar{f}_{n,=d+1}}{x_k})$, and
using Theorem~2.3 and Theorem~3.5 in~\cite{cox2005using} further shows that 
\begin{align*}
\RES{\bar{f}_{1,=d+1},\ldots,\bar{f}_{n,=d+1}}{x_k}&=\det (a_{i,j})^{(d+1)^n}\cdot \RES{(x_1^{d+1},\ldots,x_n^{d+1})}{x_k}\\
&=\pm\det (a_{i,j})^{(d+1)^n}\cdot x_k^{(d+1)^n}.
\end{align*}\smallskip

It is also well known (e.g. this follows from Theorem~\ref{thm:dandrea} below) that $\RES{\FFF}{x_k}$ is contained in the ideal $\mathcal{I}:=\langle f_1,\ldots,f_n \rangle$ defined by the polynomials $f_1,\ldots,f_n$. In particular, for polynomials $f_i\in\mathbb{Z}[x_1,\ldots,x_n]$ with integer coefficients, this guarantees the existence of an integer $\lambda$, with $\lambda\neq 0$, and polynomials $g_{i}\in\mathbb{Z}[x_1,\ldots,x_n]$ with
\begin{align}\label{idealrep}
	\lambda\cdot\RES{\FFF}{x_k}=g_{1}\cdot f_1+\cdots g_n\cdot f_n.
\end{align}
Recent work~\cite{D2013} allows us to bound the magnitude of the polynomials $g_i$ as well as the size of $\lambda$.
For this, we first write $f_i=\sum_{\alpha} c_{i,\alpha}(x_k){\bx
	^{\alpha}_{\neq k}}$ as a polynomial in $\bx_{\neq k}$ with coefficients $c_{i,\alpha}\in\ZZ[x_k]$, where $\bx_{\neq k}$ denotes all but the $k$'th variable. We further introduce a variable $u_{i,\alpha}$ for every coefficient polynomial $c_{i,\alpha}$.
Let $\bu_i=(u_{i,\alpha})_{\alpha}$ be the variables corresponding to the polynomial $f_i$, and let $\bu=(\bu_1,\ldots,\bu_n)$ denote the variables for all polynomials. Then, $\FFF$ can be considered as a system consisting of $n$ polynomials in $n-1$ variables $\bx_{\neq k}$ with coefficients $\bu$. Thus, its resultant $\res(\FFF)$ is a polynomial in $\QQ[\bu]$, which is further contained in the ideal $\langle f_1,\ldots, f_n \rangle\subset \QQ[\bu,\bx_{\neq k}]$. The following theorem, which is a consequence of Theorem~4.28 in~\cite{D2013} (see also~\cite[pp.~6]{D2013}), gives bounds on the degree and height of the polynomials in the cofactor-representation of $\res(\FFF)$ in this ideal.
\begin{theorem}[\cite{D2013} Consequence of Theorem~4.28]\label{thm:dandrea}
	Given a polynomial system $\F=(\f_1,\ldots,\f_{n+1})$ with $\f_i= \sum_{\alpha}u_{i,\alpha}\bx^\alpha\in\ZZ[\bu,\bx]$ of total degree $d_i$ in $\bx=(x_1,\ldots,x_n)$. Then, for any $k\in [n]$, there exists a $\lambda\in\ZZnz$ and polynomials $\g_i\in\ZZ[\bu,\bx]$ such that
	\begin{align*}
		\lambda\cdot \res(\F)                    & = \sum_{i\in[n+1]} \g_i\cdot \f_i,                                         \\
		\deg_{\bu_{j}}(\g_i\f_i)                 & \le \prod_{\ell\neq j}d_\ell \;\text{ and}\;                               \\
		\tau_{\lambda \RES{\F}{x_k}},\tau_{\g_i} & \le (6n+10)\log(n+3)D_\F \quad \text{ for }j\in[n] \text{ and } i\in[n+1],
	\end{align*}
	where $\tau_p$ denotes the bit-size of a polynomial $p\in\ZZ[\bu,\bx]$.
\end{theorem}
We can now derive bounds on the degree and the bit-sizes of the polynomials $g_i$ as well as on the bit-size of $\lambda$ in (\ref{idealrep}) from the above theorem:
\begin{corollary}\label{cor:arithmetichilbert}
	Given a zero-dimensional polynomial system $\FFF=(f_1,\ldots,f_n)$ with polynomials $f_i\in\mathbb{C}[\bx]$, we can explicitly compute (see (\ref{AFF}) and (\ref{BFF})) positive integers $A_\FFF$ and $B_{\FFF}$, with $A_\FFF=\tilde{O}\big(n D_\mathcal{F}\big)$ and $B_\FFF=\tilde{O}(n\cdot D_{\FFF}+ \tau_\FFF\cdot  \max_i\frac{D_\FFF}{d_i})$, such that there exists an integer $\lambda\in\ZZnz$ and polynomials $g_{i}\in \CC[\bx]$ with
	\begin{align*}
		|\lambda|                    & \le 2^{A_\FFF},                                       \\
		\deg g_i                     & \le D_\FFF,\text{ }\tau_{g_i}\le B_{\FFF},\text{ and} \\
		\lambda\cdot \res(\FFF, x_k) & = \sum_{i=1}^n g_{i}\cdot f_i.
	\end{align*}
	If all polynomials $f_i$ have only integer coefficients, then we may further assume that the polynomials $g_{i}$ have only integers coefficients as well.
\end{corollary}
\begin{proof}
	For each $i\in [n]$, write $f_i(\bx) = \sum_{\alpha}u_{i,\alpha}\bx_{\neq k}^\alpha$ as a polynomial in the variables $\bx_{\neq k}$ and with coefficients $u_{i,\alpha}\in\mathbb{C}[x_k]$. Theorem~\ref{thm:dandrea} now guarantees the existence of a positive $\lambda\in\ZZnz$ and polynomials $g_i\in\ZZ[\bu,\bx_{\neq k}]$
	with $\lambda\cdot \RES{\FFF}{x_k} = \sum_{i\in[n]} g_i\cdot f_i$. Notice that since $\bu$ only depends on $x_k$, we may consider each $g_i$ as an element in $\mathbb{C}[\bx]$. In addition, we have $\deg_{\bu_{j}}(g_i f_i)\le \prod_{\ell\neq j}d_\ell$, and thus $\deg_{\bx}(g_i)\le\deg_{\bx}(g_i f_i)\le D_{\FFF}=d_1\cdots d_n$ as each $u_{j,\alpha}$ has degree bounded by $d_j$. 
	We can now write each polynomial $g_i$ as $g_{i}=\sum_{\alpha:|\alpha|\le D_\FFF}P_{i,\alpha}(\bu)\cdot \bx_{\neq k}^\alpha$ with polynomials $P_{i,\alpha}=\sum_{\beta=(\beta_1,\ldots,\beta_N):|\beta|\le D_\FFF/d_i}c_{i,\alpha,\beta}\cdot \bu^\beta$. From Theorem~\ref{thm:dandrea}, we conclude that $c_{i,\alpha,\beta}$ are integers of absolute value $|c_{i,\alpha,\beta}|<2^{A_\FFF}$, where
	\begin{align}\label{AFF}
		A_\FFF:=\lceil (6n+4)\log(n+2)D_{\FFF}\rceil= O\big(n D_\mathcal{F} \log(n)\big).
	\end{align}
In addition, $N\le\sum_{i=1}^n\binom{d_i+n}{d_i}\le n\cdot \binom{d_\FFF+n}{d_\FFF}\le n(d_\FFF+n)^{n}$  denotes the number of distinct coefficients $u_{i,\alpha}$. Further notice that, for each $\beta$, $\bu^\beta$ is a product of at most $D_\FFF$ univariate polynomials in $\CC[x_k]$, each of degree at most $d_\FFF$ and of norm bounded by $2^{\tau_\FFF}$. Hence, it can be written as a sum of at most $(d_{\FFF}+1)^{D_\FFF}$ terms, each of absolute value at most $2^{\tau_{\FFF}\cdot D_\FFF/d_i}$. We conclude that the norm of $g_i$ is bounded by 
	\begin{align*}\binom{D_\FFF+N}{D_\FFF}\cdot(d_{\FFF}+1)^{D_\FFF}\cdot  2^{A_\FFF}\cdot2^{\tau_{\FFF}\cdot D_\FFF/d_i}&\le [(n(d_\FFF+n)^n+D_\FFF)\cdot (d_\FFF+1)]^{D_\FFF}\cdot 2^{A_\FFF}\cdot2^{\tau_{\FFF}\cdot D_\FFF/d_i}\\
	&\le [(n+1)(d_\FFF+n)^{n+1}]^{D_\FFF}\cdot 2^{A_\FFF}\cdot2^{\tau_{\FFF}\cdot D_\FFF/d_i}\\
	&\le [(d_\FFF+n)^{6n+4}\cdot 2^{A_\FFF}\cdot2^{\tau_{\FFF}\cdot D_\FFF/d_i}\le 2^{B_\FFF},
	\end{align*} 
	where we define
	\begin{align}\label{BFF}
		B_\FFF:=2\cdot\lceil D_\FFF\cdot (6n+4)\log(d_\FFF+n) \rceil+\tau_\FFF\cdot \max_i\frac{D_\FFF}{d_i}=\tilde{O}(n\cdot D_{\FFF}+ \tau_\FFF\cdot  \max_i\frac{D_\FFF}{d_i}).
	\end{align}
The final claim follows from the fact that $f_i\in\mathbb{Z}[\bx]$ for all $i$ implies that $u_{i,\alpha}\in\ZZ[x_k]$ for all $i,\alpha$, and thus $g_j\in\ZZ[\bx]$ for all $j$.
\end{proof}

\subsection{Generic Position via Rotation}\label{subsec:genericposition}

In the previous subsection, we have outlined how to project the solutions of a polynomial onto one of the coordinate axis. One subtlety of the approach was that certain mild conditions on the input polynomials need to be fulfilled in order to guarantee that the roots of the hidden variable resultant are exactly the projections of the (finite) solutions of the initial system; see Lemma~\ref{lem:mildconditions}. Another drawback of the approach is that distinct solutions might be projected onto the same point or onto two very nearby points on the coordinate axis, that is, the actual distance between distinct solutions is no longer preserved after the projection. We will show how to address these issues by using a random rotation of the coordinate system. We first start with the special case of dimension $2$.

\begin{lemma}\label{genericrotation}
	Let $p_{\ell}=(x_\ell,y_\ell)\in\CC^2$ be $N$ points such that $\|p_{\ell}\|\neq 0$ for all $\ell=1,\ldots,N$.
	Let $k$ be chosen uniformly at random from $[2^{L}]$. Then, with probability at least $1-\frac{N}{2^L}$, for each point
	\[
		p_\ell'=\left(\begin{matrix} x_\ell' \\ y_\ell'\end{matrix}\right):=S_k(L)\cdot \left(\begin{matrix} x_\ell \\ y_\ell\end{matrix}\right), \quad\text{with}\quad S_k(L):=\left(
		\begin{matrix}
			\frac{1-(k\cdot 2^{-L})^2}{1+(k\cdot 2^{-L})^2}    & -\frac{2\cdot (k\cdot 2^{-L})}{1+(k\cdot 2^{-L})^2} \\
			\frac{2\cdot (k\cdot 2^{-L})}{1+(k\cdot 2^{-L})^2} & \frac{1-(k\cdot 2^{-L})^2}{1+(k\cdot 2^{-L})^2}
		\end{matrix}
		\right)\in\operatorname{SO}(2),
	\]
	it holds that $\min(|x_\ell'|, |y_\ell'|)>2^{-(L+2)}\cdot \|p_{\ell}\|$ for all $\ell$.
\end{lemma}

\begin{proof}
	Notice that each matrix $S_k(L)$ is a rotation matrix with respect to the angle $\phi_k\in [0,\pi/2]$ with $\cos\phi_k=\frac{1-(k\cdot 2^{-L})^2}{1+(k\cdot 2^{-L})^2}$ and $\sin\phi_k=\frac{2\cdot (k\cdot 2^{-L})}{1+(k\cdot 2^{-L})^2}$. We further note that the function $h(t)=(\tfrac{1-t^2}{1+t^2}, \tfrac{2t}{1+t^2})$  describes the trace of a point on the quarter-circle.
	Moreover, we have $\dot h(t)=(\tfrac{-4t}{(1+t^2)^2}, \tfrac{-2t^2+2}{(1+t^2)^2})$, and since $|\dot h(t)| = \tfrac{2}{1+t^2}$ is a decreasing function in $t$, it follows that the difference between two consecutive angles $\phi_{k+1}$ and $\phi_{k}$ is decreasing in $k$.
	We thus conclude that all differences are lower bounded by $\phi_{2^L} - \phi_{2^{L}-1} = \int_{1-2^{-L}}^{1}\tfrac{2}{1+t^2}dt\ge \int_{1-2^{-L}}^{1}dt= 2^{-L}$. Now, let $L_k\subset\RR^2$ be the line passing through the origin and the point $(\cos\phi_k,\sin\phi_k)$, and let $L_k^{\perp}\subset\RR^2$ be line that passes through the origin and is orthogonal to $L_k$. In addition, for each point $p_\ell=(\Re(x_\ell)+\mathbf{i}\cdot\Im(x_\ell),\Re(y_\ell)+\mathbf{i}\cdot\Im(y_\ell))$, we define
	\[
		\bar{p}_\ell=\begin{cases} (\Re(x_\ell),\Re(y_\ell)) &\mbox{if } \Re(x_\ell)^2+\Re(y_\ell)^2\ge  \Im(x_\ell)^2+\Im(y_\ell)^2   \\
		(\Im(x_\ell),\Im(y_\ell)) & \text{otherwise}. \end{cases}.
	\]
	Then, $\bar{p}_\ell$ is a point in $\RR^2$ with $\|\bar p_\ell\|_2\ge \|p_\ell\|/\sqrt{2}$. Let $\Delta_\ell\subset\RR^2$ be the disc centered at $\bar{p}_\ell$ of radius $r_\ell= 2^{-L-2}\cdot \|p_\ell\|$. Let $q,r\in\Delta_\ell$ be any two points in $\Delta_\ell$ and $\alpha$ be the angle at the origin of the triangle given by the origin and the points $q$ and $r$. Then, it holds that
	\[
		\alpha\le 2\cdot\arctan \left(\frac{2^{-L-2}\cdot \|p_\ell\|}{\|p_\ell\|/\sqrt{2}}\right)< 2\arctan(2^{-L-1})< 2\cdot 2^{-L-1}\le 2^{-L}.
	\]
	Since the angle between any two distinct lines $L_k$ and $L_{k'}$ is lower bounded by $2^{-L}$, it thus follows that there can be at most one $k$ such that $L_k$ or $L_k^\perp$ intersects $\Delta_\ell$.
	Hence, if we pick a $k\in\{1,\ldots, 2^L\}$ uniformly at random and choose $L_k$ and $L_k^\perp$ as the axis of the coordinate system obtained by rotating the initial system by $\phi_k$, then, with probability at least $1-\frac{N}{2^L}$, the new coordinates $(\bar{x}_\ell',\bar{y}_\ell')$ of each point $\bar{p}_\ell$ will meet the condition that $\min(|\bar{x}_\ell'|,|\bar{y}_\ell'|)>2^{-L-2}\cdot \|p_\ell\|$. Hence, the same holds true for the points $S_k(L)\cdot p_\ell$.
\end{proof}
We now turn to the general $n$-dimensional case. For integers $k$ and $L$ and distinct indices $i,j\in\{1,\ldots,n\}$, we define
\begin{align}\label{rotationmatrix}
	S_{k}^{[ij]}(L):=
	\left(
	\begin{matrix}
	1      & \cdots & 0                                                  & \cdots & 0                                                   & \cdots & 0      \\
	\vdots & \vdots & \vdots                                             & \vdots & \vdots                                              & \vdots & \vdots \\
	0      & \cdots & \frac{1-(k\cdot 2^{-L})^2}{1+(k\cdot 2^{-L})^2}    & \cdots & -\frac{2\cdot (k\cdot 2^{-L})}{1+(k\cdot 2^{-L})^2} & \cdots & 0      \\
	\vdots & \vdots & \vdots                                             & \vdots & \vdots                                              & \vdots & \vdots \\
	0      & \cdots & \frac{2\cdot (k\cdot 2^{-L})}{1+(k\cdot 2^{-L})^2} & \cdots & \frac{1-(k\cdot 2^{-L})^2}{1+(k\cdot 2^{-L})^2}     & \cdots & 0      \\
	\vdots & \vdots & \vdots                                             & \vdots & \vdots                                              & \vdots & \vdots \\
	0      & \cdots & 0                                                  & \cdots & 0                                                   & \cdots & 1
	\end{matrix}
	\right)\in\operatorname{SO}(n),
\end{align}
to be a rotation matrix that operates on the $i$-th and $j$-th coordinate only. We further define the set of rotation matrices
\begin{align}\label{formula:rot_matrices}
	\mathcal{S}_N:=\left\{\prod_{i,j\in[n]^2:i<j} S_{k_{ij}}^{[ij]}(L) : k_{ij}\in[2^{L}]\text{ for all }i,j\right\}, \text{ where } L:= 4\lceil \log (2n^2 N)\rceil.
\end{align}

\begin{lemma}\label{rotationgeneral}
	Let $N$ be a positive integer and $\bp_{\ell}\in\CC^n$ be $N'$, with $N'\le N$, points such that $\|\bp_{\ell}\|\neq 0$ for all $\ell=1,\ldots,N'$. $\mathcal{S}_N$ and $L$ are defined as in (\ref{formula:rot_matrices}). Then, it holds 
	\begin{itemize}
	\item[(a)] Choosing integers $k_{ij}\in[2^{L}]$ for every pair $i,j$ uniformly at random yields, with probability at least $3/4$, a rotation matrix
	$S\in \SSS_N$ such that,
	for each point $\bp_{\ell}':=S(L)\cdot \bp_{\ell}$, it holds that  $\min_i |p_{\ell,i}'|\ge(2 n^2N )^{-16 n}\cdot\|\bp_\ell\|$.
	\item[(b)]
	There is an integer $\lambda$ of bit-size $\tilde O( n^2 \log N)$ such that the entries of $\lambda S$ and $\lambda S^{-1}$ are integer numbers of bit-size $\tilde O( n^2 \log N)$ as well.
	\end{itemize}
\end{lemma}
\begin{proof}
	The proof follows almost immediately from Lemma~\ref{genericrotation}. Namely, with probability at least $1-N/2^L$, both entries $p_{\ell,i}'$ and $p_{\ell,j}'$ of each point $\bp_{\ell}':=S^{[ij]}_{k_{ij}}(L)\cdot \bp_{\ell}$ will have absolute value at
	least $2^{-(L+2)}\cdot\max(|p_{\ell,i}|,|p_{\ell,j}|)$. Since at least one of the coordinates of $\bp_{\ell}$ has absolute value $\|\bp_{\ell}\|$, we conclude that, with probability $(1-N/2^{L})^{\binom{n}{2}}>(1-N/2^{L})^{\frac{n^2}{2}}>1-\frac{n^2/2}{2^{L}/N}>3/4$,
	each coordinate of each point $\bp_{\ell}'=S(L)\cdot \bp_{\ell}$ has absolute value at least
	\[
		2^{-n\cdot(L+2)}\cdot \|\bp_{\ell}\| \ge 2^{-n\cdot(4(\log(2n^2N) + 1) + 2)}\cdot \|\bp_{\ell}\| \ge (2^{16 \log(2n^2 N)})^{-n}\cdot \|\bp_{\ell}\| =  (2 n^2N )^{-16 n}\cdot \|\bp_{\ell}\|.
	\]
	It remains to show the existence of an integer $\lambda$ of bit-size $\tilde O( n^2 \log N)$ such that the entries of $\lambda S$ and $\lambda S^{-1}$ are of that bit-size as well. 
    Each entry of a matrix $S_{k_{ij}}^{[ij]}(L)$ is rational number with denominator $2^{2L} + k_{ij}^2$ of bit-size $O(L)$. The matrix $S$ is a product of $O(n^2)$ many such matrices, thus for $\lambda = \prod_{i,j\in[n]^2:i<j} (2^{2L} + k_{ij}^2) \le (2^{2L+1})^{n^2} = 2^{\tilde O(n^2\log N)}$ it holds that $\lambda S$ is integer. 
    Notice that $S$ is contained in $\operatorname{SO}(n)$, which implies that its entries have absolute value at most 1. It  thus follows that the integer entries of $\lambda S$ are of bit-size $\tilde O(n^2\log N)$ as well.
	In addition, the inverse of $S_{k_{ij}}^{[ij]}(L)$ is simply given by $S^{[ij]}_{-k_{ij}}(L)$, and thus $S^{-1}=\prod_{i,j\in[n]^2:i<j} S_{-k_{ij}}^{[ij]}(L)$, which yields comparable bounds for the entries of $S^{-1}$ as for $S$.
\end{proof}
We will later make use of the above result when considering the set of non-zero solutions of a polynomial system $\mathcal{F}=0$. In general, some of these solutions might project (via resultant computation with respect to some variable $x_k$) onto zero or onto values close to zero. However, in our algorithm, we are aiming for projections that are of comparable size as the size of the corresponding solutions. In order to achieve this, we first consider a random rotation of the system given by some rotation matrix $S$ from the set $\SSS_N$, with $N:=D_\FFF$ the B\'ezout bound on the total number of solutions. This yields the ``rotated system'' $\FFF':=\FFF\circ S^{-1}$ whose solutions are exactly the rotations of the initial solutions by means of the rotation matrix $S$. Then, with high probability, each of the coordinates of the solutions of $\FFF'=0$ are of absolute value comparable to the norm of the solutions of $\FFF=0$. In addition, it is also likely that the rotated system fulfills the condition from Lemma~\ref{lem:mildconditions} for each coordinate.

\begin{lemma}\label{lem:rotationsystem}
	Let $\FFF=(f_1,\ldots, f_n)$ be a polynomial system as in~\eqref{formula:original_system}, $S\in \SSS_{D_\FFF}$ be a randomly chosen matrix, and let
	$\FFF':=\FFF\circ S^{-1}$ be the corresponding rotated system. Then, with probability larger than $1/2$, it holds:
	\begin{enumerate}[(a)]
		\item For each $k\in[n]$, each of the polynomials $f_i\circ S^{-1}\in\mathcal{F}'$ contains a monomial of degree $d_i$ that does not depend on $x_k$. \label{lemmarot+a}
		\item For each solution $\bz\in\CC^n\backslash 0$ of $\FFF=0$, it holds that $\min_i |z'_i|\ge (2n^2D_\FFF )^{-16n}\cdot \|\bz\|$, where $\bz'=(z'_1,\ldots,z'_n):=S\cdot \bz$ is the corresponding (rotated) solution of $\FFF'=0$.\label{lemmarot+b}
	\end{enumerate}
\end{lemma}

\begin{proof}
	Since $D_\FFF$ constitutes an upper bound on the number of solutions of $\FFF=0$, it follows from Lemma~\ref{rotationgeneral} (with $N=\mathcal{D}_\FFF$) that, with probability at least $3/4$, the inequality in (b) is fulfilled. It thus suffices to prove that, with probability larger than $2/3$, the condition in (a) is fulfilled for each coordinate $x_k$. For this, let
	\[
		f(\bx)=\sum_{\alpha}c_{\alpha}\bx^{\alpha}\in\CC[x_1,\ldots,x_n]
	\]
	be a polynomial of total degree $d$, and let $S(\mathbf{k})^{-1}=\left(a_{rs}(\mathbf{k})\right)_{rs}$ be the matrix depending on the values $\mathbf{k}:=(k_{ij})_{i,j}$. Notice that each entry $a_{rs}(\mathbf{k})$ is a rational function in $\mathbf{k}$ with numerators and denominators of total degree (in $\mathbf{k}$) at most $2n^2$.
	Further notice that $S(\mathbf{k})^{-1}$ maps the point $(1,0,\ldots,0)$ to the first column of $S(\mathbf{k})^{-1}$ and that a full-dimensional subset $T$ of the strictly positive part $\{\bx\in\RR^n:\|\bx\|_2=1\text{ and }\bx>0\}$ of the $(n-1)$-dimensional sphere $\mathbb{S}^{n-1}\subset \mathbb{R}^n$ is reached via a suitable choice of $\mathbf{k}\in\mathbb{R}^{n^2}$.
	Composing $f$ and $S(\mathbf{k})^{-1}$ now yields
	\[
		F(\bx,\mathbf{k})=\sum_{\alpha=(\alpha_1,\ldots,\alpha_n)}c_{\alpha}\cdot (a_{11}(\mathbf{k})\cdot x_1+\cdots+a_{1n}(\mathbf{k})\cdot x_n)^{\alpha_1}\cdots (a_{n1}(\mathbf{k})\cdot x_1+\cdots+a_{nn}(\mathbf{k})\cdot x_n)^{\alpha_n},
	\]
	and the coefficient $C(\mathbf{k})$ of the monomial $x_1^d$ is thus given by
	\[
		C(\mathbf{k})=\sum_{\alpha:|\alpha|=d} c_{\alpha}\cdot a_{11}(\mathbf{k})^{\alpha_1}\cdots a_{n1}(\mathbf{k})^{\alpha_n}.
	\]
	We first argue that $C(\mathbf{k})$ does not vanish identically. Let $\hat{f}:=\sum_{\alpha:|\alpha|=d} c_{\alpha}\cdot x_1^{\alpha_1}\cdots x_n^{\alpha_n}$ be the corresponding homogenous polynomial of degree $d$ such that $\hat{f}(a_{11}(\mathbf{k}),\ldots,a_{n1}(\mathbf{k}))=C(\mathbf{k})$. Assume that $C(\mathbf{k})=0$ for all $\mathbf{k}$, then this implies that $\hat{f}$ vanishes on each point in $T$.
	Since the vanishing set of any non-zero homogenous polynomial in $n$ variables has dimension at most $n-2$, we conclude that $\hat{f}$ is the zero-polynomial, and thus $c_{\alpha}=0$ for all coefficients of $\hat f$. This contradicts our assumption on $f$.

	Hence, it follows that $C(\mathbf{k})$ is a non-zero rational function in $\mathbf{k}$.  In addition, each term $c_{\alpha}\cdot a_{11}(\mathbf{k})^{\alpha_1}\cdots a_{n1}(\mathbf{k})^{\alpha_n}$ has a numerator of total degree at most $2n^2d$ in $\mathbf{k}$ and a denominator of the form $\prod_{i,j}(2^{2L}+k_{i,j}^2)^{e_{i,j}}$, with $e_{i,j}\in\mathbb{N}$, of degree at most $2n^2d$ in $\mathbf{k}$.
	This shows that $C(\mathbf{k})$ can be written as a rational function in $\mathbf{k}$ of total degree $2n^2d+ n^4d\le 2n^4d$ as $\prod_{i,j=1:i<j}^n (2^{2L}+k_{i,j}^2)^{n^2d}$ constitutes a common denominator of all terms.
	According to the Schwartz-Zippel lemma, we thus conclude that choosing $k_{i,j}$ uniformly at random from $\{1,\ldots,2^{4\lceil \log(2n^2D_\FFF)\rceil}\}$ guarantees with probability at least $\rho:=1-2n^4d\cdot 2^{-4\lceil \log(2n^2D_\FFF)\rceil}$ that $C(\mathbf{k})\neq 0$.
	In the case where $f=f_i$ is one of the polynomials from $\FFF$, we thus obtain a probability of at least
	\[
		\rho_i:=1-2n^4d_i\cdot 2^{-4\lceil \log(2n^2D_\FFF)\rceil}
		\ge 1-\frac{2n^4d_i}{(2n^2D_{\FFF})^4}\ge 1-\frac{1}{8n^4}
	\]
	such that $f_i$ contains a term of the form $c\cdot x_1^{d_i}$ with a non-zero constant $c$. Since the same argument applies to any variable $x_k$ and to any of the $n$ polynomials $f_i$, the claim follows.
\end{proof}

From the above lemma, we conclude that by choosing a suitably random rotation matrix from the set $\SSS_{D_\FFF}$, we can ensure with high probability that there is a one-to-one correspondence between the (finite) solutions of $\FFF$ and the roots of the resultant polynomial $\res(\FFF,x_k)$, which are the projections of the solutions on the $x_k$-axis. In addition, the absolute value of each projection compares well to the absolute value of the corresponding solution. In what follows, we will use the following definition of the set of \emph{admissible rotation matrices with respect to a given system $\FFF$}, i.e., matrices $S\in\SSS_{D_\FFF}$ such that the statements \eqref{lemmarot+a} and \eqref{lemmarot+b} from the above Lemma~\ref{lem:rotationsystem} hold.

\begin{definition}\label{goodmatrices}\emph{(Admissible Matrices)}
	For a given polynomial system $\FFF$ we say that a rotation matrix $S\in \SSS_{D_\FFF}$ is \emph{admissible with respect to $\FFF$} if the statements \eqref{lemmarot+a} and \eqref{lemmarot+b} from Lemma~\ref{lem:rotationsystem} hold.
	We further denote by
	\[
		\SSS_\FFF:=\{S\in\SSS_{D_\FFF}:S\text{ is admissible with respect to }\FFF\}\subset\SSS_{D_\FFF}
	\]
	the set of \emph{admissible matrices with respect to $\FFF$}.
\end{definition}

Notice that, even though it is difficult (probably as difficult as computing all solutions of $\FFF$) to determine whether a certain matrix in $\SSS_{D_\FFF}$ is admissible with respect to $\FFF$, the previous lemma shows that at least half of the matrices in $\SSS_{D_\FFF}$ are admissible.

\providecommand{\ubg}{\gamma}

\section{The Algorithm}\label{sec:algorithm}
We first sketch our algorithm $\cert$ and then prove its correctness. We refer the reader to the pseudo-code in Algorithm~\ref{algo:truncate} for details regarding $\cert$. The algorithm can be roughly split into 3 main steps:\\

\begin{algorithm}[t!]
\BlankLine
 \Input{Zero-dimensional system $\FFF=(f_1,\ldots, f_n)$, polydisc $\mathbf{\Delta}=\mathbf{\Delta}_r(\bm)$, and an integer $K\in\{0,\ldots,d_\FFF\}$.}
 \Output{An integer $k\in\NN\cup\{-1\}$. If $k\ge 0$, the polydisc $\mathbf{\Delta}$ contains exactly $k$ solutions of $\FFF$ (counted with multiplicity). If $k=-1$, nothing can be said.}
 \BlankLine

 \nonl\texttt{//} \textit{* Shift and Truncation *}\texttt{//}\;
 $L:=\lceil\log\frac{r}{32n(K+1)^n}\rceil$\;
 Compute a $(K+1)\cdot L$-bit approximation
 $\Phi'
 =(\phi_1',\ldots,\phi_n')
 $ of $\FFF[\bm]_{\le K}=(f_1[\bm]_{\le K}, \ldots, f_n[\bm]_{\le K}). $	
 \BlankLine
 \medskip\nonl\texttt{//} \textit{* Adding a degree $(K+1)$-perturbation ;  as mentioned, this step seems to be only necessary in theory. In practice, we recommend to directly proceed with $\Phi:=\Phi'$. *}\texttt{//}\;
 $\Phi(\bx)
 :=(\phi_1,\ldots,\phi_n), \text{ with }\phi_i:= x_i^{K+1}+\phi'_i$

 \medskip\nonl\texttt{//} \textit{* Solving the truncated system *}\texttt{//}\;
Compute a list $(\mathbf{\Delta}_1,k_1),\ldots, (\mathbf{\Delta}_\ell, k_\ell)$ of disjoint polydiscs $\mathbf{\Delta}_i=\mathbf{\Delta}_{r_i}(\bm_i)$ of radius at most $2^{-L}$ and corresponding multiplicities $k_i$ such that each $\mathbf{\Delta}_i$ contains exactly $k_i$ solutions of $\Phi$, and each solution of $\Phi$ is contained within one $\mathbf{\Delta}_i$.\;\label{line:oracle}
 $k:=\sum_{i:\|\bm_i\| < \frac{r}{2n}} k_i$\;
 $k^+:=\sum_{i:\|\bm_i\| < 2nr} k_i$\;
 \medskip

\If{$k=k^+$  \label{line:ifrootscontained}}{
  \medskip\nonl\texttt{//} \textit{* Projection step *}\texttt{//}\;
  Pick $S\in\SSS_{D_\Phi}$ uniformly at random and compute the rotated system
  $\Phi^*:=\Phi\circ S^{-1}=(\phi_i\circ S^{-1})_i$.\;
  \For{$\ell=1,\ldots, n$}{\nonl
   $(b_\ell^-, k^-_\ell, \LB_\ell^-) := \TTT_*(\Delta_{\frac{r}{\sqrt{n}}}(0), \res(\Phi^*, x_\ell))$\nonl\\
   $(b_\ell^+, k^+_\ell ,\LB_\ell^+) := \TTT_*(\Delta_{\sqrt{n}r}(0), \res(\Phi^*, x_\ell))$.\;
  }
  \If{$\bigwedge_{\ell\in[n]}b_\ell^-\wedge \bigwedge_{\ell\in[n]}b_\ell^+$ \label{line:ifTktest}}{
   \medskip\nonl\texttt{//} \textit{* Bound Computation and Comparison *}\texttt{//}\;
   $\begin{array}{ll}
    \UB(\bm,r) & :=r^{K+1} \cdot d_\FFF^n 2^{\tau_\FFF + 2} [M(\bm)\cdot (n+d_\FFF)^2]^{d_\FFF}\\  
    B_{\Phi^*}   & :=2\cdot\lceil D_{\Phi^*}\cdot (6n+4)\cdot\log(d_{\Phi^*}+n) \rceil + \tau_{\Phi^*}\cdot \frac{D_{\Phi^*}}{k+1}\\
    \LB(\bm,r) & :=\min_\ell\min(\LB_\ell^-,\LB_\ell^+)\cdot \left(n\cdot\binom{D_{\Phi^*} + n}{D_{\Phi^*}}\cdot 2^{B_{\Phi^*}}\right)^{-1}
\end{array}$\label{line:lowerbound}
   
   \If{$\UB(\bm, r)\le \LB(\bm, r)$}{
    \Return{$k$}
   }
  }
 }
 \Return{$-1$}
    \BlankLine
 \caption{$\cert(\FFF,\mathbf{\Delta},K)$}
 \label{algo:truncate}
\end{algorithm}

\noindent\textbf{\emph{Step 1: Shifting and Truncation.}} Given a polynomial system $\FFF:=(f_1,\ldots,f_n)$, a polydisc $\mathbf{\Delta}=\mathbf{\Delta}_r(\bm)$, and an integer $K\in\{0,\ldots,d_\FFF\}$, we define a ``precision'' $$L:=\lceil\log\frac{r}{32n(K+1)^n}\rceil.$$ Then, in a first step, we compute a $(K+1)\cdot L$-bit approximation
\[
 \Phi'(\bx)
 :=(\phi_1',\ldots,\phi_n'), \text{ with }\phi_i'\in\QQ[\bx]_{\le K},
\]
of $\FFF[\bm]_{\le K}(\bx)$, i.e., we compute $\phi_i'$ such that $\|\phi_i'-f_i[\bm]_{\le K}\|<2^{-(K+1)L}$ and $2^{(K+1)L}\cdot \phi_i'\in\ZZ[\bx]$ for all $i$. Recall that the centered polynomial system $\FFF[\bm](\bx)$ was defined as
\[
 \FFF[\bm](\bx):=(f_{1}[\bm](\bx),\ldots,f_{n}[\bm](\bx))=(f_1(\bm+\bx),\ldots,f_n(\bm+\bx)),
\]
for $\bm\in\CC^n$ and that the truncation
\begin{align*}\label{truncsystem}
 \FFF[\bm]_{\le K}(\bx)
 :=(f_1[\bm]_{\le K},\ldots,f_n[\bm]_{\le K}),
\end{align*}
of the centered system $\FFF[\bm]$, is defined by simply omitting all terms of $f_{i}[\bm](\bx)$ of total degree more than $K$, as defined in Section~\ref{subsec:defnot}. We further define
\[
 \Phi(\bx)
 :=(\phi_1,\ldots,\phi_n), \text{ with }\phi_i:=x_i^{K+1}+\phi'_i,
\]
the system obtained by adding the term $x_i^{K+1}$ of degree $K+1$ to the polynomial $\phi_i'$. This step seems to be odd at first sight, however, it ensures certain properties of $\Phi$. In particular, $\Phi$ is guaranteed to have no zeros at infinity (and thus being zero-dimensional as well) according to Corollary~\ref{artificialperturbation} and our considerations in the corresponding example. This further implies that $\Phi=0$ has exactly $D_{\Phi}=(K+1)^n$ finite solutions counted with multiplicity. Also, our choice of $\Phi$ allows us to bound the leading coefficient of $\res(\Phi,x_\ell)$ for all $\ell=1,\ldots,n$, which turns out to be useful in the analysis of our approach.\medskip

\noindent\emph{Remark.} In practice, the latter step does not seem to be necessary in most cases, and thus we recommend to simply proceed with $\Phi:=\Phi'$ and to check $\Phi'$ for being zero-dimensional. Also, when implementing our algorithms, we observed that proceeding with $\Phi'$ instead of $\Phi$ only improves the overall performance.\\

\noindent\textbf{\emph{Step 2: Solving $\Phi$.}} We will later prove that, under the assumption that $L$ is sufficiently large (or equivalently $\mathbf{\Delta}$ is sufficiently small), and $\bz$ is a $k$-fold solution of the initial system with $\|\bm-\bz\|<2^{-L}$, the system $\Phi$ (as well as $\Phi'$ for generic choice of its coefficients) yields a cluster of $k$ (not necessarily distinct) solutions with norm less than $4\cdot 2^{-L}$, whereas all other solutions have norm larger than $\delta_0\gg 2^{-L}$. Here, $\delta_0$ is a constant that depends on the polynomial system but not on $L$; see~Theorem~\ref{thm:countroots} for the exact definition of $\delta_0$ and further details.
We first check whether there exists a cluster of solutions of $\Phi$ near the origin that is well separated from all other solutions of $\Phi$. For this, we use a certified method (e.g.~\cite{DBLP:conf/issac/BrandS16}) to compute all solutions of $\Phi$. Here, by computing all solutions, it is meant to compute a set of disjoint discs, each of size less than $2^{-L}$, together with the number of solutions contained in each disc such that the union of all discs contains all complex solutions. For the more involved problem of computing isolating regions of comparable size, the following theorem applies.

\begin{theorem}\cite[Thm.~9, 10]{DBLP:conf/issac/BrandS16}\label{cost:solving}
There is a Las Vegas algorithm to compute isolating regions of size less than $2^{-\rho}$ for all complex solutions of a zero-dimensional polynomial system $\FFF=(f_1,\ldots, f_n)$, with integer polynomials $f_i\in\ZZ[\bx]$, using 
\[
    \tilde{O}(n^{(n-1)(\omega+1)+1}(nd_\FFF+\tau_\FFF){d_\FFF}^{(\omega+2)n-\omega-1}+n\cdot {d_\FFF}^n\cdot \rho)
\]
bit operations in expectation. 
\end{theorem}

Since $2^{(K+1)L}\cdot \phi_i$ is a polynomial of degree $K+1$ with integer coefficients of magnitude $\tau=O(KL+n+\tau_\FFF+d_\FFF\LOG(\bm))$, we conclude from the above theorem that the cost for solving the system $\Phi=0$ is bounded by 
\begin{align}\label{cost_solving_truncated_system}
\tilde{O}(n^{(n-1)(\omega+1)+1}(nK+KL+\tau+d_\FFF\LOG(\bm))\cdot (K+1)^{(\omega+2)n-\omega-1})
\end{align}
bit operations in expectation. Finally, we check whether the polydisc $\mathbf{\Delta^-}:=\mathbf{\Delta}_{r/(2n)}(0)$ contains the same number $k'$ of solutions of $\Phi$ as the enlarged polydisc $\mathbf{\Delta^+}:=\mathbf{\Delta}_{2nr}(0)$.
Notice that, from the above remark, this holds true if $\bm$ is an $L$-bit approximation of a $k$-fold zero of $\FFF$ for large enough $L$ as then $4\cdot 2^{-L}<r/(2n)$ and $\delta_0>2nr$.
If the zeros of $\Phi$ do not fulfill the latter condition, we return $-1$. Otherwise, we proceed.\medskip 

\noindent\emph{Remark.} We remark that computing the solutions of $\Phi$ is typically much more affordable than computing the solutions of the initial system $\FFF$ directly, in particular, in the case where $n$ is small and $K\ll d$. Notice that, for $n$ of constant size, the cost for solving the initial system directly scales like $d^{(\omega+2)n-\omega-1}\tau_\FFF$, whereas the cost for solving the truncated system scales like $(KL+\tau_\FFF+d\LOG(\bm))\cdot (K+1)^{(\omega+2)n-\omega-1}$. Hence, for $L$ and $m$ of moderate size, the running times might differ by factor of size $\approx (d/(K+1))^{(\omega+2)n-\omega-1}$.\\ 


\noindent\textbf{\emph{Step 3: Passing from $\Phi$ to $\FFF$.}}
In the final step, we aim to certify that $\FFF[\bm]$ has the same number of zeros  (i.e.~$k$ counted with multiplicity) in $\mathbf{\Delta}:=\mathbf{\Delta}_{r}(0)$ as $\Phi$. In order to do so, we aim to apply the following generalization of Rouch\'e's Theorem to $\Phi$ and $\FFF[\bm]$, see~\cite[Thm.~2.1]{DBLP:journals/tcs/VerscheldeH94} or~\cite[Thm.~1]{Lloyd75} for a proof.

\begin{theorem}[Multidimensional Rouch\'e]\label{Rouche}
    Let $\FFF=(f_1,\ldots,f_n)$ and $\GGG=(g_1,\ldots,g_n)$, with $f_i,g_i\in\CC[\bx]$ for all $i$, define polynomial mappings from $\CC^n$ to $\CC^n$. If, for a given bounded domain $\mathbf{D}\subset\CC^n$, we have
    \[
        \|\FFF(\bx)-\GGG(\bx)\|<\|\FFF(\bx)\|\quad\text{for all }\bx\in\partial \mathbf{D},
    \]
    where $\partial \mathbf{D}$ is the boundary of $\mathbf{D}$, then $\FFF$ and $\GGG$ have finitely many zeros in $\mathbf{D}$ and the number of zeros (counted with multiplicities) of $\FFF$ and $\GGG$ in $\mathbf{D}$ is the same.
\end{theorem}

In order to apply the above theorem to $\FFF:=\Phi$ and $\GGG:=\FFF[\bm]$, we derive an upper bound $\UB(\bm,r)$ on the absolute error
\begin{align*}
    \sup_{\bx\in\CC^n:\|\bx\|=r}\|\FFF[\bm](\bx)-\Phi(\bx)\|&=\sup_{\bx\in\CC^n:\|\bx\|=r}  \max_{i\in[n]} |f_i[\bm](\bx)-\phi_i(\bx)|\\
    &\le \sup_{\bx\in\CC^n:\|\bx\|=r}  \max_{i\in[n]} (|f_i[\bm](\bx)-\phi_i'(\bx)|+|x_i|^{K+1})\\
    &\le r^{K+1}+\sup_{\bx\in\CC^n:\|\bx\|=r}  \max_{i\in[n]} |f_i[\bm](\bx)-\phi_i'(\bx)|
\end{align*}
when passing from $\Phi$ to $\FFF[\bm]$ as well as a lower bound $\LB(\bm,r)$ on the norm of $\Phi(\bx)$ on the boundary of the polydisc $\mathbf{\Delta}$. The construction of $\UB(\bm,r)$ is rather straightforward using Lemma~\ref{lemma:truncated approximation}. That is, we may choose
\begin{align}\label{def:UB}
    \UB(\bm,r):=r^{K+1} \cdot {d_\FFF}^n\cdot 2^{\tau_\FFF + 2}\cdot [M(\bm)\cdot (n+d_\FFF)^2]^{d_\FFF} 
\end{align}

In contrast, the construction of $\LB(\bm, r)$ is more involved:
We already mentioned that if $L$ is large enough, then there are $k$ zeros $\bz_1,\ldots,\bz_k$ of $\Phi$ that have norm less than $4\cdot 2^{-L}$, whereas all other zeros have norm $\delta_0\gg 2^{-L}$.
Hence, under this assumption, picking\footnote{In practice, we recommend to consider $S=\operatorname{id}_n$ and thus $\Phi^*=\Phi$ as the initial choice as this turns out to be sufficient in most cases.} a random rotation matrix $S$ from $S_{D_\Phi}=S_{(K+1)^n}$ and considering a corresponding rotation of the coordinate system, guarantees (see Lemma~\ref{lem:rotationsystem}), with probability larger than $1/2$, that the
projection of any zero of the ``rotated system'' 
$$\Phi^*=(\phi_1^*,\ldots,\phi_n^*):=\Phi\circ S^{-1}$$
on any coordinate axis, except for the $k$ solutions $S\cdot \bz_i$, yields a value that is large compared to $r$. Hence, in this case, the hidden-variable resultant $R_\ell^*:=\operatorname{Res}(\Phi^*,x_\ell)$ of $\Phi^*$ has $k$ roots of absolute value less than $r$, whereas all other roots of $R_\ell^*$ have absolute value $\gg r$. 
Notice that each $\phi_j^*\in\QQ[\bx]$ is a polynomial of degree $K+1$ with rational coefficients, and 
according to Lemma~\ref{generalbounds} and Lemma~\ref{lem:rotation_coeffsize}, we have 
\[
    \|\phi_j^*\|=2^{\tilde{O}(n+\tau_\FFF+d_\FFF\LOG(\bm))}.
\]
Lemma~\ref{rotationgeneral} further yields the existence of an integer $\lambda$ of absolute value $2^{\tilde{O}(n^3\log K)}$ with $\lambda\cdot S^{-1}\in \ZZ^{n\times n}$. Hence, we conclude that each term of degree $K+1$ of $\lambda^{K+1}\cdot\phi^*_j$ has integer coefficients, and~\cite[Theorem 3.1]{cox2005using} further yields that
\[
\res(\lambda^{K+1}\cdot\Phi^*,x_\ell)=\lambda^{n(K+1)^{n}}\cdot\res(\Phi^*,x_\ell).
\]
It thus follows that
\[
|\operatorname{LC}(\RES{\Phi^*}{x_\ell})|= \lambda^{-n(K+1)^{n}}\cdot |\operatorname{LC}(\RES{\lambda^{K+1}\Phi^*}{x_\ell})|\ge \lambda^{-n(K+1)^{n}}=2^{-\tilde{O}((K+1)^n)},
\]
where we use Corollary~\ref{artificialperturbation} to show that $|\operatorname{LC}(\RES{\lambda^{K+1}\Phi^*}{x_\ell})|$ is a positive integer, hence larger than or equal to $1$.
Since $\RES{\Phi^*}{x_\ell}$ is contained in the ideal generated by the polynomials $\phi_j^*$, we may write
\begin{align}\label{cofactorrep}
 \RES{\Phi^*}{x_\ell}=g_{\ell,1}\cdot \phi_1^*+\cdots+g_{\ell,n}\cdot \phi_n^*
\end{align}
with polynomials $g_{\ell,j}\in\QQ[\bx]$ of total degree bounded by $D_{\Phi^*}=(K+1)^n$.
Corollary~\ref{cor:arithmetichilbert} further yields the following upper bound on the size of the coefficients of the $g_{\ell,j}$'s:
\begin{align*}
 \log\|g_{\ell,j}\| & \le B_{\Phi^*}
 =\tilde{O}(D_{\Phi^*}\cdot n+\tau_{\Phi^*}\cdot \frac{D_{\Phi^*}}{K+1}))
 =\tilde{O}((K+1)^n+(K+1)^{n-1}\cdot(\tau_\FFF+d\LOG(\bm))).
\end{align*}
Using Lemma~\ref{generalbounds}, part~\ref{parta}, this further yields a corresponding upper bound
\begin{align}\label{bound:B}
 \ubg:=\binom{n + D_{\Phi^*}}{D_{\Phi^*}}\cdot 2^{B_{\Phi^*}}=2^{\tilde{O}((K+1)^n+(K+1)^{n-1}\cdot(\tau_\FFF+d\LOG(\bm)))}
\end{align}
such that
$\max_{\ell,j}\sup_{\bx:\|\bx\|\le 1}|g_{\ell,j}(\bx)|\le \ubg$.\medskip

\noindent\emph{Remark.} The reader might wonder why we do not compute the above cofactor representation (\ref{cofactorrep}) directly and then derive bounds on the size of $\max_{i,j}\sup_{\bx:\|\bx\|=1}|g_{\ell,j}(\bx)|$ using interval arithmetic, but instead use Corollary~\ref{cor:arithmetichilbert}? The simple reason is that, at least in practice, computing the polynomials $g_{i,j}$ turns out to be considerably more costly than computing the resultant polynomials $R_\ell^*(x)=\RES{\Phi^*}{x_\ell}$ only. 
In contrast, our approach of computing the bound $\ubg$ does not require to compute the polynomials $g_{i,j}$, and thus comes at almost no additional cost. We further remark at this point that we will use the bounds from Corollary~\ref{cor:arithmetichilbert} in our complexity analysis of the algorithm.\medskip

In the next step, we compute lower bounds $\LB_\ell^-$ and $\LB_\ell^+$ for $|R_\ell(x)^*|$ on the boundary of the two discs $D^-:=\Delta_{r/\sqrt{n}}(0)\subset\CC$ and $D^+:=\Delta_{r\cdot\sqrt{n}}(0)\subset\CC$, respectively. For this, we use the so-called $\TTT_k$-test, an approach that has recently been proposed in an algorithm for complex root isolation~\cite{DBLP:journals/corr/BeckerS0Y15}. 

\begin{lemma}[\cite{DBLP:journals/corr/BeckerS0Y15}]\label{pellet}
 Let $f\in\CC[x]$ be a uni-variate polynomial of degree $d$ and let $\Delta:=\Delta_r(0)\subset\CC$ be the disc with radius $r$ centered at $0$. The so-called $T_k$-test returns a pair
 \begin{align}\label{Tktest}
  \mathcal{T}_{k}(\Delta,f)=(b, \LB)=
  \Big( \frac{|f^{(k)}(0)|r^{k}}{k!} -\frac{3}{2}\cdot\sum_{i\neq k} \frac{|f^{(i)}(0)|r^{i}}{i!} >0,
  \frac{1}{3}\cdot \frac{|f^{(k)}(0)|r^{k}}{k!}\Big).
 \end{align}
 If $b=$\textsc{True}, we say that $\mathcal{T}_{k}(\Delta,f)$ succeeds. If $\mathcal{T}_{k}(\Delta,f)$ succeeds, $\Delta$ contains exactly $k$ roots counted with multiplicity and $$5\cdot\LB>|f(x)|>\LB\quad\text{for all }x\in\partial\Delta_r(0).$$ In addition, if $\Delta_{r/(16d)}(0)$ as well as $\Delta_{16d^4r}(0)$ contain exactly $k$ roots, then $\TTT_k(\Delta,f)$ succeeds. We further define 
  \begin{align}\label{Tstartest}
  \mathcal{T}_\star(\Delta,f)=
  \begin{cases}
   \Big( \textsc{True},k,
  \frac{1}{3}\cdot \frac{|f^{(k)}(0)|r^{k}}{k!}\Big)&\text{ if }\mathcal{T}_k(\Delta,f) \text{ succeeds for some }k,\\
  (\textsc{False},-1)&\text{ otherwise.}
  \end{cases}
  \end{align}
\end{lemma}

Now, suppose that $\TTT_*(D^-,R_\ell^*)=(b_\ell^-, k^-_\ell, \LB_\ell^-)$ as well as $\TTT_*(D^+,R_\ell^*)=(b_\ell^+, k^+_\ell, \LB_\ell^+)$ succeed for all $\ell$, then $\min(\LB_\ell^-,\LB_\ell^+)$ constitutes a lower bound for $|R_\ell^*|$ on the boundary of $D^-$ as well as $D^+$.
From (\ref{cofactorrep}), (\ref{bound:B}), and the definition of $\LB(\bm,r)$,
we now conclude that
\begin{align}\label{def:LB}
 \|\Phi^*(\bx)\|>\LB(\bm,r):=\frac{\min_\ell\min( \LB_\ell^-,\LB_\ell^+)}{n\ubg},\text{ for all }\bx\text{ with }\|\bx\|=\frac{r}{\sqrt{n}}\text{ or }\|\bx\|=\sqrt{n}\cdot r.
\end{align}

Since the maximum and minimum of a holomorphic function (in several variables) on a bounded domain is taken at its boundary, we further conclude that the above inequality holds for any $\bx$ with $r/\sqrt{n}\le\|\bx\|\le \sqrt{n}\cdot r$. Notice that the rotation of the system by means of the rotation matrix maintains the $2$-norm $\|.\|_2$ of any point. Thus, the the norm of any point $\bx$ differs from the norm of the rotated point $S\cdot\bx$ by a factor that is lower and upper bounded by $r/\sqrt{n}$ and $\sqrt{n}\cdot r$, respectively. Hence, from the above bound on $\|\Phi^*\|$, we conclude that
\begin{align}\label{def:LB:unrotatetd}
    \|\Phi(\bx)\|>\LB(\bm,r)\quad\text{for any }\bx\text{ with }\|\bx\|=r.
\end{align}

Now in order to apply Rouch\'e's Theorem to $\Phi$ and $\FFF[\bm]$, it suffices to 
check whether $\LB(\bm, r)>\UB(\bm, r)$, in which case we have shown that $\Phi$ and $\FFF[\bm]$ have the same number of roots in $\mathbf{\Delta}_r(0)$. Hence, we return \textsc{True} in this case. Otherwise, the algorithm returns \textsc{False}.\\

In the next section, we will show that, if $\bm$ is a sufficiently good approximation (i.e.~for large enough $L$) of a $k$-fold solution of $\FFF$, our algorithm succeeds. Here, we only give an informal argument: Notice that, for large $L$, the bound $\UB(\bm,r)$ scales like $C\cdot r^{K+1}$ for some constant $C$. The bound $\ubg$ does not depend on $L$, hence $\LB(\bm,r)$ scales like $\min_\ell\min(\LB_\ell^-,\LB_\ell^+)$ for large enough $L$.
However, in this situation, each $R_\ell^*$ has a cluster of $k$ roots near the origin that is well separated from all of its remaining roots, and thus $\min(\LB_\ell^-,\LB_\ell^+)$ scales like $[|{R_\ell^*}^{(k)}(0)|/(\sqrt{n}^k k!)]\cdot r^k$. Hence, we conclude that $\LB(\bm,r)$ scales like $C'\cdot r^k$ for some constant $C'$, which implies that $\LB(\bm,r)$ must be smaller than $\UB(\bm,r)$ for large enough $L$. We remark that the precise argument is slightly more involved as many subtleties need to be addressed. 
In particular, we need to show that $|{R_\ell^*}^{(k)}(0)|/(\sqrt{n}^k k!)$ does not depend on $r$ if $r$ is small enough, even though the definition of $R^*$ strongly depends on the choice of $m$, $r$, and the rotation matrix $S$. We will give details in the next section.


\section{Analysis}\label{sec:analysis}

\providecommand{\s}{\sigma}
\providecommand{\del}{\partial}
\providecommand{\delt}{\delta}

We start by introducing some further notation. For a zero-dimensional polynomial system $\FFF=(f_1,\ldots, f_n)$ in $n$ variables, let $\bz_1,\ldots, \bz_N$ denote its zeros. 
We define 
\[
    \s(\bz_i,\FFF):=\min_{j\neq i}\|\bz_i - \bz_j\| \quad \text{and} \quad
    \del(\bz_i,\FFF) := \prod_{j\neq i} \|\bz_i-\bz_j\|^{\mu(\bz_j, \FFF)}
\]
%
to be the \emph{separation of $z_i$ with respect to $\FFF$} and  
the \emph{geometric derivative of $\FFF$ at $z_i$}, respectively.
We remark that these terms are derived from the interpretation of these quantities in the univariate case, where the separation of a root $z_0$ of a polynomial $f\in\CC[x]$ is defined in exactly the same way, and the first non-vanishing derivative
\[
    \left|\frac{\partial^{\mu(z_0,f)} f}{\partial z^{\mu(z_0,f)}}(z_0)\right|=\operatorname{LC}(f)\cdot\prod_{z\neq z_0:f(z)=0}|z-z_0|^{\mu(z,f)}
\] 
of $f$ at $z_0$ can be expressed as a product involving the leading coefficient of $f$ and the distances between $z_0$ and the other roots. We first provide some bounds on $\|\bz_i\|$, $\sigma(\bz_i,\FFF)$, and $\partial(\bz_i,\FFF)$ for the special case where each $f_i$ has only integer coefficients. For similar bounds that are also adaptive with respect to the sparseness of the given system, we refer to~\cite{DBLP:conf/issac/EmirisMT10}.

\begin{lemma}\label{boundsonsizeandseparation}
    Let $\FFF=(f_i)_{i=1,\ldots,n}$ be a zero-dimensional system with integer polynomials $f_i$, and let $\bz_1,\ldots,\bz_N$ denote the zeros of $\FFF$. Then it holds:
    \begin{align*}
        \sum_{i=1}^N\mu(\bz_i,\FFF)\cdot \LOG(\bz_i)&=\tilde{O}(n\cdot B_\FFF)=\tilde{O}(n\cdot [n\cdot D_\FFF+\tau_\FFF\cdot\max_i\frac{D_\FFF}{d_i}]),\text{ with }B_\FFF\text{ as in~(\ref{BFF})}\\
        |\log\sigma(\bz_i,\FFF)|&= \tilde{O}(D_\FFF\cdot B_\FFF),\\
        \LOG(\partial(\bz_i,\FFF)^{-1})&=\tilde{O}(n\cdot D_\FFF\cdot [B_\FFF+\LOG(\bz_i)])=\tilde{O}(n^2\cdot D_\FFF\cdot B_\FFF),\text{ and}\\
        \LOG(\sigma(\bz_i,\FFF)^{-1})&=\tilde{O}(D_\FFF\cdot\LOG(\bz_i)+n\cdot B_\FFF+\LOG(\partial(\bz_i,\FFF)^{-1})) = \tilde{O}(n^2\cdot D_\FFF\cdot B_\FFF).
    \end{align*}
\end{lemma}

\begin{proof}
    From Corollary~\ref{cor:arithmetichilbert}, we conclude that $\res(\FFF,x_\ell)$ is an integer polynomial of magnitude $(D_\FFF,B_\FFF)$ for all $\ell=1,\ldots,n$. 
    Since the $\ell$-th coordinate $z_{i,\ell}$ of each solution of $\FFF=0$ is a root of multiplicity at least $\mu(\bz_i,\FFF)$ of $\res(\FFF,x_\ell)$ and since the Mahler measure 
    \[
        \operatorname{Mea}(\res(\FFF,x_\ell))=\operatorname{LC}(\res(\FFF,x_\ell))\cdot\prod_{z\in\CC:\res(\FFF,x_\ell)(z)=0}M(z)^{\mu(z,\res(\FFF,x_\ell))}
    \] 
    of $\res(\FFF,x_\ell)$ is upper bounded by its $2$-norm $\|\res(\FFF,x_\ell))\|_2\le \sqrt{D_F}\cdot 2^{B_\FFF}$ (e.g. see~\cite{yap2000fundamental}), it follows that 
    \[
        \sum_{i=1}^N \mu(\bz_i,\FFF)\cdot\LOG(\bz_i)\le D_\FFF+\sum_{\ell=1}^n\log(\Mea(\res(\FFF,x_\ell)))
        =\tilde{O}(n B_\FFF).
    \]
    For the second claim, notice that $\sigma(\bz_i,\FFF)\ge\sigma(z_{i,\ell},\res(\FFF,x_\ell))$ for at least one $\ell$ (as two distinct solutions must differ in at least one coordinate), and that the separation of an integer polynomial of magnitude $(D_\FFF,B_\FFF)$ is lower bounded by $2^{-\tilde{O}(D_\FFF\cdot B_\FFF)}$; e.g. see~\cite{DBLP:journals/jsc/MehlhornSW15} for a proof. 
    For the bound on $\partial(\bz_i,\FFF)$, notice that 
    \[
        \partial(\bz_i,\FFF)
        \ge \frac{\prod_{\ell=1}^{n}\partial(z_{i,\ell},\res(\FFF,x_\ell))}{\prod_{\ell=1}^{n}\prod_{z\neq z_{i,\ell}:\res(\FFF,x_\ell)(z)=0}M(z-z_{i,\ell})^{\mu(z,\res(\FFF,x_\ell))}}.
    \]
    According to the proof of~\cite[Thm.~5]{DBLP:journals/jsc/MehlhornSW15}, it holds that $\partial(z_0,f)=2^{-\tilde{O}(dL)}$ for any root $z_0$ of a polynomial $f\in\ZZ[x]$ of magnitude $(d,L)$. This shows that $\partial(z_{i,\ell},\res(\FFF,x_\ell))=2^{-\tilde{O}(D_\FFF B_\FFF)}$ for all $\ell$. It remains to derive an upper bound on the denominator in the above fraction. 
    For this, we define $R_\ell:=\res(\FFF,x_\ell)[z_{i,\ell}]$. Then, it holds that
    \[
        \prod\nolimits_{\ell=1}^{n}\prod\nolimits_{z\neq z_{i,\ell}:\res(\FFF,x_\ell)(z)=0}\max(1,|z-z_{i,\ell}|)^{\mu(z,\res(\FFF,x_\ell))}
        =\prod_{\ell=1}^{n}\frac{\operatorname{Mea}(R_\ell)}{\operatorname{LC}(R_\ell)}.
    \]
    According to Lemma~\ref{generalbounds}, $R_\ell$ is a polynomial of magnitude $(D_\FFF,\tilde{O}(B_\FFF+D_\FFF\cdot\LOG(z_{i,\ell})))$, and, in addition, it has the same leading coefficient as $\res(\FFF,x_\ell)$. In particular, its leading coefficient is a non-zero integer, and thus of absolute value larger than or equal to $1$. Thus, we have
    \begin{align*}
    \frac{\operatorname{Mea}(R_\ell)}{\operatorname{LC}(R_\ell)}&\le \|R_\ell\|_2
    =2^{\tilde{O}(B_\FFF+D_\FFF\cdot\LOG(z_{i,\ell}))}
    =2^{\tilde{O}(B_\FFF+D_\FFF\cdot\LOG(\bz_{i}))},
    \end{align*}
    which shows that $\LOG(\partial(\bz_i,\FFF)^{-1})=\tilde{O}(nD_\FFF(B_\FFF+\LOG(\bz_i)))$.

    For the last claim, notice that $\sigma(\bz_i,\FFF)$  appears as one of the factors in the definition of $\partial(\bz_i,\FFF)$. Since the product of all remaining factors is upper bounded by 
    \[
        \prod_{\bz_j\neq \bz_i}M(\bz_j-\bz_i)^{\mu(\bz_j,\FFF)}
        \le \prod_{\bz_j\neq \bz_i} [2\cdot M(\bz_j)\cdot M(\bz_i)]^{\mu(\bz_j,\FFF)}
        \le 2^{D_\FFF}\cdot M(\bz_i)^{D_\FFF}\cdot\prod_{j=1}^N M(\bz_j)^{\mu(\bz_j,\FFF)},
    \]
    the claim follows directly from the bound on $\sum_{i=1}^N\mu(\bz_i,\FFF)\cdot\LOG(\bz_i)$ and on $\LOG(\partial(\bz_i,\FFF)^{-1})$.
\end{proof}

We are now ready to derive one of our main results in this paper. More specifically, the following theorem shows that, in a sufficiently small neighborhood (which we will also quantify) of a $k$-fold solution $\bz=\bz_i$ of $\FFF=0$, $\|\FFF(\bx)\|$ scales like $c\cdot \|\bx\|^k$ with $c$ a constant. We further argue that this implies that a sufficiently good approximation $\Phi$ of the shifted and truncated system $\FFF[\bz]_{\le K}$, with arbitrary $K\ge k$, has a cluster of $k$ solutions near the origin, whereas all remaining solutions are well separated from this cluster. We also give bounds on the approximation error that involve the quantities $\sigma(\bz,\FFF)$ and $\partial(\bz,\FFF)$ that are intrinsic to the hardness of the given polynomial system.

\begin{theorem}\label{thm:countroots}
 Let $\FFF$ be a zero-dimensional system, $\bz$ a zero of $\FFF$ of multiplicity $k$, and $\bm$ be an approximation of $\bz$ with $\|\bm-\bz\|<2^{-L}$. Let $K\ge k$, and $\Phi'=(\phi_i')_{i=1,\ldots,n}$ be a $(K+1)\cdot L$-bit approximation of $\FFF[\bm]_{\le K}$ with polynomials $\phi'_i$ of degree at most $K$, and let  $a_1,\ldots,a_n\in\CC$ be arbitrary complex values of magnitude $0\le |a_i|\le 1$ for all $i$.
 Then, the polynomial system 
  $$\Phi:=(\phi_i'+a_i\cdot x_i^{K+1})_{i=1,\ldots,n}$$
is zero-dimensional, and there exists an $L_0\in\NN$ such that, for any $L\ge L_0$, $\Phi$ has exactly $k$ zeros (counted with multiplicity) of norm smaller than $4\cdot 2^{-L}$, whereas all other zeros have norm larger than $\delta_0:=\frac{\sigma(\bz,\FFF)}{(2n^2D_\FFF)^{32n}}$.
 	In the special case, where each polynomial in $\FFF$ has only integer coefficients, it holds that
 	\[
	L_0=\tilde{O}(
        n\cdot D_\FFF\cdot [n^3+\max_i\frac{\tau_\FFF+d_\FFF}{d_i}+\LOG(\bz)]+\LOG (\del(\bz,\FFF)^{-1}))
        ).
    \]
\end{theorem}
\begin{proof}
We denote $\bz_1,\ldots,\bz_N$, with $\bz=\bz_i$, the zeros of $\FFF$.
    Let $S\in \SSS_{D_{\FFF}}$ be an admissible rotation matrix with respect to $\FFF$ as well as with respect to the shifted system $\FFF[\bz]$. Notice that such a matrix exists as more than half of the matrices in $\SSS_{D_{\FFF}}$ are admissible with respect to $\FFF$ and more than half of the matrices are admissible with respect to $\FFF[\bz]$. Let $\FFF^*:=\FFF\circ S^{-1}$ be the corresponding ``rotation'' of $\FFF$ and $\bz^*_1,\ldots, \bz^*_N$ be the zeros of $\FFF^*$ such that $\bz_j^*=(z_{j,1}^*,\ldots,z_{j,n}^*)=S\cdot \bz_j$. Since $S$ is admissible with respect to $\FFF[\bz]$, Lemma~\ref{lem:rotationsystem} yields that  
    \begin{align*}
    |z_{j,\ell}^*-z_{i,\ell}^*|&\ge (2n^2D_{\FFF})^{-16n}\cdot\|\bz_j-\bz\|
    \ge \frac{\sigma(\bz,\FFF)}{(2n^2D_{\FFF})^{16n}}
    \end{align*}
    for all $\ell$ and $j\neq i$.
    In addition, since $S$ is also admissible with respect to $\FFF$, Lemma~\ref{lem:mildconditions} and Lemma~\ref{lem:rotationsystem} guarantees that each root of the resultant polynomial $\res(\FFF^*,x_\ell)$ is the projection of a finite zero of $\FFF^*$ on the $x_\ell$-coordinate. 
    Thus, $\res(\FFF^*,x_\ell)$ has a $k$-fold root at $z_{i,\ell}^*$, whereas all other roots $z_{j,\ell}^*$ of $\res(\FFF^*,x_\ell)$ have distance at least $\frac{\sigma(\bz,\FFF)}{(2n^2D_{\FFF})^{16n}}$ to $z_{i,\ell}^*$. 
Now, applying Lemma~\ref{pellet} to a disc with center $z_{i,\ell}^*$ and arbitrary radius smaller than $$r_0^*:=\frac{\sigma(\bz,\FFF)}{16D_\FFF^4\cdot (2n^2D_{\FFF})^{16n}},$$ yields that
 \[
  |\res(\FFF^*,x_\ell)(x)|>\frac{1}{3k!}\cdot\left|\frac{\partial^k \res(\FFF^*,x_\ell)}{\partial x^k}(z_{i,\ell}^*)   \right|\cdot |x-z_{i,\ell}^*|^k\quad\text{for all }x\in\CC\text{ with }|x-z_{i,\ell}^*|<r_0^*.
 \]
 \providecommand{\LCRxl}{\operatorname{LC}_\ell}
 Denoting $\LCRxl:=\operatorname{LC}(\res(\FFF^*,x_\ell))$, this further yields 
 \begin{align*}
 \left|\frac{\partial^k \res(\FFF^*,x_\ell)}{\partial x^k}(z_{i,\ell}^*)   \right|&=|\LCRxl|\cdot\prod_{j\neq i}|z_{i,\ell}^*-z_{j,\ell}^*|^{\mu(z_{j,\ell}^*,\res(\FFF^*,x_\ell))}\\
 &\ge |\LCRxl|\cdot\prod_{j\neq i} \left(\frac{\|\bz-\bz_j\|}{(2n^2D_\FFF)^{16n}}\right)^{\mu(\bz_j,\FFF)}\\
 &\ge  \frac{|\LCRxl|}{(2n^2D_\FFF)^{16nD_\FFF}}\cdot \partial(\bz,\FFF),
 \end{align*}
 and thus it follows that 
 \begin{align}\nonumber
  |\res(\FFF^*,x_\ell)(x)|&>\frac{|\LCRxl|}{3k!\cdot(2n^2D_\FFF)^{16nD_\FFF}}\cdot \partial(\bz,\FFF)\cdot |x-z_{i,\ell}^*|^k\\ \nonumber
  &> \frac{|\LCRxl|}{4D_\FFF!\cdot(2n^2D_\FFF)^{16nD_\FFF}}\cdot \partial(\bz,\FFF)\cdot |x-z_{i,\ell}^*|^k\\
  &> \frac{|\LCRxl|}{(4n^2D_\FFF)^{17nD_\FFF}}\cdot \partial(\bz,\FFF)\cdot |x-z_{i,\ell}^*|^k
  \quad\text{if }|x-z_{i,\ell}^*|<r_0^*.\label{boundonRES}
 \end{align}
 
%
 Furthermore, $\res(\FFF^*,x_\ell)$ is contained in the ideal spanned by the polynomials $\FFF^* = (f_1^*, \ldots, f_n^*)$, that is, there exist polynomials $g_{\ell,j}\in\CC[\bx]$ with
 $
  \res(\FFF^*,x_\ell)=\sum_{j=1}^n g_{\ell,j} f_j^*
 $.
 According to Corollary~\ref{cor:arithmetichilbert}, we may assume that
 \begin{align*}
  \log\|g_{\ell,j}\| & \le B_{\FFF^*}                                                    
  =\tilde{O}(n\cdot D_{\FFF^*}+\tau_{\FFF^*}\cdot\max_i\frac{D_\FFF}{d_i})
  =\tilde{O}(n\cdot D_{\FFF}+(\tau_\FFF+d_\FFF)\cdot \max_i\frac{D_\FFF}{d_i})
 \end{align*}
 for all $\ell,j$, where the last inequality follows from Lemma~\ref{lem:rotation_coeffsize}.
 Using Lemma~\ref{generalbounds} then implies that
 \begin{align}\label{bound:gijs}
  \gamma_{\FFF^*}:=\binom{n+D_\FFF}{D_{\FFF}}\cdot 2^{B_{\FFF^*}}\cdot (M(\bz)+1)^{D_\FFF}\ge \sup_{\bx:\|\bx-\bz_i^*\|\le 1} |g_{i,j}(\bx)|\text{ for all }\ell,j.
 \end{align}
 Now, combining \eqref{boundonRES} and \eqref{bound:gijs} yields
 \begin{align*}
  \|\FFF^*(\bx)\|
    & \ge
  \Big[\frac{\min_\ell |\LCRxl| }{n\cdot\gamma_{\FFF^*}\cdot (4n^2D_\FFF)^{17nD_\FFF}}\cdot \del(\bz,\FFF)\Big] \cdot \|\bx-\bz_i^*\|^k 
 \end{align*}
 for all $\bx\in\CC^n$ with $\|\bx-\bz_i^*\|<r_0^*$.

So what can we conclude about our initial (non-rotated) system?
 Since a rotation maintains the Euclidean distance and since the max-norm differs from the Euclidean norm by a factor of at most $\sqrt{n}$, it follows that a point $\bx$ of max-norm $\|\bx\|$ is rotated via $S$ (or $S^{-1}$) onto a point $\bx'$ of max-norm
 $\|\bx'\|\le\|\bx'\|_2=\|\bx\|_2\le \sqrt{n}\cdot \|\bx\|$.
 Hence, it holds that every point $\bx^*=S\bx$ with $\|\bx-\bz\|<r_0:=r_0^*/\sqrt{n}$ satisfies $\|\bx^*-\bz_i^*\|<r_0^*$.
 Thus, for all $\bx$ with $\|\bx-\bz\|<r_0$, it holds that  \begin{align}\label{boundforF:boundary}
  \begin{split}
  \|\FFF(\bx)\|
    &  \ge \underbrace{\Big[\frac{\min_\ell |\LCRxl| }{n\cdot \sqrt{n}^k\cdot\gamma_{\FFF^*}\cdot (4n^2D_\FFF)^{17nD_\FFF}}\cdot \del(\bz,\FFF)\Big]}_{=:c} \cdot \|\bx-\bz\|^k 
 \end{split}
 \end{align}
  Now, suppose that $L\ge \log(8/r_0)$ and thus $ 2^{-L}<\frac{r_0}{8}$. Since $\bm$ is an approximation of $\bz$ with $\|\bz-\bm\|<2^{-L}$, $\FFF[\bm]$ has exactly one solution (namely, $\mathbf{\hat{z}}:=\bz-\bm$) of multiplicity $k$ in 
  $\mathbf{\Delta}_{r_0/2}(0)$ and
  \begin{align*}
   \| \FFF[\bm](\bx)\|=\|\FFF(\bx+\bm)\|>\frac{c}{2^k}\cdot \|\bx\|^k \quad\text{for all }\bx\in\CC^n\text{ with }2^{-L+1}<\|\bx\|<\frac{r_0}{2},
  \end{align*}
  as, for such $\bx$, it holds that $\|\bx\|/2<\|\bx+\bm-\bz\|<r_0$.
  Applying Lemma~\ref{lemma:truncated approximation} to each $\phi_i$ and using the fact that $M(\bm)\le 2M(\bz)$ then shows that, for all $\bx$ with $2^{-L+1}<\|\bx\|<r_0/2$, it holds that
  \begin{align*}
   \|\Phi(\bx)-\FFF[\bm](\bx)\|&\le 2^{-(K+1)L}+\|\bx\|^{K+1}\cdot {{d_\FFF}^n}\cdot 2^{\tau_\FFF+d_{\FFF}} \cdot [M(\bz)\cdot (n+{d_\FFF})^2]^{{d_\FFF}}\\
   &\le \|\bx\|^{K+1}\cdot {{d^n_\FFF}}\cdot 2^{\tau_\FFF+d_{\FFF}+1}\cdot [M(\bz)\cdot (n+{d_\FFF})^2]^{{d_\FFF}} .
  \end{align*}
  Notice that, due to the construction of $\Phi$ and Corollary~\ref{artificialperturbation}, $\Phi$ is zero-dimensional. Hence, Rouch\'{e}'s Theorem applied to $\FFF[\bm]$ and $\Phi$ shows that the polydisc $\mathbf{\Delta}_\rho(0)$ contains the same number of solutions of $\Phi$ and $\FFF[\bm]$ if $2^{-L+1}<\rho<r_0/2$ and if, 
  in addition, $\rho$ fulfills the following inequality 
  \begin{align*}
   \rho^{K+1}\cdot d^n_\FFF\cdot 2^{\tau_\FFF+d_{\FFF}+1}\cdot [M(\bz)\cdot (n+{d_\FFF})^2]^{{d_\FFF}}<\frac{c}{2^k}\cdot \rho^k.
  \end{align*}
  Equivalently, we must have $2^{-L+1}<\rho<r_0/2$ and 
  \[
  \rho^{K+1-k}<\frac{c}{2^k\cdot {d^n_\FFF}\cdot 2^{\tau_\FFF+d_{\FFF}+1} [M(\bz)\cdot (n+{d_\FFF})^2]^{{d_\FFF}}}.
  \]
  Hence, for 
  \begin{align}\label{boundonL}\nonumber
   L>L_0 & :=\max\left[\log\frac{8}{r_0},\log \frac{2^k\cdot {d_\FFF}^n 2^{\tau_\FFF+d_{\FFF}+1} [M(\bz)\cdot (n+{d_\FFF})^2]^{{d_\FFF}}}{c}\right] \\ \nonumber
                & =\tilde{O}(D_\FFF\cdot (n+\max_i\frac{\tau_\FFF+d_\FFF}{d_i}+\LOG(\bz))+\LOG (\del(\bz,\FFF)^{-1})+\LOG(\sigma(\bz,\FFF)^{-1})\\
                &\quad\quad-|\log   \min_\ell |\LCRxl| |).
  \end{align}
  each polydisc $\mathbf{\Delta}_{\rho'}(\bz)$, with arbitrary radius $\rho'\in (4\cdot 2^{-L},r_0/4)$, contains exactly $k$ zeros of $\Phi$. Since $\delta_0<r_0/4$, this proves the first part of the theorem.
  
  It remains to prove the claim bound on $L_0$ for the special case, where $\FFF$ is a polynomial system defined over the integers. For this, we need to estimate the size of the leading coefficient of $\res(\FFF^*,x_\ell)$. Notice that there exists an integer $\lambda$ of size $2^{\tilde{O}(n^3\log d_\FFF)}$ with $\lambda\cdot S^{-1}\in\ZZ^{n\times n}$, and thus 
  \[
  \FFF': (f_1',\ldots,f_n')=(\lambda^{\deg f_1^*}\cdot f_1^*,\ldots,\lambda^{\deg f_n^*}\cdot f_n^*)
  \] is a polynomial system with integer coefficients, which shows that $|\operatorname{LC}(\res(\FFF',x_\ell))|\ge 1$. Using~\cite[Thm.~2.3 and~3.5]{cox2005using} then shows that
  \[
  |\LCRxl|=\lambda^{-nD_\FFF}\cdot  |\operatorname{LC}(\res(\FFF',x_\ell))|\ge \lambda^{-nD_\FFF}=2^{-O(n^4D_\FFF)}.
  \]
  Hence, the bound follows from (\ref{boundonL}) and the bound for $\LOG\sigma(\bz,\FFF)^{-1}$ from Lemma~\ref{boundsonsizeandseparation}.
  \end{proof}

From the previous Theorem, we now immediately obtain the following result by setting $\bm:=\bz$ and $\Phi:=\FFF[\bz]_{\le K}$ for an arbitrary $K\ge k$.

  \begin{corollary}\label{samemultiplicity:truncated}
   Let $\bz$ be a $k$-fold zero of a zero dimensional system $\FFF$ and $K\ge k$. Then, $\FFF[\bz]_{\le K}$ has a $k$-fold zero at the origin, and all other zeros have norm larger than $\delta_0:=\frac{\sigma(\bz,F)}{(2n^2D_\FFF)^{32n}}$.\footnote{We remark that $\FFF[\bz]_{\le K}$ does not necessarily have to be zero-dimensional. Rouch\'e's Theorem only guarantees that the polydisc $\mathbf{\Delta}_{\delta_0}(0)$ contains exactly $k$ zeros of $\FFF[\bz]_{\le K}$.}
  \end{corollary}

  We can now show that Algorithm~\ref{algo:truncate} terminates and yields a correct result assuming that $L$ is large enough and the oracle, which provides an approximation $m$ of the solution $\bz$, returns a correct answer.

  \begin{theorem}\label{thm:main}
   If \emph{$\cert(\FFF, \mathbf{\Delta},K)$} returns an integer $k\ge 0$, then the polydisc $\mathbf{\Delta}=\mathbf{\Delta}_r(\bm)$ contains exactly $k$ solutions of $\FFF$ counted with multiplicity.
   Vice versa, suppose that $\bz$ is a solution of $\FFF$ of multiplicity $k$ and $K\ge k$, then there exists a positive integer $L^*$ of size
   \[
   L^*= 
\tilde{O}(L_0+(K+1)^n\cdot \log\frac{1}{\delta_0}+d_\FFF\cdot\LOG(\bz)).
    \]
    with $L_0$ and $\delta_0$ as in Theorem~\ref{thm:countroots}, such that $\cert(\FFF, \mathbf{\Delta},K)$ returns $k$ with probability at least $1/2$ if $r\le 2^{-L^*}$ and $\|\bz-\bm\|<\frac{r}{64n(K+1)^n}$. If $\FFF$ has only integer coefficients, it holds:
\begin{align*}
L^*&=\tilde{O}(D_\FFF\cdot \max_i\frac{d_\FFF+\tau_\FFF}{d_i}+D_\FFF\cdot\LOG(\bz)+\LOG(\partial(\bz,\FFF)^{-1})+(K+1)^n\cdot \LOG(\sigma(\bz,\FFF)^{-1}))\\
&=\tilde{O}((K+1)^n\cdot D_\FFF\cdot [D_\FFF+\tau_\FFF\cdot\max_{\ell}\frac{D_\FFF}{d_\ell}]).
\end{align*}
  \end{theorem}
  \begin{proof}
   For the first part, we proceed similarly as in the proof of Theorem~\ref{thm:countroots}, however, we work with the system $\Phi$ instead of the initial system $\FFF$. From Line~\ref{line:ifrootscontained} in the algorithm, we already know that $\Phi$ has exactly $k$ solutions with ($\|.\|$-) norm less than $\frac{r}{n}$, whereas all other solutions have norm at least $nr$. 
   Now, when considering a random rotation matrix $S\in\SSS_{D_\Phi}$, the corresponding rotated system $\Phi^*=\Phi \circ S^{-1}$ has exactly $k$ solutions with norm less than $\frac{r}{\sqrt{n}}$, whereas all other solutions have norm at least $\sqrt{n}\cdot r$. We may now further write
   \begin{align}\label{cofactorPHIstar}
    \res(\Phi^*,x_\ell)=\gamma_{\ell,1}\cdot \phi_1^*+\cdots + \gamma_{\ell,n}\cdot \phi_n^*
   \end{align}
   with polynomials $\gamma_{\ell,j}\in\CC[\bx]$. From Part~\ref{partd} of Lemma~\ref{generalbounds} and Corollary~\ref{cor:arithmetichilbert} we conclude that
   \[
   	\binom{D_{\Phi^*}+n}{D_{\Phi^*}}\cdot B_{\Phi^*}
	\ge \sup_{\bx:\|\bx\|\le 1}|\gamma_{\ell,j}(\bx)|\text{ for all }\ell,j.
	\]
   In addition, since $\TTT_*(\Delta_{\frac{r}{\sqrt{n}}}(0), \res(\Phi^*, x_\ell))=(\textsc{True},k^-_\ell,\LB_\ell^-)$ and $\TTT_*(\Delta_{\sqrt{n}r}(0), \res(\Phi^*, x_\ell))=(\textsc{True},k^+_\ell,\LB_\ell^+)$ for all $\ell$ (Line~\ref{line:ifTktest}), we conclude from Lemma~\ref{pellet} that 
   \[
   	|\res(\Phi^*,x_\ell)(\bx)|>\min(\LB_\ell^-,\LB_\ell^+)\text{ for all }i\text{ and all }\bx\text{ with }|x_\ell|=\frac{r}{\sqrt{n}}\text{ or }|x_\ell|=\sqrt{n} r.
	\]
   Hence, using (\ref{cofactorPHIstar}), this shows that
   \[
    \|\Phi^*(\bx)\|\ge \frac{\min_{\ell=1,\ldots,n}\min(\LB_\ell^-,\LB_\ell^+)}{n\cdot\binom{D_{\Phi^*}+n}{D_{\Phi^*}}\cdot B_{\Phi^*}}\text{ for all }\bx\text{ with }\|\bx\|=\frac{r}{\sqrt{n}}\text{ or }\|\bx\|=\sqrt{n}r.
   \]
   Since a holomorphic mapping cannot take its minimum or maximum in the interior of some domain, it thus follows that the above inequality even holds for any $\bx$ with $\frac{r}{\sqrt{n}}\le \|\bx\|\le \sqrt{n}r$. It thus follows that
   \[
    \|\Phi(\bx)\|\ge \frac{\min_{\ell=1,\ldots,n}\min(\LB_\ell^-,\LB_\ell^+)}{n\cdot\binom{D_{\Phi^*}+n}{D_{\Phi^*}}\cdot B_{\Phi^*}}\text{ for all }\bx\text{ with }\|\bx\|=r.
   \]
    According to (\ref{def:UB}), $\UB(\bm,r)$ constitutes an upper bound on the error $\| \FFF[\bm](\bx)-\Phi(\bx)\|$ for any $\bx$ with $\|\bx\|\le 1$. Hence, in particular, we also have
   \[
    \| \FFF[\bm](\bx)-\Phi(\bx)\|\le \UB(\bm,r)\text{ for all }\bx\text{ with }\x\text{ with }\|\bx\|=r.
   \]
   Hence, using Rouch\'e's Theorem, we conclude that $\FFF[\bm]$ and $\Phi$ have the same number of solutions in the polydisc $\mathbf{\Delta}_r(0)$.

   It remains to prove the second claim. For this, suppose that $\FFF$ has a $k$-fold solution at $\bz$ with $\|\bm-\bz\|<\frac{r}{64n(K+1)^n}$ and that
\begin{align}\label{cond:L}
   \begin{split}
       L:=\lceil \log\frac{32n(K+1)^n}{r}\rceil>L_1
       &:=\max(L_0,\log \frac{1}{\delta_0}+\log [2048\cdot n^{2}\cdot (2n^2D_\FFF(K+1)^n)^{32n}])\\
       &=\tilde{O}(L_0+\log\frac{1}{\delta_0}+n^2\log D_\FFF),
    \end{split}
\end{align} 
with $\delta_0=\frac{\sigma(\bz,\FFF)}{(2n^2D_\FFF)^{32n}}$ as defined in Theorem~\ref{thm:countroots}. Let $\Phi$ be an approximation of $\FFF[\bm]_{\le K}$ as defined in Theorem~\ref{thm:countroots}. 
   Then, $\Phi$ has $k$ solutions $\bz_1,\ldots,\bz_k$ of norm $\|\bz_i\|<4\cdot 2^{-L}<r/(2n)$, whereas all remaining solutions, denoted by $\bar{\bz}_1,\ldots,\bar{\bz}_m$, 
   have norm $\|\bar\bz_j\|> \delta_0>2nr$. We thus conclude that the if-condition is satisfied in Line~\ref{line:ifrootscontained}.
   Now, when choosing a random rotation matrix $S\in\SSS_{D_\Phi}$, the solutions $\bz_i$ near the origin are mapped to solutions $\bz_i^*:=S\circ \bz_i$ of $\Phi^*=\Phi\circ S^{-1}$ with norm $\| \bz_i^*\|<4\sqrt{n}\cdot 2^{-L}\le\frac{r}{16\sqrt{n}(K+1)^n}$, whereas the remaining solutions are mapped to solutions $\bar{\bz}^*_j$ of $\Phi^*$ with norm $\|\bar{\bz}^*_j\|>\delta_0/\sqrt{n}>2\sqrt{n}r$.
   In addition, with probability more than $1/2$, we have 
   \[
    |z^*_{j,\ell}|
    \ge \|\bar{\bz}_{j}\| \cdot (2n^2D_{\Phi^*})^{-16n}
    >\delta_0\cdot (2n^2(K+1)^n)^{-16n}
    > 16\cdot\sqrt{n}\cdot (K+1)^{4n} r
     \]
     for all $j\in[m]$ and all $\ell\in[n].$ 
   This implies that each of the resultant polynomials $\res(\Phi^*,x_\ell)$ has $k$ roots of absolute value less than $\frac{r}{16(K+1)^n\sqrt{n}}$, whereas all remaining roots are of absolute value larger than $16\cdot\sqrt{n}\cdot (K+1)^{4n} r$. 
   Notice that each polynomial $\res(\Phi^*,x_\ell)$ has degree $(K+1)^n$, and thus  Lemma~\ref{pellet} guarantees success in Line~\ref{line:ifTktest} of the algorithm. 
   Now, recall the lower bounds $\LB_\ell^-$ and $\LB^+_\ell$
for $|\res(\Phi^*,x_\ell)|$ on the boundary of $\Delta_{r/\sqrt{n}}(0)$ and $\Delta_{\sqrt{n}r}(0)$, respectively, as computed in Line~\ref{line:lowerbound}.
For arbitrary $x\in\CC$ with $|x|=r/\sqrt{n}$, we have
 \begin{align*}
 \LB_\ell^-
 & \ge
 \frac{ |\res(\Phi^*,x_\ell)(x)|}{5}>\frac{|\operatorname{LC}(\res(\Phi^*,x_\ell))|}{5}\cdot \left(\frac{r}{2\sqrt{n}}\right)^k\cdot  \left(\frac{\delta_0\cdot (2n^2(K+1)^n)^{-16n}}{2}   \right)^{(K+1)^n-k}\\
 &>|\operatorname{LC}(\res(\Phi^*,x_\ell))|\cdot  (2n^2(K+1)^n)^{-32n(K+1)^n} \cdot r^{k}\cdot \delta_0^{(K+1)^n}.
 \end{align*}
Using the fact $\LB_\ell^+ \ge  \frac{ |\res(\Phi^*,x_\ell)(x)|}{5}$ for all $x$ with $|x|=2\sqrt{n}r$, an analogous computation shows that $\LB_\ell^+$ fulfills the same bound, that is,
 \[
 \LB_\ell^+\ge
 |\operatorname{LC}(\res(\Phi^*,x_\ell))|\cdot  (2n^2(K+1)^n)^{-32n(K+1)^n} \cdot r^{k}\cdot \delta_0^{(K+1)^n}.
 \]
From Lemma~\ref{artificialperturbation} and our construction of $\Phi^*$, the leading coefficient of each polynomial $\res(\Phi^*,x_\ell)$ is a non-zero integer, hence we obtain that
\begin{align*}
\LB(\bm,r)=\frac{\min_{\ell=1,\ldots,n}\min(\LB_\ell^-,\LB_\ell^+)}{n\cdot\binom{D_{\Phi^*}+n}{D_{\Phi^*}}\cdot 2^{B_{\Phi^*}}}
&\ge\underbrace{\frac{\delta_0^{(K+1)^n}}{n(2n^2(K+1)^n)^{32n(K+1)^n}\cdot\binom{D_{\Phi^*}+n}{D_{\Phi^*}}\cdot 2^{B_{\Phi^*}}}}_{=:C}\cdot r^{k}.
\end{align*}
Notice that, for small $r$, $\operatorname{LB}(\bm,r)$ scales like $C\cdot r^k$, with a constant $C$ that does not depend on $r$. 
The upper bound
\[
\UB(\bm,r)  =r^{K+1} \cdot \underbrace{{d_\FFF}^n 2^{\tau_\FFF + 2} [M(\bm)\cdot (n+d_\FFF)^2]^{d_\FFF}}_{=:C'}
\]
scales like $C'\cdot r^{K+1}$, and thus our algorithm succeeds if $r$ fulfills the condition in (\ref{cond:L}) (i.e. $\lceil \log\frac{32n(K+1)^n}{r}\rceil\ge L_1$) and $r^{K-k+1}<\frac{C}{C'}$. Both condition are fulfilled if
\begin{align*}
\log(1/r)&\ge L^*:=\max(L_1,
\log\frac{C'}{C})= 
\tilde{O}(L_0+(K+1)^n\cdot \log\frac{1}{\delta_0}+d_\FFF\cdot\LOG(\bz)).
\end{align*}
The claimed bound on $L^*$ for the special case where $\FFF$ is defined over the integers follows directly from the corresponding bound on $L_0$ from Theorem~\ref{thm:countroots} and our bounds on $\LOG(\bz)$, $\LOG(\sigma(\bz,\FFF)^{-1})$, and $\LOG(\partial(\bz,\FFF)^{-1})$ from Lemma~\ref{boundsonsizeandseparation}.
  \end{proof}
\section{Application: Computing the Zeros of a Bivariate System}\label{sec:experiments}

In this section, we report on an application of our technique in the context of elimination methods for the bivariate case. More precisely, we incorporate the algorithm $\cert$ as an inclusion predicate in the \textsc{Bisolve} algorithm~\cite{DBLP:journals/tcs/BerberichEKS13,DBLP:journals/jc/KobelS15}. Comparing \textsc{Sage} implementations of the original \textsc{Bisolve} algorithm and its modified variant, we empirically show that the idea of truncating the original system with respect to the multiplicity of the solution yields a considerable performance improvement.  
\textsc{Bisolve} is a classical elimination method for computing the real~\cite{DBLP:journals/tcs/BerberichEKS13} or complex~\cite{DBLP:journals/jc/KobelS15} zeros within a given polydisc $\mathbf{\Delta}=\Delta_1\times\Delta_2\subset \CC^2$ of a bivariate system 
\[
  \FFF:f_1(x_1,x_2)=f_2(x_1,x_2)=0,\text{ with polynomials }f_1,f_2\in \ZZ[x_1,x_2].
\] 
It achieves the best known complexity bound (i.e. $\tilde{O}(d_\FFF^6+d_\FFF^5\cdot \tau_\FFF)$ bit operations for computing all complex solution) that is currently known for this problem, and its implementation shows superior performance when compared to other complete and certified methods. 
As we aim to modify the \textsc{Bisolve} algorithm at some crucial steps, we start with a brief description of the original version.

\paragraph{\textsc{Bisolve} in a Nutshell.}  
In an initial \emph{projection phase}, \textsc{Bisolve} computes a set $C$ of \emph{candidate regions} using resultant computation and univariate root finding. More specifically, we first compute the hidden-variable resultants $R_\ell(x):=\res(\FFF, x_\ell)$ for $\ell=1,2$. Then, for each root $z_{\ell,i}$ in $\Delta_\ell$, we compute an isolating disc $\Delta_{\ell,i}$ such that $\TTT_\star(\Delta_{\ell,i},R_\ell)=(\text{\textsc{True}},k_{\ell,i},\LB_{\ell,i})$. 
That is, the $\TTT_\star$-test succeeds and yields the multiplicity of $z_{\ell,i}$ as a root of $R_\ell$ as well as a lower bound for $|R_\ell|$ on the boundary of $\Delta_{\ell,i}$.
By taking the pairwise product of any two discs $\Delta_{1,i}$ and $\Delta_{2,j}$, we obtain a set $C$ of polydiscs $\mathbf{\Delta}_{i,j}:=\Delta_{1,i}\times \Delta_{2,j}$ in $\CC^2$. 
Notice that each solution $\bz$ in $\mathbf{\Delta}$ of $\FFF$ must be one of the \emph{candidate solutions} $\bz_{i,j}:=(z_{1,i},z_{1,j})$ as each coordinate of $\bz$ is a root of the corresponding polynomial $R_{\ell}$. Hence, each solutions must be contained in one of the candidate regions, even though most candidate regions do not contain any solution. In addition, each candidate region $\mathbf{\Delta}_{i,j}$ contains at most one solution, which must be $\bz_{i,j}$ 

In the \emph{validation phase}, the algorithm checks for every candidate region $\mathbf{\Delta}_{i,j}$ whether it contains a solution or not. In other words, we check whether the corresponding candidate solution $\bz=\bz_{i,j}$ is actually a solution or not. The approach used in \textsc{Bisolve} shares many similarities to the algorithm $\cert$ as proposed in this paper. That is, we write
\[
  R_\ell=g_{\ell,1}\cdot f_1 + g_{\ell,2}\cdot f_2\text{ with }g_{\ell,1},g_{2,\ell}\in\ZZ[x,y]\text{ and }\ell=1,2
\]
and compute an upper bound $\UB$ for $|g_{\ell,i}(\bx)|$ for $\ell=1,2$, $i=1,2$, and arbitrary $\bx\in \mathbf{\Delta}_{i,j}$. Similar as in $\cert$,  this is achieved without actually computing the polynomials $g_{\ell,1}$ and $g_{\ell,2}$, but by exploiting the fact that these polynomials can be written as determinants of ``Sylvester-like'' matrices\footnote{Notice that this is one crucial point, where our novel approach differs from $\textsc{Bisolve}$. Namely, for $\cert$, we use the results from~\cite{D2013} on the arithmetic Nullstellensatz to derive corresponding bounds on the cofactors $g_{\ell,j}$. This was necessary for generalizing the method to arbitrary dimension. Another crucial difference is that no truncation of the system is considered in \textsc{Bisolve}.}; see~\cite{DBLP:journals/jc/KobelS15} for details.  
Together with the lower bounds $\LB_{1,i}$ and $\LB_{2,j}$ as computed above this yields 
a lower bound 
$\LB^*=\frac{\min(\LB_{1,i},\LB_{2,j})}{2\UB}$ for $\|\mathcal{F}\|_\infty=\max(|f_1|,|f_2|)$ on the boundary of $\mathbf{\Delta}_{i,j}$. 

Now, in order to discard or certify $\bz$ as a solution, \textsc{Bisolve} proceed in rounds, where a $2^m$-bit approximation $\zeta$ of $\bz$ is computed at the beginning of the $m$-th round. 
As an \emph{exclusion predicate}, interval arithmetic is used in order to compute a superset $\square f_\ell(\mathbf{\Delta}_{2^{-m}}(\zeta))$ of $f_\ell(\mathbf{\Delta}_{2^{-m}}(\zeta))$ for $\ell=1,2$. If we can show that either $f_1$ or $f_2$ does not vanish, the candidate is discarded. As an \emph{inclusion predicate} the above lower bound $\LB^*$ on the boundary of $\mathbf{\Delta}_{i,j}$ is compared to the values that $f_1$ and $f_2$ take at the approximation $\zeta$ of the candidate $\bz$. 
More specifically, if $\max(|f(\zeta)|,|g(\zeta)|)<\LB^*$, then $\mathbf{\Delta}_{i,j}$ contains a solution; see Theorem 4 in \cite{DBLP:journals/tcs/BerberichEKS13}. 
If neither the exclusion nor the inclusion predicate applies, we proceed with the next round.

\paragraph{The \textsc{BisolvePlus} routine.}
Notice that, even though \textsc{Bisolve} computes the set
\[
    \mathbf{Z}:=\{\bz_{i,j}\in\mathbf{\Delta}:f_1(\bz_{i,j})=f_2(\bz_{i,j})=0\}
\]
of all solutions of $\FFF$ within $\mathbf{\Delta}$, it does not reveal the 
multiplicity $k$ of a specific solution $\bz=\bz_{i,j}=(z_{1,i},z_{2,j})\in\mathbf{Z}$. 
However, due to the properties of the resultant polynomials, it holds that $k=\mu(z_{1,i},R_1)$ if the following two conditions are both fulfilled: 
\begin{align}\label{condition1}
    &\deg_{x_2} f_1=\deg f_1\text{ and }\deg_{x_2} f_2=\deg f_2 \\ \label{condition2}
    &\forall x_2\in\CC\setminus\{ z_{2,j}\} : 
    f_1(z_{1,i},x_2)\neq 0\text{ or }f_2(z_{1,i},x_2)\neq 0
\end{align}

{\LinesNumberedHidden
\begin{algorithm}[t!]
    \BlankLine
    \Input{Zero-dimensional bivariate system $\FFF: f_1(x_1,x_2)=f_2(x_1,x_2)=0$ with $f_1,f_2\in\QQ[x]$, polydisc $\mathbf{\Delta}_r(\bm)$}
    \Output{$\mathbf{Z^+}$ such that $\mathbf{Z^+}=\{(\bz,k):f_1(\bz)=f_2(\bz)=0\text{, }k=\mu(\bz,\FFF)\}$.}
    \BlankLine
    
    \For{$\rho:=1,2,4,\ldots$.}{
        Choose a matrix $S\in\mathcal{S}_{D_\FFF}$ and compute $\FFF^*=(f^*_1,f^*_2) = \FFF\circ S^{-1}$\;
        \If{$\deg_{x_2}f^*_1 = \deg f^*_1$ and $\deg_{x_2} f^*_2 = \deg f^*_2$}{
            \BlankLine
            
            Call \textsc{Bisolve} with input $\FFF^*$ and $\mathbf{\Delta}_{2r}(\bm)$ to compute (for $\ell=1,2$):
            \begin{itemize}
                \item Discs $\Delta_{\ell,i}$, $i=1,\ldots,i_\ell$, that isolate the roots $z_{\ell,i}$ of $R_\ell:=\res(\FFF^*,x_\ell)$.
                \item The multiplicity $k_{\ell,i}=\mu(z_{\ell,i})$ of $z_{\ell,i}$ as a root of $R_\ell$.
                \item The set $\mathbf{Z}:=\{\bz_{i,j}=(z_{1,i},z_{2,j}):f^*_1(\bz_{i,j})=f^*_2(\bz_{i,j})=0\}$ of all solutions of $\FFF^*$.
            \end{itemize}
            \For{each solution $\bz=\bz_{i,j}\in\mathbf{Z}$}{
                \For{$\ell=1,2$}{
                    Compute disjoint discs $D_{\ell,1},\ldots,D_{\ell,s_\ell}$ of radius less than $2^{-\rho}$ such that
                    \begin{itemize}
                        \item Each disc $D_{\ell,s}$ contains at least one root of $f^*_\ell(z_i,x_2)\in\CC[x_2]$.
                        \item $\bigcup_{s=1}^{s_\ell} D_{\ell,s}$ contains all roots of $f^*_\ell(z_i,x_2)\in\CC[x_2]$.
                    \end{itemize}
                }
                Determine the set 
                \[
                    \mathbf{D^*}:=\{D_{1,s}:\exists D_{2,s'}\text{ with }D_{1,s}\cap D_{2,s'}\neq \emptyset\}
                \] 
                of all discs $D_{1,s}$ that have non-empty intersection with one of the discs $D_{2,s'}$.
                \BlankLine
                \If{for all $D_{1,s}\in \mathbf{D^*}$ it holds that $D_{1,s}\subset \Delta_i$}{
                    Set $b_{i,j}=\textsc{True}$\;
                }
            } 
            \If{$\bigwedge_{i,j}b_{i,j}$}{
                \Return{$\{(S^{-1}\circ\bz_{i,j},k_i):\bz_{i,j}\in \mathbf{Z}\text{ and }\bz_{i,j}\in \mathbf{\Delta}_r(\bm)\}$}
            }
        }
    }
    \caption{$\textsc{BisolvePlus}$}
    \label{algo:bisolveplus}
\end{algorithm}
}

The first condition guarantees that there is no solution of $\FFF$ at infinity above any $z\in\CC$, whereas the second condition guarantees that there is no other finite (complex) solution of $\FFF$ that shares the first coordinate with $\bz$.
We remark that it is easy to check the first condition, however, checking the second condition is more difficult. This is due to the fact that $\bz$ might be the only solution in $\mathbf{\Delta}$ of $\FFF$ with $x_1=z_{1,i}$, but there is a another solution of $\FFF$ with $x_2=z_{i,1}$ that is not contained within $\mathbf{\Delta}$.
We aim to address this problem by the following approach (see also Algorithm~\ref{algo:bisolveplus}): 

Let $L\in\NN$ be fixed non-negative integer. In a first step, we check whether (\ref{condition1}) is fulfilled. If this is not the case, we return \textsc{False}, otherwise, we proceed. Now, for each solution $\bz_{i,j}\in \mathbf{Z}$ and each $\ell\in\{1,2\}$,
we use a complex root finder\footnote{Each of the methods in~\cite{DBLP:conf/issac/BeckerS0XY16,DBLP:journals/corr/BeckerS0Y15,DBLP:journals/jsc/MehlhornSW15} applies to polynomials with arbitrary complex coefficients. Also, for computing only approximations of the roots, there are no restrictions on the multiplicities of the roots. In our implementation, we use a strongly simplified variant of the algorithm from~\cite{DBLP:conf/issac/BeckerS0XY16}.} to compute a set of pairwise disjoint discs $D_{\ell,j}$ of radius less than $2^{-\rho}$ such that each disc contains at least one root and the union of all discs $D_{\ell,j}$ contains all complex roots of $f_\ell(z_{1,i},x_2)\in\CC[x]$. Then, we determine all discs $D_{1,j_1},\ldots,D_{1,j_s}$ that have a non-empty intersection with one of the discs $D_{2,j'}$. It follows that each common root of $f(z_{1,i},x_2)$ and $f(z_{1,i},x_2)$ must be contained in one of the discs $D_{1,j_{s'}}$. Hence, if each of these discs is contained in $\Delta_{2,i}$, then $x_2=z_{2,j}$ is the unique solution of $f(z_{1,i},x_2)=f(z_{2,i},x_2)=0$, and thus (\ref{condition2}) is fulfilled. In this case, we may conclude that $\mu(z_{1,i},R_1)$ equals the multiplicity of $\bz$. If we succeed in computing the multiplicities for all solutions in $\mathbf{Z}$, we return the solutions together with their corresponding multiplicities. Otherwise, we return \textsc{False}.

Obviously, the above approach cannot succeed if one of the above conditions is not fulfilled. However, even if both conditions are fulfilled, it may still fail due to the fact that $\rho$ has not been chosen large enough. 

\begin{lemma}\label{successofAlgobisolveplus}
Suppose that both conditions (\ref{condition1}) and (\ref{condition2}) are fulfilled. Then, there exists a $L_0\in\NN$ such that Algorithm~\ref{algo:bisolveplus} succeeds for all $L>L_0$.
\end{lemma}

\begin{proof}
Let $\epsilon$ be a lower bound on the distance between any distinct roots of $f_1(z_{1,i},x_2)$ and $f_2(z_{1,i},x_2)$. Now, if $2^{-L}<\epsilon/4$, then two discs $D_{1,j'}$ and $D_{2,j''}$ can only intersect if they contain a common root of $f_1(z_{i,1},x_2)$ and $f_2(z_{i,1},x_2)$. Since $z_{2,j}$ is the only common root, we thus conclude that each of the discs $D_{1,j_{s'}}$ must contain $z_{2,j}$. Hence, if $L$ is large enough, then $\Delta_2$ contains $D_{1,j_{s'}}$.
\end{proof}

The problem with this approach is that we do neither know in advance whether the condition \eqref{condition2} is fulfilled nor do we know whether $\rho$ has been chosen sufficiently large. In order to overcome this issue, we consider a rotation of the system by means of a rotation matrix $S\in\mathcal{S}_{D_\FFF}$. Then, with probability at least $1/2$, both conditions \eqref{condition1} and \eqref{condition2} are fulfilled for the rotated system $\FFF^*:=\FFF\circ S^{-1}$. 
We now proceed in rounds (numbered by $m$), where, in each round, we choose a matrix $S\in\mathcal{S}_{D_\FFF}$ at random and run Algorithm~\ref{algo:bisolveplus}
with input $\FFF^*=\FFF\circ S^{-1}$ and $\rho:=2^{m}$. 
Since there are only finitely many different choices for $S$ and since the conditions \eqref{condition1} and \eqref{condition2} are fulfilled for at least the half of the systems $\FFF^*$, 
Lemma~\ref{successofAlgobisolveplus} guarantees that, for sufficiently large $\rho$, Algorithm~\ref{algo:bisolveplus} returns the solutions of $\FFF^*$ in $\mathbf{\Delta}$ together with the corresponding multiplicities with probability at least $1/2$.

\paragraph{New Validation Phase.}
We are now ready to modify \textsc{Bisolve} by using an inclusion predicate based on the algorithm $\cert$. More specifically, let $C:=\{\bz_{i,j}\}_{i,j}$ be the set of candidate solutions and let  $\mathbf{\Delta}_{i,j}=\Delta_{1,i}\times\Delta_{2,j}$ be the corresponding candidate regions as computed in the projection phase of \textsc{Bisolve}. 
The validation routine, see also Algorithm~\ref{algo:validate}, that is called for each candidate solution $\bz=\bz_{i,j}=(z_{1,i},z_{2,j})$ again works in rounds, where in round $m$, we compute a $L=2^m$-bit approximation $\zeta=(\zeta_{1},\zeta_{2})$ of $\bz$ such that $\|\zeta-\bz\|<2^{-L}$ and $\Delta_{64\max(d_1,d_2)^2\cdot 2^{-L}}(\zeta)\subset \Delta_{2,j}$.
The exclusion predicate is identical to the original \textsc{Bisolve} routine, i.e., we check whether we can guarantee that $f_1$ or $f_2$ does not vanish on $\mathbf{\Delta}_{2^{-L}}(\zeta)$ by evaluating interval extensions $\square f_\ell(\mathbf{\Delta}_{2^{-L}}(\zeta))$ for $\ell=1,2$ using interval arithmetic. The inclusion predicate now works as follows. There is still one tiny detail that prevents us from directly plugging in $\cert$ as an inclusion predicate. Namely, even if the candidate solution would actually turn out to be a solution of the system,  we do not know the multiplicity $k$ of this solution. Thus, we have to search for the multiplicity $k$. 
As we have seen in the previous section that $\cert(\FFF,\mathbf{\Delta}_{2^{-L}}(\zeta),K)$ actually succeeds for any $K\ge k$, we can use exponential search for $k$ by calling $\cert(\FFF,\mathbf{\Delta}_{2^{-L}}(\zeta),K)$ for $K=1,2,4,\ldots,2^{\lceil\log(\min(k_1, k_2))\rceil}$, where $k_\ell$ for $\ell\in\{1,2\}$ is the multiplicity of $z_\ell$ as a root of $R_\ell(x):=\res(\FFF, x_\ell)$. 
In the calls to $\cert$, we use the above described \textsc{BisolvePlus}-routine in order to implement the computation of the solutions of the truncated system in Line~\ref{line:oracle} of Algorithm~\ref{algo:truncate}.
We remark that our actual implementation of the new inclusion predicate differs slightly from the description of $\cert$ in one more detail. For efficiency reasons, we consider a partial change of order of the three considered steps \emph{Solving the truncated system}, \emph{Projection step}, and \emph{Bound Computation and Comparison}. 
{\LinesNumberedHidden
\begin{algorithm}[t!]
\Input{Zero-dimensional system $\FFF=(f_1,f_2)$ with polynomials $f_1,f_2\in\ZZ[\bx]$ of degrees $d_1$ and $d_2$, respectively, polydisc $\mathbf{\Delta}=\Delta_1\times\Delta_2\subset\CC^2$, candidate $\bz=(z_1,z_2)$ and multiplicities $k_\ell$ s.t.~the multiplicity of $z_\ell$ as a root of $R_\ell(x):=\res(\FFF, x_\ell)$ is $k_\ell$ for $\ell\in\{1,2\}$.}
\Output{$k\in\NN_0$. If $k\ge 1$, then there is a unique solution $\bz$ of $\FFF=0$ of multiplicity $k$ within the polydisc $\mathbf{\Delta}$. Otherwise,  $\mathbf{\Delta}$ contains no solution.}
\BlankLine
\For{$L=1,2,4,\ldots$}{
  Compute $L$-bit approximation $\zeta$ of $\bz$\;
  \If{$\Delta_{2^{-L + 2\lceil\log(\min(k_1, k_2))\rceil+6}}(\zeta)\subset \mathbf{\Delta}$}{
      \If{$0\notin \square f_1(\mathbf{\Delta}_{2^{-L}}(\zeta))$ or $0\notin \square f_2(\mathbf{\Delta}_{2^{-L}}(\zeta))$}{
          \Return{$0$}
        }
        \For{$K=1,2,4,\ldots,2^{\lceil\log(\min(k_1, k_2))\rceil}$}{
            \If{$\cert(\FFF,\Delta_{2^{-L}}(\zeta), K) = k\ge 1$}{
                \Return{$k$}
            }
        }
    }
}
\BlankLine
\caption{\textsc{Validate}}
\label{algo:validate}
\end{algorithm}}

\subsection{Setting}
We performed experiments on a compute server with 48 Intel (R) Xeon (R) CPU E5-2680 v3 @ 2.50GHz cores and a total of 256 GB RAM running Debian GNU/Linux 8. All code was implemented in SageMath version 7.6, release date 2017-03-25. 

\begin{figure}[ht!]
	\begin{center}
		\subfloat{
			\includegraphics[width=0.48\textwidth]{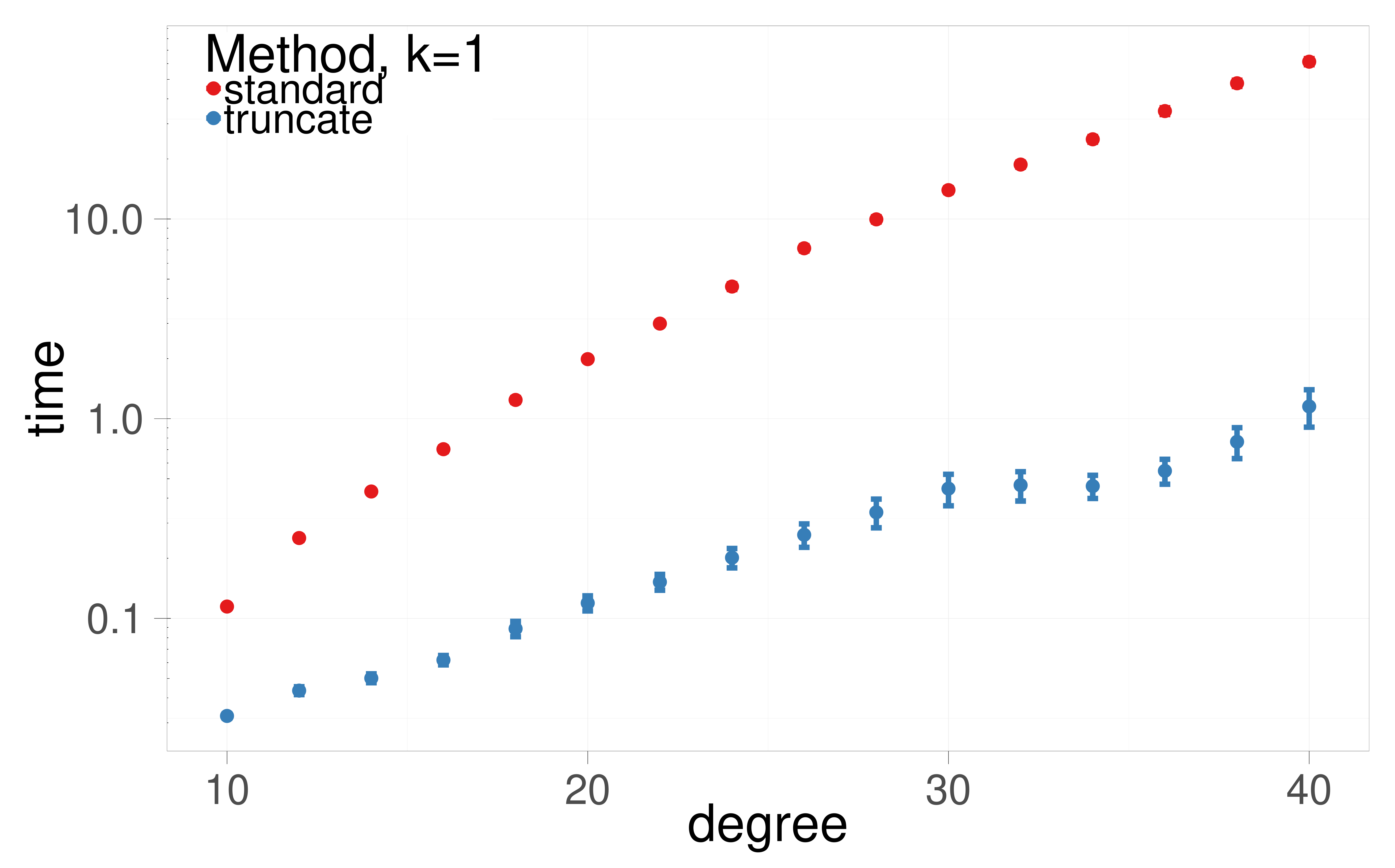}
		}
		\subfloat{
			\includegraphics[width=0.48\textwidth]{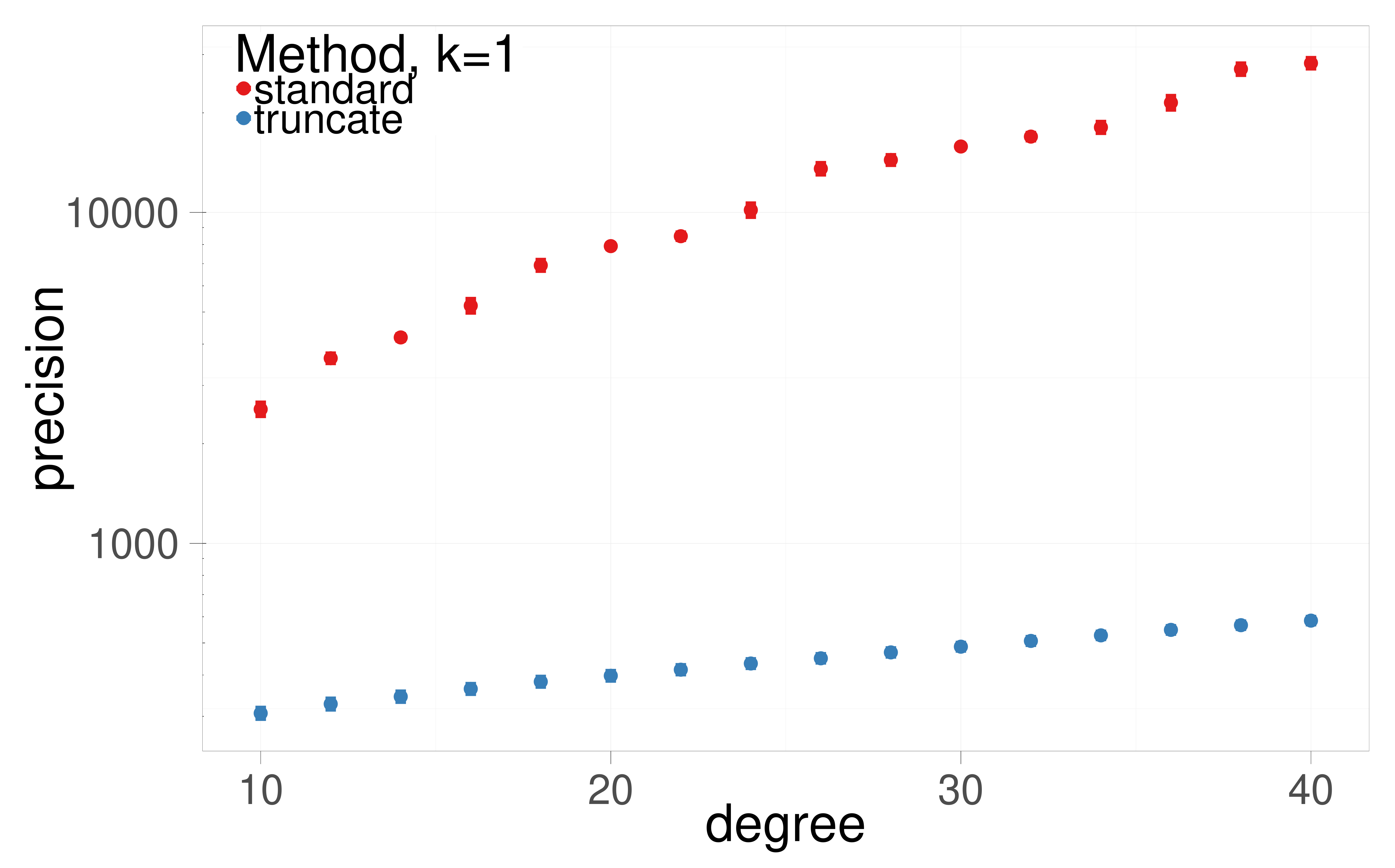}
		}
	\end{center}
    \vspace{-13mm}
    \begin{center}
		\subfloat{
			\includegraphics[width=0.48\textwidth]{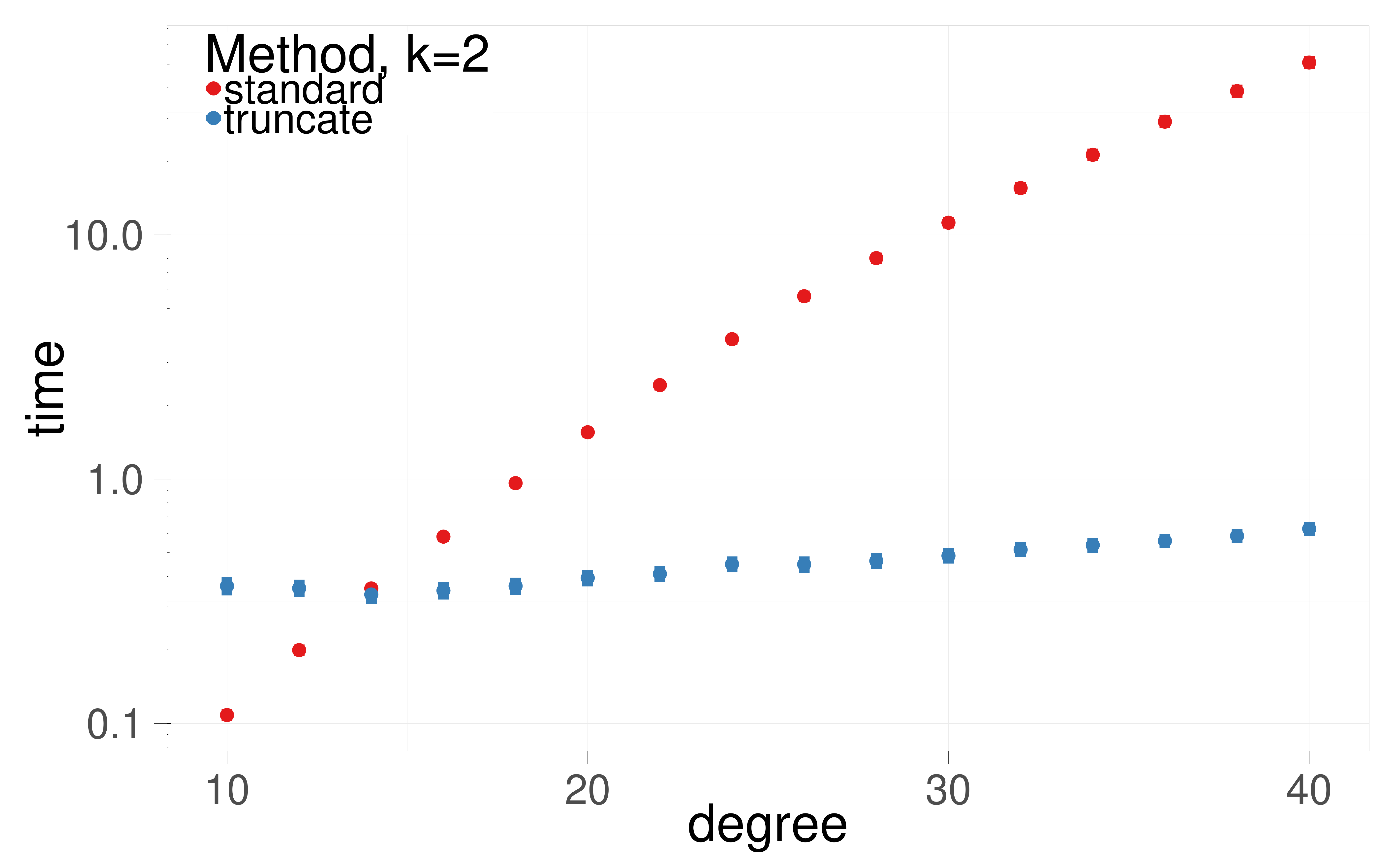}
		}
		\subfloat{
			\includegraphics[width=0.48\textwidth]{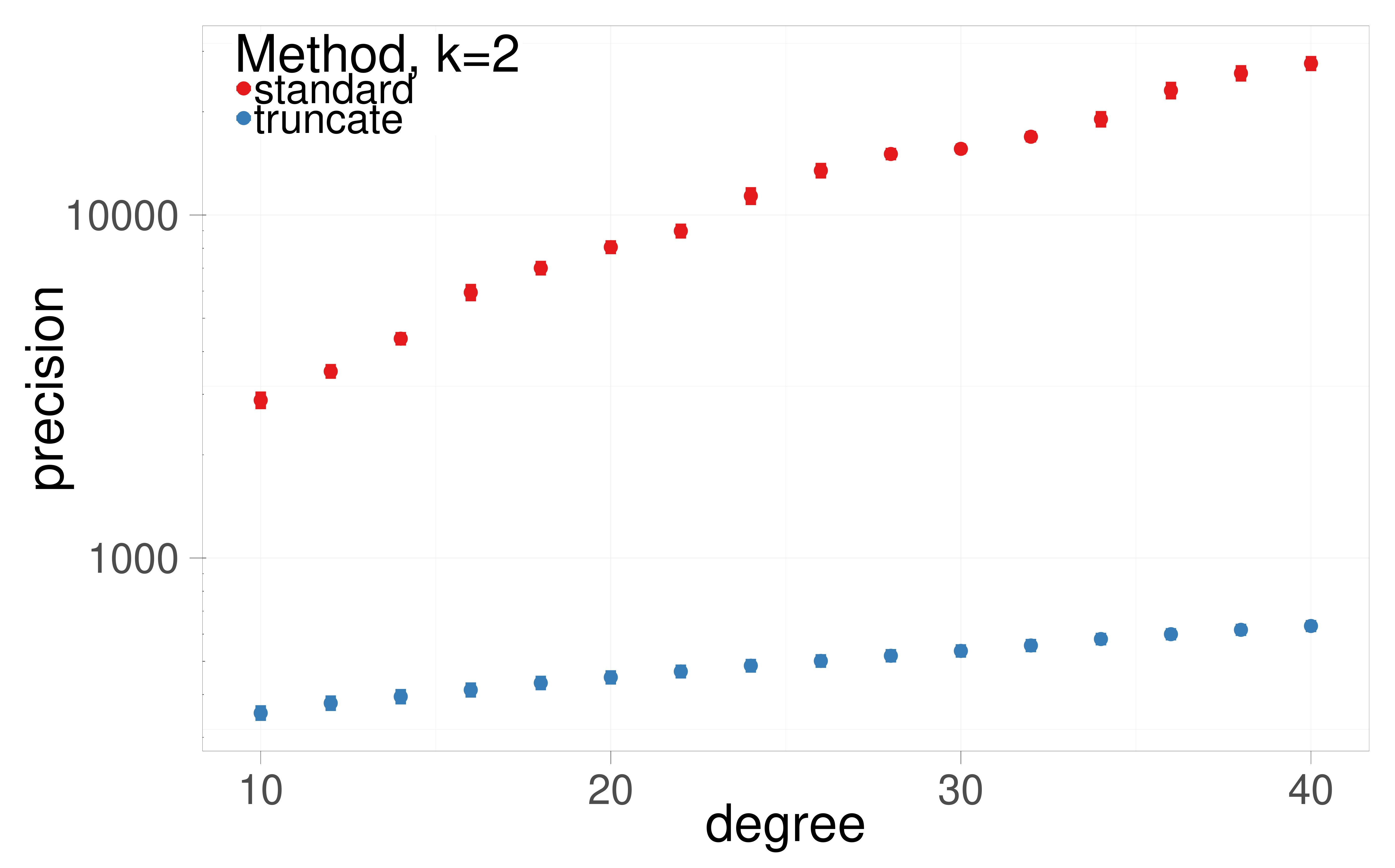}
		}
	\end{center}
	\vspace{-13mm}
	\begin{center}
		\subfloat{
			\includegraphics[width=0.48\textwidth]{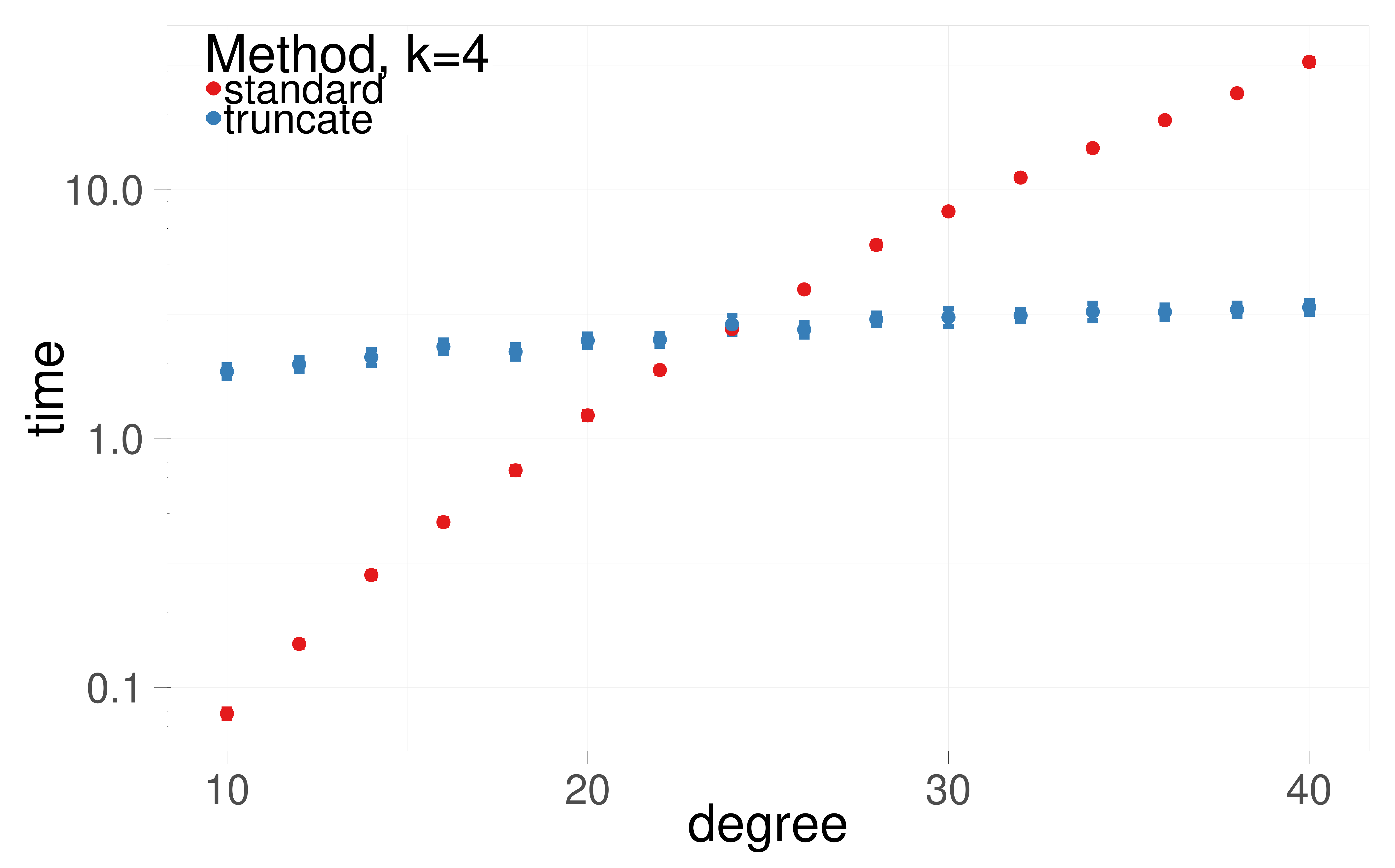}
		}
		\subfloat{
			\includegraphics[width=0.48\textwidth]{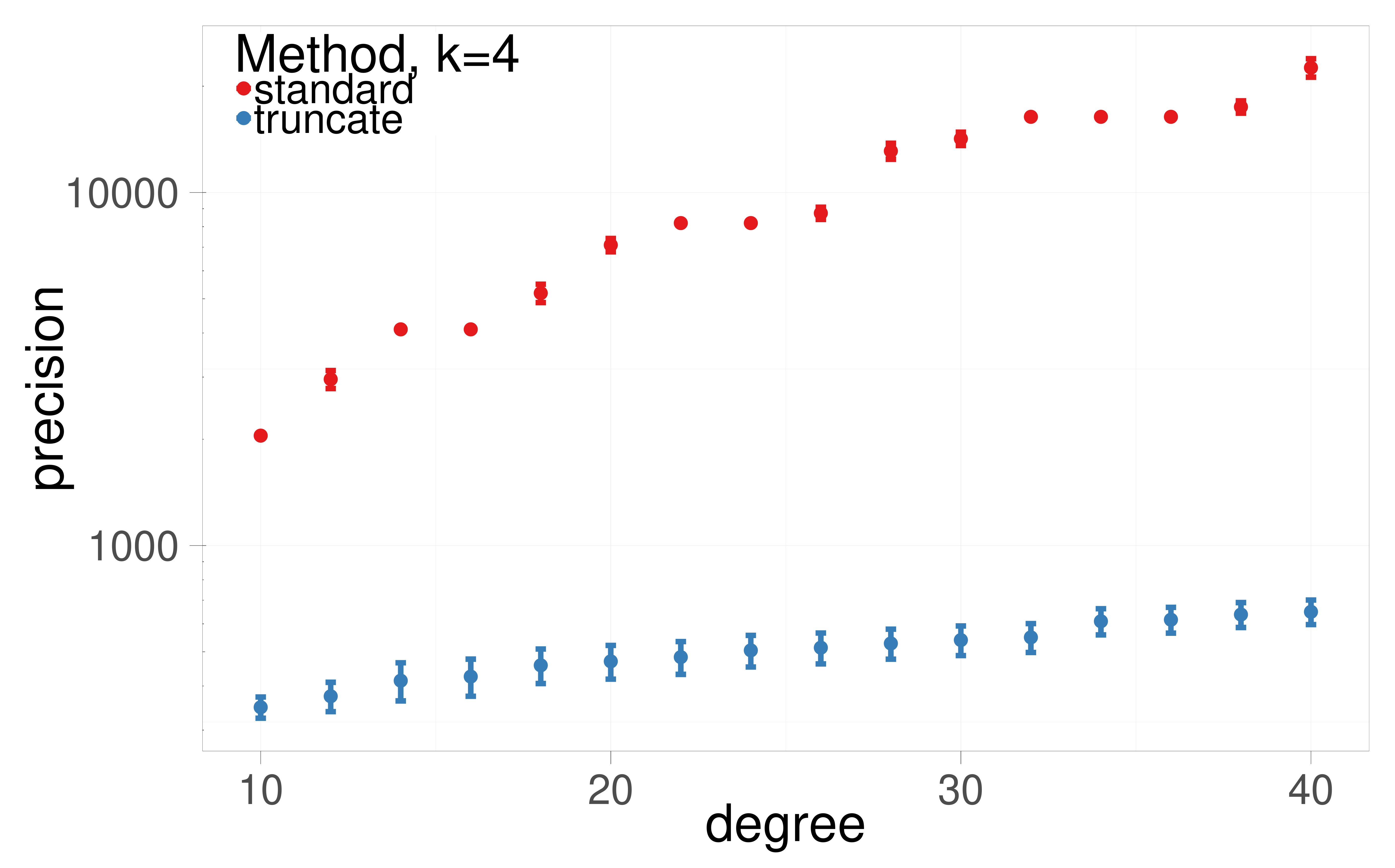}
		}
	\end{center}
	\vspace{-13mm}
	\begin{center}
		\subfloat{
			\includegraphics[width=0.48\textwidth]{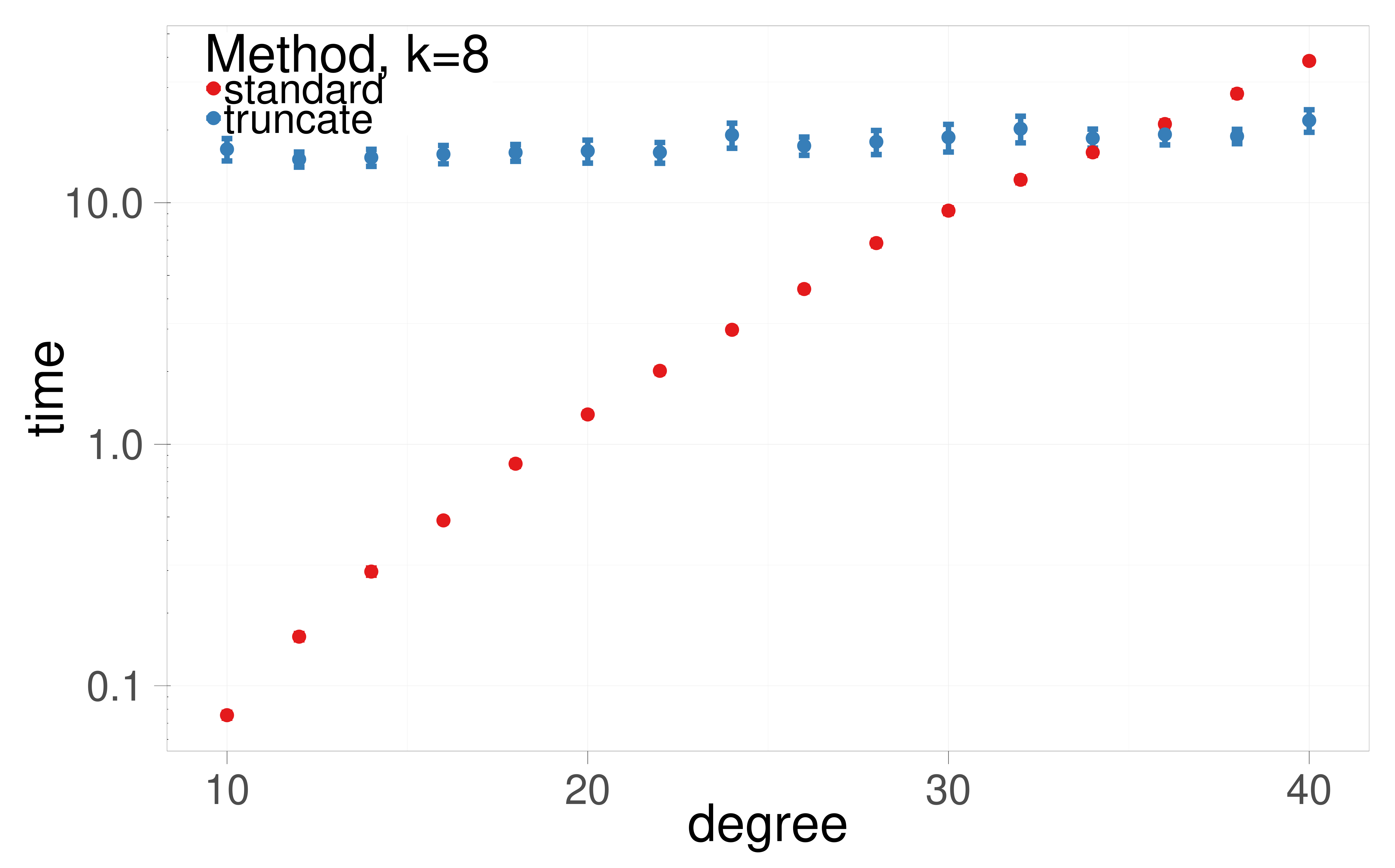}
		}
		\subfloat{
			\includegraphics[width=0.48\textwidth]{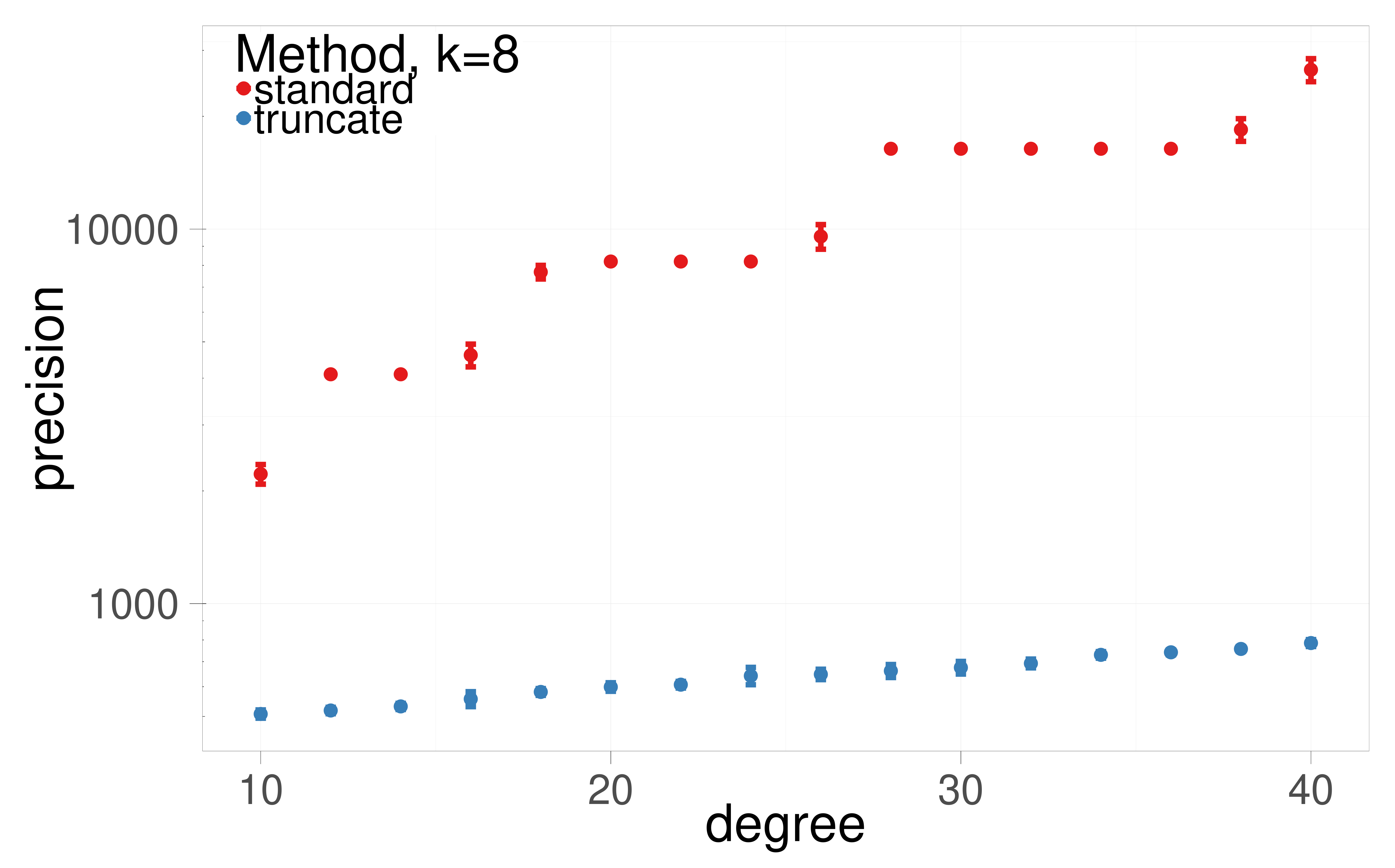}
		}
	\end{center}
	\caption{
		Evaluation for validation of $k$-fold roots, for $k=1,2,4,8$. In all plots the degree is on the horizontal axis. On the left the validation time is on the vertical axis and on the right the precision demand is on the vertical axis. The red dots correspond to the validation method of the original \textsc{Bisolve} routine, called \texttt{standard}. The blue dots correspond to the validation method that uses the new inclusion predicate, called \texttt{truncate}. 
		}
	\label{fig:plots}
\end{figure}

\begin{figure}
    \begin{center}
        \subfloat{
            \includegraphics[width=0.48\textwidth]{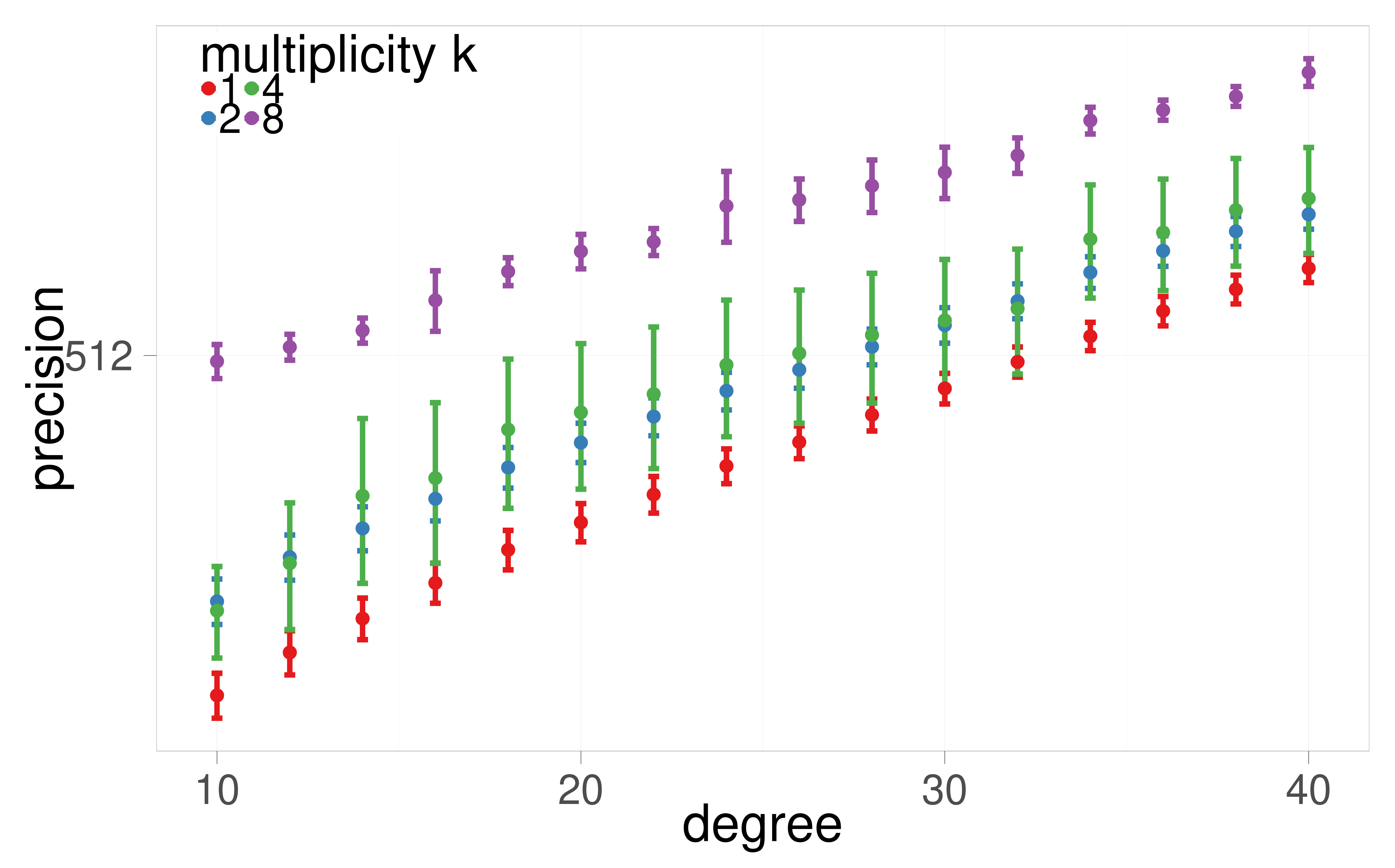}
        }
        \subfloat{
            \includegraphics[width=0.48\textwidth]{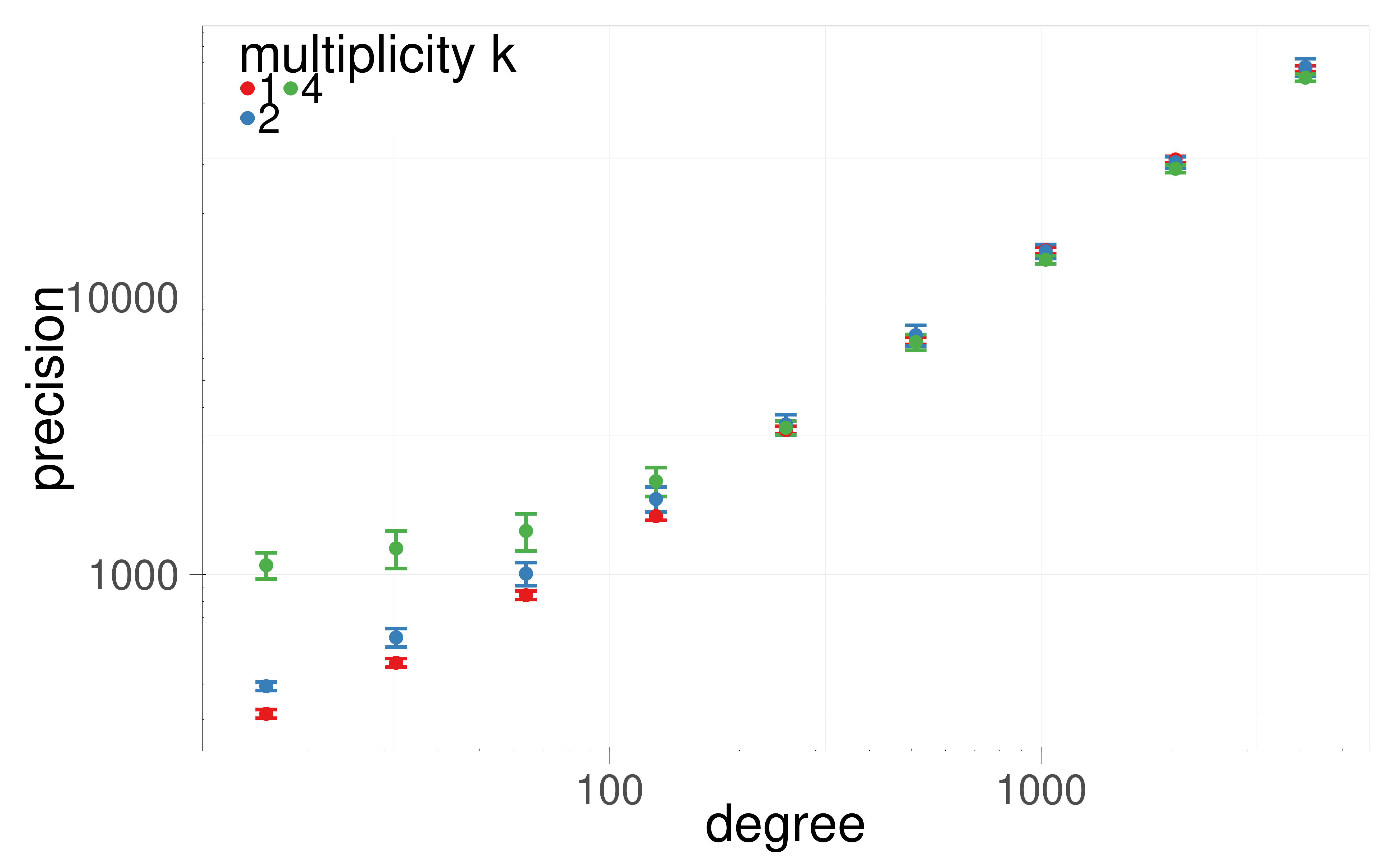}
        }
    \end{center}
    \caption{Comparison of the dependence of the precision demand on the degree for different values of $k$. For \texttt{herwig\_hauser}-instances on the left, and for \texttt{random}-instances on the right.}
    \label{fig:k_plots}
\end{figure}

\subsection{Instance Generation} 
The instances on which we compared the implementations are generated as follows. Given a trivariate polynomial $P\in\ZZ[x,y,z]$. There are several different ways of obtaining two bivariate polynomials $f,g$ from $P$ that have solutions of higher multiplicity. The different ways are encoded by the strings \texttt{0xx}, \texttt{0xy}, \texttt{0yy}, \texttt{x0y}, \texttt{y0x} in the file names. The following table summarizes the meaning of these abbreviations. We denote $p_v = \partial_v p$ for any polynomial $p\in R[v]$ for some ring $R$.
\begin{align*}
	\texttt{0xx} & \hspace{1cm} f = \res(P, P_z, z)  & \hspace{1cm} g = f_x \cdot f_x \\ 
	\texttt{0xy} & \hspace{1cm} f = \res(P, P_z, z)  & \hspace{1cm} g = f_x \cdot f_y \\  
	\texttt{0yy} & \hspace{1cm} f = \res(P, P_z, z)  & \hspace{1cm} g = f_y \cdot f_y \\ 
	\texttt{x0y} & \hspace{1cm} f = \res(P, P_z, z) \cdot f_x  & \hspace{1cm} g = f_y \\ 
	\texttt{y0x} & \hspace{1cm} f = \res(P, P_z, z) \cdot f_y  & \hspace{1cm} g = f_x 
\end{align*}
From the resulting system $f,g$, we construct the sheared system $f, g \leftarrow f(ax+by, cx+dy), g(ax+by, cx+dy)$ with integers $a,b,c,d$ drawn uniformly at random from $[-2,2]$. This is done in order to make degenerate situations where multiple solutions share the same $x$ or $y$-value less likely. We create an even larger set of instances by renaming the variables of $P$ from $x, y, z$ to $x, z, y$ or $y, z, x$ (or equivalently considering $P_x$ and $P_y$ instead of $P_z$). We abbreviate this choice with \texttt{xyz}, \texttt{xzy}, and \texttt{yzx}.
Now, let $z$ be a solution of such a system $f, g$ of multiplicity $k$. We pick random polynomials $p, q$ of increasing degrees and consider the systems $f\cdot p, g\cdot q$. This results in systems $f_d, g_d$ of increasing degrees $d$ that have the same solution $z$ of multiplicity $k$. For each degree $d$, we create three such system $f_d, g_d$ by multiplying $f, g$ with different random polynomials. 

There are two different classes of instances that we consider depending on how the initial trivariate polynomial $P$ is chosen. In the first class, called \texttt{herwig\_hauser}, we pick the polynomial $P$ from the set of polynomials given as three dimensional surfaces in the Herwig Hauser Classics gallery~\cite{herwig-hauser}. In the second class, called \texttt{random}, we pick $P$ randomly. In the first class called \texttt{herwig\_hauser} we let $d=10,12,\ldots, 40$, whereas in the second class \texttt{random}, we let $d=16, 32, \ldots, 4096$. We note that in the latter case we pick the random polynomial with which we multiply $f,g$ in order to get $f_d,g_d$ as sparse polynomials as otherwise evaluating $f,g$ already becomes non-trivial.

The generated instances can be found on the project page.\footnote{\url{http://resources.mpi-inf.mpg.de/systemspellet/}} A folder corresponding to a candidate contains one file called \texttt{orig.cnd}, which refers to the polynomials $f,g$. The remaining files correspond to the polynomials $f_d,g_d$ as described above. Every file contains four lines, the first two contain the system, while the third and fourth contain the boundaries $x-r, x+r$ and $y-r, y+r$ such that the solution is contained within this range.

\subsection{Experiments and Evaluation Results}
In the first experiment, we compare the running time as well as the precision demand of the two respective validation methods called \texttt{standard} for the method included in the original \textsc{Bisolve} routine and \texttt{truncate} for the method using the new inclusion predicate on the instance class \texttt{herwig\_hauser}.
In Figure~\ref{fig:plots}, we can see the evaluation for validating $k$-fold roots for $k=1,2,4,8$. The measurements are repeated three times, for each method and system. This results in 9 measurements (3 different random polynomials, 3 different runs) per degree per method. On the left, the running times are on the vertical logarithmic axis, whereas the degree of the systems is on the linear horizontal axis. 
On the right, the precision demand is on the vertical logarithmic axis, whereas the degree of the systems is on the linear horizontal axis. The error bars indicate 95\%-confidence intervals. 

We can see a clear advantage for our new method \texttt{truncate}. On average over all instances of degree 40, we obtain an improvement of a factor of $43.6$, $37.9$, $29.8$, $25.2$ for $k=1,2,4,8$ in the precision demand.
In Figure~\ref{fig:k_plots} on the left, we can see the precision demand for the \texttt{herwig\_hauser} instances for different $k=1,2,4,8$ for the \texttt{truncate} method. We can see that the precision demand increases with $k$ in a comparable amount as the theoretical worst-case bounds predict, namely, we can roughly see a quadratic dependence between the precision demand and the multiplicity $k$ in Figure~\ref{fig:k_plots} on the left.

In Figure~\ref{fig:k_plots} on the right, we can see results for the same experiment for the \texttt{random} instances. In this experiment, we only include the \texttt{truncate} method as the \texttt{original} method does not scale well enough for solving instances of that degree. Here both axis are logarithmic and the degree goes up to 4096. 
Fitting a linear model to the data points leads an estimate for the exponent of $0.99 \pm 0.05$ $0.94 \pm 0.06$, and $0.75\pm 0.13$ for $k=1,2,4$. The coefficients of determination lie above $0.94$ in all three cases that is roughly 94\% of the variance of the data can be explained by the fitted power model. Thus, we may conjecture that the precision demand depends at most a linearly on $d$. We remark that the plot suggests that the impact of the degree $d$ dominates over the impact of $k$ for very large $d$ as we cannot see a difference between the curves for different values of $k$ for large $d$. We remark that the impact of $k$ for small $d$ explains the smaller exponent in the fitted linear model for $k=4$ compared to $k=1,2$.

\medskip
\noindent \emph{The source code, the statistical data underlying the plots, the instances, and the script used for benchmarking are available for download on the project page.\footnote{\url{http://resources.mpi-inf.mpg.de/systemspellet/}}}

\printbibliography
\end{document}